\documentclass[sigconf,nonacm]{acmart}

\usepackage[nolist, printonlyused]{acronym}
\usepackage[lined,noend,ruled,commentsnumbered]{algorithm2e}
\usepackage{amsmath,mathtools}
\usepackage{booktabs}
\usepackage{enumitem}
\usepackage{graphicx}
\usepackage{multirow}
\usepackage{siunitx}
\usepackage{subcaption}
\usepackage[normalem]{ulem} %
\usepackage{hyperref}
\usepackage{cleveref} %
\usepackage{tikz}
\usepackage{colortbl}
\usepackage{array}
\usepackage{listings}

\newcolumntype{L}[1]{>{\raggedright\let\newline\\\arraybackslash\hspace{0pt}}m{#1}}
\newcolumntype{R}[1]{>{\raggedleft\let\newline\\\arraybackslash\hspace{0pt}}m{#1}}
\usetikzlibrary{arrows.meta}
\usetikzlibrary{patterns.meta}
\usetikzlibrary{shapes.geometric}
\usetikzlibrary{backgrounds, calc, hobby, positioning}
\tikzset{vertex style/.style={
    draw=black,
    thick,
    fill=white,
    text=black,
    ellipse,
    minimum width=1.2cm,
    minimum height=0.75cm,
    font=\footnotesize,
    outer sep=3pt,
  },
}
\tikzset{
  text style/.style={
    sloped, 
    text=black,
    font=\scriptsize,
    above
  }
}
\tikzstyle{bag} = [align=center]

\newcommand{\infSetPropValues}{\mathbf{C}}

\newcommand{\infSetLabels}{\mathbf{L}}
\newcommand{\infSetObjectIDs}{\mathbf{O}}
\newcommand{\infSetPropKeys}{\mathbf{P}}

\newcommand{\setNodeIDs}{N}
\newcommand{\subsetNodeIDs}{\mathsf{N}}
\newcommand{\setDirectedEdgeIDs}{E}

\newcommand{\setOfDeps}{\Sigma}
\newcommand{\setOfTransfs}{\mathcal{T}}
\newcommand{\setOfRedTransfs}{\setOfTransfs_\mathsf{red}}

\newcommand{\funcAverage}{\mathsf{avg}}
\newcommand{\funcAvg}{\funcAverage}
\newcommand{\funcDirectedEdgeSource}{\mathsf{src}}
\newcommand{\funcDirectedEdgeTarget}{\mathsf{tgt}}
\newcommand{\funcDomain}{\mathsf{dom}}
\newcommand{\funcLabelling}{\mathsf{lab}}
\newcommand{\funcMax}{\mathsf{max}}

\newcommand{\funcProperty}{\mathsf{prop}}

\newcommand{\varDescriptor}{D}
\newcommand{\varPropertyGraph}{G}
\newcommand{\varGraph}{\varPropertyGraph}
\newcommand{\varQueryPattern}{Q}
\newcommand{\varRelation}{R}

\newcommand{\varEdge}{e}

\newcommand{\varNode}{n}

\newcommand{\varPropKey}{p}

\newcommand{\varDep}{\varphi}

\newcommand{\varInfomorphism}{\psi}

\newcommand{\determ}{\Rightarrow}

\newcommand{\metric}{m}

\newcommand{\sat}{\models}
\newcommand{\tup}{\mu}

\newcommand{\relnames}{\mathcal{R}}
\newcommand{\attrnames}{\mathcal{A}}
\newcommand{\allVariables}{\mathcal{V}}
\newcommand{\pedgeRight}[1]{\xrightarrow{#1}} %
\newcommand{\pedgeRightLeft}[1]{\xleftrightarrow{#1}} %
\newcommand{\pobject}[3]{#1{:}#2{:}#3} %
\newcommand{\attr}{\mathsf{attr}}
\newcommand{\sem}[2]{\llbracket #1 \rrbracket^{#2}}

\newcommand{\deps}{\mathsf{deps}}

\newcommand{\newnode}[1]{\langle #1 \rangle}
\newcommand{\newlbl}[1]{\underline{#1}}

\newcommand{\moregeneralthan}[2]{#2\sqsubseteq #1} %

\newif\ifarxiv
\arxivtrue

\newcommand{\edbtarxiv}[2]{%
#2%
}

\edbtarxiv{\setlength{\textfloatsep}{0.3em}}{
\setlength{\textfloatsep}{0.3em}
}

\newif\ifshowchanges
\showchangesfalse %

\newcommand{\changed}[1]{%
#1%
}

\title{Graph-Native Normalization}
\author{Johannes Schrott}
\orcid{0000-0003-2689-0876}
\affiliation{%
  \institution{TU Wien}
  \city{Vienna}
  \country{Austria}
}
\email{johannes.schrott@tuwien.ac.at}

\author{Maxime Jakubowski}
\orcid{0000-0002-7420-1337}
\affiliation{%
  \institution{TU Wien}
  \city{Vienna}
  \country{Austria}
}
\email{maxime.jakubowski@tuwien.ac.at}

\author{Katja Hose}
\orcid{0000-0001-7025-8099}
\affiliation{%
  \institution{TU Wien}
  \city{Vienna}
  \country{Austria}
}
\email{katja.hose@tuwien.ac.at}
\date{February 2026}

\begin{document}
\begin{abstract}
    In recent years, \acp*{kg} -- in particular in the form of \acp*{lpg} -- have become essential components in a broad range of applications. Although the absence of strict schemas for KGs facilitates structural issues that lead to redundancies and subsequently to inconsistencies and anomalies, the problem of KG quality has so far received only little attention. 
Inspired by normalization using functional dependencies for relational data, a first approach exploiting dependencies within nodes has been proposed. 
However, real-world KGs also expose functional dependencies involving edges. 
In this paper, we therefore propose graph-native normalization, which considers dependencies within nodes, edges, and their combination. We define a range of graph-native normal forms and graph object functional dependencies and propose algorithms for transforming graphs accordingly. 
We evaluate our contributions using a broad range of synthetic and native graph datasets.
\end{abstract}


\maketitle
\sloppy

\section{Introduction}
\label{sec:introduction}
\Acfp{lpg} have emerged as a prominent form of \acp{kg}, 
seeing extensive use in a %
variety of contexts,
ranging from data integration~\cite{leser_informationsintegration_2007} to fraud detection~\cite{lin_fraudgt_2024}, 
\changed{and are a central abstraction in modern graph processing systems~\cite{SakrBVIAAAABBDV21}.} 
Especially with their recent popularity as a source of reliable knowledge in \ac{ml} and \ac{ai}, e.g.,
to mitigate hallucinations 
\changed{in \acp{llm}~\cite{LavrinovicsBBH25},
ensuring the quality of \acp{kg} becomes more and more essential~\cite{hogan_knowledge_2022}.}
As a first step towards achieving this goal,
new standards and schema-level approaches have been proposed,
e.g., the \acf{gql}~\cite{iso_gql_2024,francis_researchers_2023}, 
\textsc{PG-Schema}~\cite{angles_pg-schema_2023},
\changed{\textsc{PG-Keys}~\cite{angles_pg-keys_2021}, and transformations or alignments between them~\cite{RabbaniLBH24,AhmetajBHHJGMMM25}.} 
However, being able to define schemas is only the first step.
We still need to fully understand what a good schema is and how we can improve the schema of a \ac{kg}.
In relational databases, normalization~\cite{codd_relational_1970,codd_further_1971,batini_data_2016} is a well-known technique to transform a database schema into a schema that is commonly perceived as better with respect to a number of desirable properties. 
Such transformations are traditionally guided by exploiting \acp{fd} and are described by normal forms, 
which generally increase data quality by reducing redundancy; 
the latter not only unnecessarily inflates storage space but also serves as a source of error since all redundant information needs to be kept in sync to ensure consistency. 

Originally introduced for the relational data model in the 1970s~\cite{codd_relational_1970,codd_further_1971},
normalization has recently been adopted to the \ac{lpg} setting by \citeauthor{skavantzos_third_2025}~\cite{skavantzos_third_2025,skavantzos_normalizing_2023}. 
In these works, 
nodes are considered similar to relations that are being normalized. 
However, 
graphs are subject to additional types of redundancies.
Consider,
for example, 
the graph depicted in \Cref{fig:lecture-uses-book-denorm}, 
which captures information about courses and lecturers. 
In this graph,
we see that it might be worthwhile to enforce that \emph{each course is taught using exactly one book}. 
However, 
this is difficult to achieve in the current graph because the information is modeled as a property of edges with label \textsf{teaches} that connect teachers to courses. 
This way of modeling the information might have historical reasons 
but as the graph keeps on growing and evolving, 
 it might lead to errors. 
In \Cref{fig:lecture-uses-book-denorm},
redundancy is indicated by \dotuline{dotted lines}.
Hence, 
to reduce the degree of redundancy and therefore the potential for inconsistencies, 
the graph should be normalized so that the information about the book
\changed{is} stored together with the course instead of the \textsf{teaches} edges (see \Cref{fig:lecture-uses-book-norm}). 
In this paper,
we build upon the work proposed by \citet{skavantzos_third_2025},
which focuses on dependencies within nodes, 
and extend it to support also dependencies within edges, and dependencies between nodes and edges.
Since our approach focuses on both, nodes and edges, we call it \emph{graph-native}.
In summary, the contributions of this paper are as follows: 

\begin{enumerate}[label=(\roman*)]
\item We propose a new \changed{dependency} type, \changed{\acfp{gnfd}}, that expresses dependencies within and between nodes and edges.
\item Building upon the state of the art, we provide graph-native generalizations of \changed{existing normal forms, which we call \acfp{gnnf}.}
\item We \changed{identify redundancy patterns  for \acp{lpg} and transformations to remove them. Based thereon, we} introduce graph-native \ac{lpg} normalization algorithms that transforms a given graph into our \changed{\acp{gnnf}}. 
\item Our experimental evaluation shows how graph-native normalization guided by \acp{gnfd} \changed{reduces redundancy and affects query performance.}

\end{enumerate}
This paper is structured as follows: 
\Cref{sec:preliminaries} introduces basic concepts and a running example; \Cref{sec:related-work} reviews related work; 
\Cref{sec:dependency} defines \acp{gnfd} for \acp{lpg}; 
\Cref{sec:native-property-graph-normal-forms} analyzes redundancies and defines \changed{\acp{gnnf}}; 
\Cref{sec:normalization} presents our normalization algorithms; \Cref{sec:experiments} reports evaluation results; and \Cref{sec:conclusion} concludes with a summary and future work.

\begin{figure*}[htbp]
    \begin{subfigure}{0.39\textwidth}
        \centering
        \begin{tikzpicture}[node distance=1.5cm]

\node[vertex style=Node,align=center] (1) {\textbf{:Course}\\\texttt{n1}};
\node[text style,left of=1,xshift=-0.6em,yshift=0.6em] (1_props) {\begin{tabular}[t]{@{~}r@{~}l@{~}}
    \multirow{2}{*}{title:} & ``Database  \\
                            & \hphantom{``}Systems''     \vspace{0.2em}   \\ 
                 language:  & ``English''
\end{tabular}};

\node[vertex style=Node, below of=1,align=center,yshift=-3em] (3) {
\textbf{:Lecturer}\\\texttt{n3}}
edge [->,align=left] node[text style,rotate=-90,right]{\\\\\textbf{:teaches}\\\texttt{e2}\\at: 2026-02-16\\\dotuline{usingBook: ``Alice''}} (1);
 
\node[text style,below of=3,yshift=2.7em] (3_props) {\begin{tabular}[t]{@{~}r@{~}l@{~}} 
                 name:  & ``Maxime''
\end{tabular}};

\node[vertex style=Node, left of=3, xshift=-4.5em,align=center] (2) {
\textbf{:Lecturer}\\\texttt{n2}}
edge [->,align=left] node[text style,rotate=0,above]{\textbf{:teaches}\\\texttt{e1}\\at: 2026-02-09\\\dotuline{usingBook: ``Alice''}} (1);
\node[text style,below of=2,yshift=2.7em] (2_props) {\begin{tabular}[t]{@{~}r@{~}l@{~}} 
                 name:  & ``Johannes''
\end{tabular}};

\node[vertex style=Node, right of=3,xshift=4.5em,align=center] (4) {
\textbf{:Lecturer}\\\texttt{n4}}
edge [->,align=left] node[text style]{\textbf{:teaches}\\\texttt{e3}\\at: 2026-02-23\\\dotuline{usingBook: ``Alice''}} (1);
\node[text style,below of=4,yshift=2.7em] (4_props) {\begin{tabular}[t]{@{~}r@{~}l@{~}} 
                 name:  & ``Katja''
\end{tabular}};

\end{tikzpicture}
        \vspace*{-1em}
        \caption{Courses taught by lecturers using books}
        \label{fig:lecture-uses-book-denorm}
    \end{subfigure}
    \begin{subfigure}{0.57\textwidth}
        \centering
        \begin{tikzpicture}[node distance=1.5cm]

\node[vertex style=Node,align=center] (5) {\textbf{:Course}\\\textbf{:Annual}\\\texttt{n5}};
\node[text style,right of=5,xshift=1.75em] (5_props) {\begin{tabular}[t]{@{~}r@{~}l@{~}}
    \multirow{2}{*}{title:} & ``Management of\\
                            & \phantom{``}Graph~Data''     \vspace{0.2em}   \\ 
                 language:  & ``English''   \vspace{0.2em}\\
                 year: & \hphantom{``}2024  \vspace{0.2em}\\
                 semester: & ``Winter''
\end{tabular}};

\node[vertex style=Node,align=center,right of=5,xshift=9em,yshift=-1em] (6) {\textbf{:Course}\\\textbf{:Annual}\\\texttt{n6}};
\node[text style,right of=6,xshift=1.7em] (6_props) {\begin{tabular}[t]{@{~}r@{~}l@{~}}
    \multirow{2}{*}{title:} & ``Management of\\
                            & \phantom{``}Graph~Data''     \vspace{0.2em}   \\ 
                 language:  & ``English''   \vspace{0.2em}\\
                 year: & \hphantom{``}2026  \vspace{0.2em}\\
                 semester: & ``Summer''\\\\
                 ~
\end{tabular}};

\node[vertex style=Node, below of=6,align=center,yshift=-1.2em,xshift=-9.5em] (7) {
\textbf{:Student}\\\texttt{n7}}
edge [->,align=left] node[text style,rotate=0,above,yshift=-1em]{\textbf{:takes}\\\texttt{e6}} (6);
\node[text style,below of=7,yshift=2.7em] (7_props) {\begin{tabular}[t]{@{~}r@{~}l@{~}} 
                 name:  & ``Laura''
\end{tabular}};

\node[vertex style=Node, below of=6,align=center,yshift=-1.2em] (8) {
\textbf{:Student}\\\texttt{n8}}
edge [->,align=left] node[text style,rotate=0,above,yshift=-0.9em]{\textbf{:takes}\\\texttt{e7}} (6);
\node[text style,below of=8,yshift=2.7em] (8_props) {\begin{tabular}[t]{@{~}r@{~}l@{~}} 
                 name:  & ``Tom''
\end{tabular}};

\node[vertex style=Node, below of=6,align=center,yshift=-1.2em,xshift=9.5em] (9) {
\textbf{:Student}\\\texttt{n9}}
edge [->,align=left] node[text style,rotate=0,above,yshift=-1em]{\textbf{:takes}\\\texttt{e8}} (6);
\node[text style,below of=9,yshift=2.7em] (9_props) {\begin{tabular}[t]{@{~}r@{~}l@{~}} 
                 name:  & ``Ann''
\end{tabular}};

\draw [->, align=center] (7) to   node[text style,rotate=0,below,align=left] {\textbf{:inGroupWith}\\\texttt{e4}\\groupNo: \phantom{``}1\\\dotuline{name: ``Heroes''}} (8);
\draw [->, align=center] (9) to   node[text style,rotate=0,below,align=left] {\textbf{:inGroupWith}\\\texttt{e5}\\groupNo: \phantom{``}1\\\dotuline{name: ``Heroes''}} (8);

\end{tikzpicture}
        \vspace*{-0.5em}
        \caption{Courses taken by students that work together in groups}
        \label{fig:student-in-group}
    \end{subfigure}
    \vspace*{-1em}
    \caption{Examples of denormalized \acp{lpg} describing courses, lecturers, and students. Redundancies are \dotuline{underlined with dots}.}
    \label{fig:example}
    \vspace*{-0.7em}
\end{figure*}

\begin{figure}[htbp]
    \begin{tikzpicture}[node distance=1.5cm]

\node[vertex style=Node,align=center] (1) {\textbf{:Course}\\\texttt{n1}};
\node[text style,left of=1,xshift=-0.6em,yshift=0.6em] (1_props) {\begin{tabular}[t]{@{}r@{}l@{}}
    \multirow{2}{*}{title:~} & ``Database  \\
                            & \hphantom{``}Systems''     \vspace{0.2em}   \\ 
                 {language:~}  & ``English'' \\
                 \dotuline{usingBook:~} & \dotuline{``Alice''}
\end{tabular}};

\node[vertex style=Node, below of=1,align=center,yshift=-2.2em] (3) {
\textbf{:Lecturer}\\\texttt{n3}}
edge [->,align=left] node[text style,rotate=-90,right]{\\\\\textbf{:teaches}\\\texttt{e2}\\at: 2026-02-16} (1);
 
\node[text style,below of=3,yshift=2.7em] (3_props) {\begin{tabular}[t]{@{~}r@{~}l@{~}} 
                 name:  & ``Maxime''
\end{tabular}};

\node[vertex style=Node, left of=3, xshift=-4.5em,align=center] (2) {
\textbf{:Lecturer}\\\texttt{n2}}
edge [->,align=left] node[text style,rotate=0,above]{\textbf{:teaches}\\\texttt{e1}\\at: 2026-02-09} (1);
\node[text style,below of=2,yshift=2.7em] (2_props) {\begin{tabular}[t]{@{~}r@{~}l@{~}} 
                 name:  & ``Johannes''
\end{tabular}};

\node[vertex style=Node, right of=3,xshift=4.5em,align=center] (4) {
\textbf{:Lecturer}\\\texttt{n4}}
edge [->,align=left] node[text style]{\textbf{:teaches}\\\texttt{e3}\\at: 2026-02-23} (1);

\node[text style,below of=4,yshift=2.7em] (4_props) {\begin{tabular}[t]{@{~}r@{~}l@{~}} 
                 name:  & ``Katja''
\end{tabular}};

\end{tikzpicture}
        \vspace*{-1em}
    \caption{Normalized version of the \ac{lpg} from \Cref{fig:lecture-uses-book-denorm}}
    \label{fig:lecture-uses-book-norm}
\end{figure}

\section{Preliminaries and Running Example}
\label{sec:preliminaries}
Since this paper relies on notions from the relational data model, 
including \acfp{fd}, 
and applies them to \acfp{lpg}\changed{,
we} provide preliminary notions for these based on the literature~\cite{angles_pg-schema_2023,francis_researchers_2023,bonifati_querying_2018,alice}.

\paragraph{Labeled Property Graphs}  
We assume pairwise disjoint countable sets
of labels $\infSetLabels$, property keys $\infSetPropKeys$, constants
$\infSetPropValues$, and (graph) object identifiers $\infSetObjectIDs$.
\changed{We often} refer to object identifiers as \emph{(graph) objects}: they represent either \emph{nodes} or \emph{edges}.
A \acl{lpg} $\varPropertyGraph$ is a tuple
$(\setNodeIDs, \setDirectedEdgeIDs, \funcLabelling,
\funcDirectedEdgeSource, \funcDirectedEdgeTarget, \funcProperty)$ with:

\begin{itemize}
\item mutually disjoint finite
  subsets  $\setNodeIDs, \setDirectedEdgeIDs$  of $\infSetObjectIDs$ that represent the nodes and edges respectively; 
\item a total function $\funcLabelling: \setNodeIDs\cup\setDirectedEdgeIDs\to 2^\infSetLabels$ that maps objects to (possibly empty) sets of labels;
\item total functions $\funcDirectedEdgeSource: \setDirectedEdgeIDs\to\setNodeIDs$ and
  $\funcDirectedEdgeTarget: \setDirectedEdgeIDs\to\setNodeIDs$  that map each edge to nodes representing the edge's source and target
  respectively; and
\item a partial function
  $\funcProperty:
  (\setNodeIDs\cup\setDirectedEdgeIDs)\times\infSetPropKeys\to\infSetPropValues$
   that assigns constants to the property keys in objects.
\end{itemize}
This definition of \acp{lpg} is compatible with \changed{\ac{gql}~\cite{iso_gql_2024} and common} graph database systems, such as 
Neo4j~\cite{neo4j-datamodel} or Memgraph~\cite{memgraph-datamodel}. 
\changed{The data model of these two systems is a subset of our definition, where, contrary to $\ac{gql}$,} edges can only have exactly one label, which is called type.

\paragraph{Relation Schemas, Tuples, and Relations} 
We assume mutually disjoint countable sets of \emph{relation names} $\relnames$ and \emph{attributes}
$\attrnames$, disjoint with any of the previous sets. 
When we assume a finite set of attributes $A\subset\attrnames$ 
with a relation name $\varRelation\in\relnames$,
denoted by $\attr(\varRelation) = A$,
we call $\varRelation$ a relation schema with attributes $A$.
A \emph{tuple} \changed{$\tup$} for a relation schema $\varRelation$ is a function that assigns constants from $\infSetPropValues$ to $\attr(\varRelation)$. 
\changed{Formally, it is a total function
$\tup: \attr(\varRelation)\mapsto \infSetPropValues$. 
An \emph{instance} $I$ of a relation schema $\varRelation$
is a set of tuples for~$\varRelation$. 
Sometimes, $I$ is simply called a \emph{relation}. 
We say that two tuples $\tup_1,\tup_2$ \emph{agree} on attributes $A$ if for every $a\in A$, $\tup_1(a)=\tup_2(a)$. 
}

\paragraph{Functional Dependencies and Keys} 
Essential to all techniques related to normalization are so-called dependencies. 
A dependency is a type of constraint that holds on multiple data records -- it can cover multiple data attributes \cite{andreasen_data_2021}.
The dependency type relevant for normalization in the relational data model are \textit{\acfp{fd}}.
\Acp{fd} express ``which attributes are functional dependent on others''~\cite{codd_further_1971}.
In other words, 
this means that the values of a certain set of attributes $X$ define the values of another set of attributes $Y$.
Formally, an \ac{fd} over
a set of attributes $A$ is an expression $X\determ Y$,
with $X,Y\subseteq A$. 
An \ac{fd} $X\determ Y$ is \emph{applicable} to a relation schema $R$
if $X\cup Y\subseteq\attr(\varRelation)$.
We say that an \ac{fd}
$X\determ Y$ is satisfied by a relation instance $I$ for $R$,
denoted by $I\sat X\determ Y$, 
if for every two tuples $\tup_1,\tup_2\in I$,
whenever they agree on $X$,
they agree on $Y$.
Given a set of \acp{fd} $\Sigma$ over $A$, 
and another dependency $X\determ Y$ we say that $X\determ Y$ holds in $\Sigma$,
denoted by $\Sigma\models X\determ Y$,
if for every instance $I$ of a relation schema $R$ with $\attr(R) = A$ that satisfies all dependencies in $\Sigma$,
we have $I\models X\determ Y$.
We assume for every relation schema $R$ a (possibly empty) set of \acp{fd} $\deps(R)$.
A \emph{superkey} of a relation schema $R$ is a set $X\subseteq \attr(R)$ such that $\deps(R)\models X\determ \attr(R)$. 
A \emph{key} for $R$ is a (subset-)minimal superkey.

\paragraph{Graph Patterns}
In the \ac{lpg} literature, 
the notion of graph patterns is often used to express queries. 
We adopt this notion for our purposes and consider a semantics that corresponds to the \emph{\acf{gpc}}~\cite{francis_gpc_2023}.
We do not need the full power of \ac{gpc}, specifically, we do not consider repetitions. 
Throughout this paper, \changed{we use} two notations for patterns\changed{:
(i)~a graphical notation and (ii)~a concrete syntax.}

\begin{figure}[t]
    \centering
    \begin{tikzpicture}[node distance=1em]
        \node[vertex style=Node,align=center,minimum height=0.4cm,minimum width=0.6cm] (1) {$x_1{:}L$};
        \node[text style,below of=1,align=center,yshift=-0.4em] (2_props1) {$k_1$};
        \node[text style,below of=2_props1,align=center] (2_props2) {$k_2$};
        
        \node[vertex style=Node,align=center,minimum height=0.4cm,minimum width=0.6cm,right of=1,xshift=7em] (2) {$x_2$} edge [<-,align=left] node[text style]{$y{:}L'$} (1);
        \node[text style,right of=1,align=center,yshift=-0.7em,xshift=3em] (2_props1) {$k_3$};
    \end{tikzpicture}   
    \vspace*{-1em}
\caption{\changed{Example of the graphical graph pattern notation} }
    \label{fig:graphical-notation}
\end{figure}

\changed{
The graphical notation consists of nodes and edges.
Nodes are circles, edges are arrows.
Both nodes and edges have a variable (typically $x$ or $y$), an associated set of labels (typically denoted by $L$), and a set of property keys (typically $P=\{k_1, k_2, \dots\}$). 
Edges are directed. 
Throughout the paper, we explicitly state when edges in our notation represent both directions.} 
\changed{\Cref{fig:graphical-notation} illustrates the graphical notation with an example. Here, $x_1$ matches nodes that have (at least) labels $L$ and property keys $k_1$ and $k_2$.
Node $x_1$ has an outgoing edge $y$ with (at least) labels $L'$ and property key $k_3$ to another node $x_2$ without any requirements on its labels.}
This graphical notation corresponds to the following \ac{gpc} expression:%
\footnote{We use the filter condition \changed{$x_1.k_1=x_1.k_1$} to ensure property $k_1$ is present.}
\begin{minipage}{\columnwidth}
{\small
\[
(x_1{:}L)_{x_1.k_1=x_1.k_1\land x_1.k_2=x_1.k_2}\pedgeRight{y:L'}_{y.k_3=y.k_3}(x_2)
\]
}
\end{minipage}
Additionally, we introduce a concrete syntax for a subset of \ac{gpc},
which we call \emph{basic graph patterns} 
\changed{and often simply refer to as \emph{graph patterns}.}
We assume a countable set of \emph{object identity variables} $\allVariables$, disjoint from
$\infSetLabels\cup\infSetPropKeys\cup\infSetPropValues$.
Let
$x\in\allVariables$. 
A \emph{property variable} is then an
expression of the form $x.k$ with $k\in\infSetPropKeys$.
A \emph{variable} is either a property- or an object identity variable. 
Based thereon, we define the syntax \changed{for basic graph patterns} as follows:\\
\begin{minipage}{\columnwidth}
{\small
\[
\varQueryPattern ::={} (\pobject xLP) \mid ()\pedgeRightLeft{\pobject yLP}() \mid (\pobject xLP)\pedgeRightLeft{\pobject y{L'}{P'}}()
\]
}
\end{minipage}\vspace{0.3em}
where $\pedgeRightLeft{}$ is either $\xrightarrow{}$ or $\xleftarrow{}$,  
$x,y\in\allVariables$ \changed{are} object identity variables,
$L,L'\subseteq \infSetLabels$ \changed{is} a set of labels,
and $P,P'\in\infSetPropKeys$ \changed{is} a set of keys.
Given a graph pattern $\varQueryPattern$, 
we define its \emph{output attributes} $\attr(\varQueryPattern)$ as the set of variables that occur in the pattern. 
We have:\vspace*{-0.3em}\\
\begin{minipage}{\columnwidth}
\begin{align*}
\small
    \attr((\pobject xLP)) &:= \{x.k\mid k\in P\}\cup \{x\}\\[-0.5em]
    \attr(()\pedgeRightLeft{\pobject yLP}()) &:= \{y.k\mid k\in P\}\cup \{y\}\\[-0.5em]
    \attr((\pobject xLP)\pedgeRightLeft{\pobject y{L'}{P'}}()) &:= \{x.k\mid k\in P\}\cup \{y.k\mid k\in P'\}\cup\{x,y\}
\end{align*}
\end{minipage}\vspace{0.3em}
For an \ac{lpg} $\varGraph$ and a graph pattern $\varQueryPattern$,
the evaluation of $\varQueryPattern$ in $\varGraph$ is the relation
$\sem\varQueryPattern\varGraph$ for \changed{a} relation schema
$\varRelation_\varQueryPattern$ with $\attr(R_\varQueryPattern)=\attr(\varQueryPattern)$, as defined in
\Cref{tab:evalpattern}.

\begin{table}[tbp]
    \footnotesize
    \centering
    \caption{The evaluation $\sem\varQueryPattern\varGraph$ of pattern $\varQueryPattern$ in graph $\varGraph$}
    \label{tab:evalpattern}
    \vspace*{-1em}
    \begin{tabular}{ll}
\toprule
$\varQueryPattern$  & $\sem\varQueryPattern\varGraph$ \\
\midrule
$(\pobject xLP)$ & $\left\{ \left\{
\begin{aligned}
& x\mapsto n; \\
& x.k\mapsto k_v \\
& \text{for each $k\in P$} \\
\end{aligned}\right\}
\,\middle|\,
\begin{aligned}
& n\in \subsetNodeIDs, \funcLabelling(n) \subseteq L,\;\text{and}\\
& \funcProperty(n,k)=k_v\;\text{for each $k\in P$}
\end{aligned}\right\}$\\

$()\pedgeRightLeft{\pobject x{L}{P}}()$ & $\left\{\left\{
\begin{aligned}
& x\mapsto e; \\
& x.k\mapsto k_v \\
& \text{for each $k\in P$} \\
\end{aligned} \right\}
\,\middle|\,
\begin{aligned}
& e\in \setDirectedEdgeIDs,\funcLabelling(e) \subseteq L,\;\text{and}\\
& \funcProperty(e,k)=k_v\\
& \text{for each $k\in P$}
\end{aligned}\right\}$\\

$(\pobject xLP)\pedgeRightLeft{\pobject y{L'}{P'}}()$ & $\left\{ \tup\cup\left\{
\begin{aligned}
& y\mapsto e; \\
& y.k\mapsto k_v \\
& \text{for each $k\in P'$} \\
\end{aligned} \right\}
\,\middle|\,
\begin{aligned}
& \tup\in\sem{(\pobject xLP)}\varGraph\\
& e\in \setDirectedEdgeIDs, \funcDirectedEdgeSource(e)=\tup(x),\\
& \funcLabelling(e) \subseteq L',\;\text{and}\\
& \funcProperty(e,k)=k_v\\
& \text{for each $k\in P'$}
\end{aligned}\right\}$\\
\bottomrule
\end{tabular}
\end{table}

\paragraph{Graph Transformations} 
Finally, we adapt the notion of graph transformations from \citet{bonifati_versatile_2025} for our purposes.
A transformation is a pair $(Q_\mathit{before},Q_\mathit{after})$ where $Q_\mathit{before}$ is a graph pattern, 
and $Q_\mathit{after}$ is a graph pattern that contains constructors. 
There are two kinds of constructors: node identity- and label constructors.
Essentially, 
they are functions mapping sets of values to newly created node identities and labels respectively.%
\footnote{These constructors are better known as \emph{Skolem functions}~\cite{bonifati_versatile_2025}.} 
We will use the notation $\newnode{x}$ and $\newlbl{x}$ to represent the creation of a new node identity and a new label based on the value $x$.
Consider the first row in \Cref{tab:transformation-patterns} 
(with \textbf{ID}=$\varInfomorphism_{\text{within-n}}$, disregarding the gray arrows for now). 
The second column \changed{({$T$})} represents $Q_\mathit{before}$ and is a single node pattern.
It matches nodes $x$ with labels $L$ where the property keys $k_1$ and $k_2$ are present.
Then, the pattern \changed{in the third column ($T_{new}$)} is $Q_\mathit{after}$,
which constructs a new node identity based on $L$ and the property key $k_1$, 
and also constructs a new label based on the same values. 
Furthermore, the property keys (and associated values) $k_1$ and $k_2$ are moved from node $x$ to node $\newnode{Lk_1}$. 
The semantics of a single transformation is defined by a \emph{creation phase} in which new graph objects and values are created, 
followed by a \emph{removal phase} in which graph objects and values are removed.
For a set of graph transformations $\mathcal{T}$,
first, for each transformation $T\in\mathcal{T}$, 
the creation phase is executed, 
and then, for each $T\in\mathcal{T}$, the removal phase is executed.

\subsection*{Running Example}
To illustrate our concepts and contributions,
we introduce a running example (cf. \Cref{fig:example}). 
It captures an academic context, where courses are taught by lecturers using certain books. 
Some courses are taught once, others annually. 
The example also captures students who are taking courses and may form groups to solve course assignments together. 

In this example, we can formulate several dependencies:
\begin{itemize}[leftmargin=5em, listparindent=1.5em]
    \item[\changed{$\varDep_\text{Course}$}:\label{dep:course}] ``Courses are identified by their title and year.'' 
    \item[\changed{$\varDep_\text{Book}$}:\label{dep:book}] ``All lecturers teaching a particular course must use the same textbook.'' 
    \item[\changed{$\varDep_\text{GroupNo}$}:\label{dep:groupno}] ``The assigned number of a student group determines its name.'' 
    \item[\changed{$\varDep_\text{Annual}$}:\label{dep:annual}] ``Annual courses may only be taught in one of the two semesters of a year.'' 
\end{itemize}
We can classify these dependencies on the graph objects they use. 
Dependency~\hyperref[dep:course]{\changed{$\varDep_\text{Course}$}} is defined on nodes (Courses) and their property keys (\textsf{title} and \textsf{year}). 
As this dependency only involves node information, we call it a \emph{within-node} dependency.
Dependency~\hyperref[dep:book]{\changed{$\varDep_\text{Book}$}}
involves nodes (Lecturers) but also property keys from related edges (\textsf{usingBook}). 
\changed{We call $\varDep_\text{Book}$ a \emph{between-graph-object} dependency, as it involves both, nodes and edges.}
The last type of dependency involves only edges and their property keys. 
Dependency~\hyperref[dep:groupno]{\changed{$\varDep_\text{GroupNo}$}}
 is such an example, and we call it a \emph{within-edge} dependency. 
Dependency~\hyperref[dep:annual]{\changed{$\varDep_\text{Annual}$}} is another example of a within-node dependency.

\section{Related Work}
\label{sec:related-work}
Before continuing, we review existing work on (non-)relational normalization and (functional) dependencies for graphs.

\subsection{Relational Normalization}

The foundations of schema \textit{normalization} are inseparably connected to the \textit{relational model} \cite{codd_relational_1970}. 
The \acf{1nf} has been proposed together with the relational model by \citet{codd_relational_1970} in \citeyear{codd_relational_1970},
followed by the \acf{2nf}~\cite{codd_further_1971}, the \acf{3nf}~\cite{codd_further_1971}, 
and the \acf{bcnf}~\cite{codd_recent_1974} shortly thereafter.
An overview of the characteristics of each relational normal form is shown in \Cref{tab:rel-norm-form-characteristics}. 

\begin{table}
    \centering
    \caption{Characteristics of relational normal forms \cite{saake_datenbanken_2018,alice}}
    \vspace*{-1em}
    \label{tab:rel-norm-form-characteristics}
    \footnotesize
    \begin{tabular}{@{\enspace}L{1.2cm}@{\enspace}L{6.9cm}@{\enspace}}
    \toprule
    \textbf{Relational normal form} & \textbf{Characteristic} \\
    \midrule
       \acs{1nf}  &  All values in the relation are atomic.\\
       \acs{2nf}  &  Non-key attributes of the relation schema are fully dependent on keys.\\ %
       \acs{3nf} & Non-key attributes are not transitively dependent on keys. \\
       \acs{bcnf} & Every \ac{fd} is trivial or contains a superkey on its left-hand side. \\
   \bottomrule
\end{tabular}
\end{table}

The standard process of transforming a relation schema that is not in some desired normal form into a set of relation schemas 
is called \emph{decomposition}. 
Decomposing a relation schema brings advantages such as smaller relations, 
independent information stored separately,
and the omission of redundancies \cite{paredaens_structure_1989,alice}.
For the relational model,
decomposition techniques are primarily concerned with \emph{\aclp{fd}} (cf. \Cref{sec:preliminaries} and \Cref{tab:rel-norm-form-characteristics}) to achieve \ac{3nf} and \ac{bcnf} schemas.
For \acp{fd},
a sound and complete set of reasoning rules, 
the so-called Armstrong's Axioms,
is known and can be used, for example,
to identify redundant dependencies~\cite{armstrong_dependency_1974,delobel_decomposition_1973,alice}.

\emph{Normalization} is the process of decomposing a given relation schema $R$ in order to reach \ac{2nf}, \ac{3nf}, and \ac{bcnf}, effectively
splitting $R$ into $n$ new relation schemas $R_1,\dots,R_n$ based on its attributes, meaning $\attr(R)=\bigcup_{i=1}^n\attr(R_i)$.
Algorithms that perform decomposition are the \textit{synthesis algorithm}~\cite{alice}, 
which guarantees \ac{3nf} for the decomposed relation schemas, 
and the \textit{decomposition algorithm}~\cite{alice},
which ensures \ac{bcnf}.
Until today, normalization remains a crucial task when designing (relational) schemas \cite{saake_datenbanken_2018},
as it omits anomalies (cf. \cite{alice}), reduces redundancies and the potential of inconsistencies,
and it optimizes the performance of databases~\cite{thalheim_entity_2000}.

\subsection{Non-Relational Normalization}
Normalization can be applied beyond the relational model. 
\citet{thalheim_entity_2000} introduced \changed{\emph{normalization for the \acf{er} model}, led by the same objectives as relational normalization: operation anomalies, inconsistent data, and redundancies should be avoided. }%
Unlike relational normalization, which considers only single relations
(\emph{local normalization}), \ac{er} normalization can  consider whole diagrams or parts of diagrams (\emph{global \changed{normalization}}).

\changed{
\ac{er} normalization considers the normalization problem, where the goal is to determine whether a transformation (called ``infomorphism'' in~\cite{thalheim_entity_2000, thalheim_schema_2020}) $\varInfomorphism$ exists that can be applied to an \ac{er} diagram $E$ such that it fulfills the properties of a particular normal form (cf. \cite{thalheim_entity_2000} for details on the properties). One important property is that these transformations need to be \emph{reversible}, i.e., $\varInfomorphism_\text{inv}(\varInfomorphism(E)) = E$ holds for $\varInfomorphism$ and its inverse $\varInfomorphism_\text{inv}$. 
}

\emph{Normalization for \acp{lpg}} was first proposed by \citeauthor{skavantzos_third_2025}~\cite{skavantzos_normalizing_2023,skavantzos_third_2025}. 
\acp{lpg} are normalized using a decomposition procedure. 
In simplified terms, 
the principle is to first transform a set of nodes with properties $P$ and \changed{labels $L$ into} a relation with attributes $P$. 
Then, this relation is decomposed,
and, finally, the obtained relations are transformed back into an \ac{lpg}. 
This approach therefore reduces redundancy by resolving within-node dependencies but does not consider the additional types of dependencies that can occur within edges or graph patterns \changed{(i.e., between-graph-object dependencies).}
Hence, in this paper, we propose a generalized approach to close this gap by considering these types of dependencies as well. 

\subsection{Functional Dependencies for Graphs}
In the context of \ac{lpg} normalization, 
\changed{\citeauthor{skavantzos_uniqueness_2021}~\cite{skavantzos_third_2025,skavantzos_normalizing_2023,skavantzos_uniqueness_2021}} proposed \textit{\acp{gfd}} to express \acp{fd} within nodes, 
and \textit{\acp{guc}} to express keys.
\acp{gfd} and \changed{\acp{guc}} are defined on \changed{a subset of} properties of a subset of 
nodes $\setNodeIDs$ that
have assigned a set of labels $L \subseteq \infSetLabels$ and have defined values for property keys $P \subseteq \infSetPropKeys$.
The semantics of \acp{gfd} is similar to the semantics of \acp{fd}.
For example, the \ac{gfd} 
\begin{minipage}{\columnwidth}
\vspace*{-0.3em}
\footnotesize
\begin{align*}
    \changed{\varDep_\text{\ac{gfd}-Annual}}:=&\{\textsf{Course},\textsf{Annual}\}:
    \{\textsf{title},\textsf{year},\textsf{semester}\}:\\&
    \{\textsf{title},\textsf{year}\}\determ\{\textsf{semester}\}
\end{align*}
\end{minipage}
expresses dependency \hyperref[dep:annual]{\changed{$\varDep_\text{Annual}$}} (cf. Section~\ref{sec:preliminaries}).
Similarly, a \ac{guc} \\
\begin{minipage}{\columnwidth}
\footnotesize
$$\changed{\varDep_\text{\ac{guc}-Course}}:={\{\textsf{Course}\}:\{\textsf{title},\textsf{year}\}:\{\textsf{title},\textsf{year}\}}$$
\end{minipage}
expresses dependency \hyperref[dep:course]{$\varDep_\text{Course}$},
i.e., courses are identified by their title and year.
Notably, \acp{gfd} and \acp{guc} are concerned with nodes.
In our terminology, they can only express within-node dependencies.

Besides \acp{gfd} and \ac{guc},
a variety of other dependencies \changed{applicable to \acp{lpg}} has been defined. 
Including \acp{gfd}, 
\citet{manouvrier_graph_2025} analyzed \num{13} types of dependencies.
However, the focus of most of these dependency types is on the identification of graph objects, 
on the specification of integrity constraints involving constant values, 
or on 
inference of graph objects.
Except for the previously mentioned \acp{gfd}, 
none of the dependency types was designed for the purpose of \emph{normalizing} graphs~\cite{manouvrier_graph_2025}. 

\changed{Not for \acp{lpg}, but for arbitrary $n$-ary relations \citet{hellings_implication_2014} 
defined \acfp{fc}. They define an axiomatic system to allow reasoning over \acp{fc}. In particular, \acp{fc} can be used to define \acp{fd} over RDF-based graphs.
}

\changed{Since \Cref{sec:normalization} defines a \emph{graph-native} approach to normalization based on dependencies involving nodes and edges, we need a dependency type tailored to this task.}
\vspace{-0.5em}

\section{Graph Object Functional Dependencies}
\label{sec:dependency}
As a first step towards graph-native normalization,
we define \acfp{gnfd} as a dependency type for \acp{lpg} that goes beyond within-node dependencies.
A \ac{gnfd} $\varDep$ is an expression of the form:\\[-0.4em]
\begin{minipage}{\columnwidth}
    \[
  \varDep ::={} \varQueryPattern :: X \determ Y
\]
\end{minipage}\\
where $\varQueryPattern$ is a graph pattern and $X,Y$ are sets of
variables (i.e., object identity- or property variables; cf. Section~\ref{sec:preliminaries}).
We refer to the pattern~$\varQueryPattern$ of a $\ac{gnfd}$ as its \emph{scope}, and to $X\determ Y$ as its \textit{descriptor} $\varDescriptor$.
We call a \ac{gnfd} $Q::X\determ Y$ \emph{trivial} if $Y\subseteq X$. 

Whether a dependency \changed{is a \emph{within}- or \emph{between}-graph-object dependency} simply depends on the number of graph object identity variables involved in the dependency.
We make this more precise:
\begin{definition}
    We call a \ac{gnfd} $\varphi = Q::X\determ Y$ a \emph{within}-graph-object dependency when the variables in $X\cup Y$ use exactly one object identity variable $x$. If $x$ occurs as a node (resp. edge) in the scope $Q$, we call it a \emph{within-node} (resp. edge) dependency.
    Otherwise, when two object identity variables occur in $X\cup Y$, we call $\varphi$ a \emph{between}-graph-object dependency.
\end{definition}

Recall that graph patterns $\varQueryPattern$ match subgraphs in a given graph~$\varPropertyGraph$, 
with the requirement
that certain sets of labels $L$ and keys $P$ must be present in the matched graph objects.
The result of evaluating a graph pattern in a given graph is a relation where the relation schema consists of attributes that are given by the variables mentioned in the pattern.
Then, \ac{fd} $X\determ Y$ must hold in this relation.

Formally, a graph $\varGraph$ satisfies a \ac{gnfd} $\varQueryPattern::X\determ Y$ if
$\sem\varQueryPattern\varGraph \sat X\determ Y$.
Furthermore, we call a set of \acp{gnfd} $\Sigma$ a GN-Schema. 
Therefore, more generally, we say a graph $G$ satisfies a schema $\Sigma$, 
denoted by $G\sat \Sigma$ if it satisfies each of the dependencies in $\Sigma$.

\begin{example}\label{ex:formal-deps}
By using \acp{gnfd}, we are able to express all textually described dependencies of the running example (cf. \Cref{sec:preliminaries}).\\
\begin{minipage}{\columnwidth}
\vspace{0.2em}
    {\footnotesize
    \begin{itemize}[leftmargin=5em, listparindent=0em]
        \item[\changed{$\varDep_\text{Course}$}:]
        $(\pobject x{\{\mathsf{Course}\}}{\{\mathsf{title},\mathsf{year}\}}) :: x.\mathsf{title}, x.\mathsf{year}\determ x$\vspace{0.3em}
        \item[\changed{$\varDep_\text{Book}$}:] $()\pedgeRight{\pobject y{\{\mathsf{Teaches}\}}\{\mathsf{using}\}}(\pobject {x}{\{\mathsf{Course}\}}{\emptyset}):: x\determ y.\mathsf{using}$\vspace{0.3em}
        
        \item[\changed{$\varDep_\text{GroupNo}$}:] $()\pedgeRight{\pobject y{\{\mathsf{inGroupWith}\}}\{\mathsf{groupNo},\mathsf{name}\}}() ::$ $y.\mathsf{groupNo}\determ y.\mathsf{name}$\vspace{0.3em}
        \item[\changed{$\varDep_\text{Annual}$}:] $(\pobject x{\{\mathsf{Course},\mathsf{Annual}\}}\{\mathsf{title},\mathsf{year},\mathsf{sem}\}) ::$ $ x.\mathsf{title}, x.\mathsf{year}\determ x.\mathsf{sem}$
    \end{itemize}
    }
    \end{minipage}\vspace{0.2em}\\
Dependency \hyperref[dep:course]{\changed{$\varDep_\text{Course}$}} acts as a uniqueness (superkey) constraint and shows how to express \acp{guc} using \acp{gnfd}. 
Dependency \hyperref[dep:book]{\changed{$\varDep_\text{Book}$}} highlights our generalization of \acp{gfd}; we are able to express dependencies \emph{between} graph objects.
Similarly, \hyperref[dep:groupno]{\changed{$\varDep_\text{GroupNo}$}} shows how \acp{gnfd} can express within-edge dependencies. 
Finally, example \hyperref[dep:annual]{\changed{$\varDep_\text{Annual}$}} shows how to express \acp{gfd} using within-node \acp{gnfd}. 
\end{example}

\paragraph{Axiomatization}
Functional dependencies in the relational model have an axiomatization called \emph{Armstrong's axioms}.
Analogously, we also have reasoning rules for \acp{gnfd}.
First, we state axioms that are intrinsic to \acp{lpg}. 
We call the following three axioms \emph{structurally implied} because they are given by the graph structure of the data model.
Let $G$ be a graph.
\begin{description}
    \item[Object Identity (Node)] $G\models (\pobject o\emptyset k):: o\determ o.k$ 
    \item[Object Identity (Edge)] $G\models ()\pedgeRight{\pobject o\emptyset k}():: o\determ o.k$
    \item[Edge Endpoint] $G\models (x)\pedgeRight{e}(y):: e\determ x$ and $e\determ y$ 
\end{description}
Before we continue with the associated reasoning rules,
we need to introduce the notion of \emph{pattern inclusion}.
Let $Q,Q'$ be basic graph patterns.
We say $Q$ is more general than $Q'$, 
which we write as $Q\sqsupseteq Q'$ if,
for any graph $G$, $\sem QG \supseteq \pi_{\attr(Q)}(\sem {Q'}G$).\footnote{We assume the standard relational projection semantics for $\pi_{\attr(Q)}$, i.e., we project on the attributes of $Q$.} 
We characterize basic graph pattern inclusion as follows:
\begin{table}
\small
    \centering
    \footnotesize
        \caption{Conditions for $Q\sqsupseteq Q'$}
        \vspace*{-1em}
    \begin{tabular}{cl L{2.2cm}}
    \toprule
        $Q$ & $Q'$ & $Q\sqsupseteq Q'$ if \\
        \midrule
        $(\pobject xLP)$ & $(\pobject {x'}{L'}{P'})$ & $L'\subseteq L$ and $P'\subseteq P$ \\  
        $(\pobject xLP)$ & $(\pobject {x'}{L'}{P'})\pedgeRightLeft{\pobject y{L''}{P''}}()$ & $L'\subseteq L$ and $P'\subseteq P$ \\
         $(\pobject {x}{L_x}{P_x})\pedgeRightLeft{\pobject y{L_y}{P_y}}()$ & $(\pobject {x'}{L_x'}{P_x'})\pedgeRightLeft{\pobject y{L_y'}{P_y'}}()$ & $L_x'\subseteq L_x$, $P_x'\subseteq P_x$;\newline  $L_y'\subseteq L_y$, $P_y'\subseteq P_y$ \\
    \bottomrule
\end{tabular}
    \label{tab:moregeneralthan}
\end{table}
\begin{lemma}\label{lem:moregeneralthan}
    $Q\sqsupseteq Q'$ iff they are as shown in Table~\ref{tab:moregeneralthan}.
\end{lemma}
\begin{proof}[Proof sketch]
    One can verify that the relations shown in Table~\ref{tab:moregeneralthan} hold. Furthermore, the remaining case is where an edge occurs in $Q$, but not in $Q'$. Suppose $Q := (\pobject {x}{L}{P})\pedgeRight{\pobject y{L'}{P'}}()$ and $Q' := (\pobject {x'}{L''}{P''})$. Then, even if $L''\subseteq L$ and $P''\subseteq P$, pattern $Q$ is more restrictive since it requires the presence of an edge.
\end{proof}
Furthermore, Table~\ref{tab:moregeneralthan} gives us a straightforward algorithm to check whether $\moregeneralthan{Q}{Q'}$ by matching the patterns and checking the subsets of property keys and labels.

\changed{We now state the reasoning rules for \acp{gnfd}.}
Let $Q,Q'$ be basic graph patterns.
\begin{description}
    \item[Reflexivity] If $Y\subseteq X$, then $\varQueryPattern::X\determ Y$.
    \item[Augmentation] If $\varQueryPattern::X\determ Y$ and
  $Z\subseteq \attr(\varQueryPattern)$, then
      $\varQueryPattern::XZ\determ YZ$.
    \item[Decomposition] If $\varQueryPattern::X\determ ZY$, then
      $\varQueryPattern::X\determ Y$.
    \item[Restriction] If $Q::X\determ Y$, and $\moregeneralthan{Q}{Q'}$, then $Q'::X\determ Y$.
    \item[Transitivity] If $\varQueryPattern::X\determ Y$, $\varQueryPattern'::Y\determ Z$ and $\moregeneralthan{Q'}{Q}$, then $\varQueryPattern':: X\determ Z$.
\end{description}

\begin{theorem}
  The axioms above provide a sound and complete axiomatization of \acp{gnfd}.
\end{theorem}
\begin{proof}[proof sketch]
  In case $Q = Q'$ the soundness and completeness of the axiomatization, within the scope $Q$, follows from Armstrong's Axioms. However, generally, the completeness of the axiomatization is non-trivial. To show completeness,
  we rely on the insight that \acp{gnfd} can be encoded into \changed{\emph{\acfp{fc}}} over $n$-ary relations, 
  as defined by \citeauthor{hellings_implication_2014}~\cite{hellings_implication_2014}. 
  In their work, they provide a complete axiomatization of \changed{\acp{fc}}, which we can use to show the completeness of our axiomatization.
  To this end, we must verify that (1) we can encode \acp{lpg} in a $n$-ary relation; (2) we can translate \acp{gnfd} into \changed{\acp{fc}}; and (3) our axioms correspond to theirs.
  
  First, for (1), any \acp{lpg} can be encoded into a single relation with arity 4, e.g., as done in \emph{domain graphs}~\cite{vrgoc_milleniumdb_2023}. 
  Then, for (2), we have to show that \changed{\acp{fc}} can express \acp{gnfd}:
  graph patterns without repetitions can be expressed over this 4-ary relation by `patterns' from \citet{hellings_implication_2014}, since they are essentially conjunctive queries.
  Furthermore, their constraint itself \emph{is} an FD, as are our descriptors. 
  For (3) one can verify that we faithfully adapted \citeauthor{hellings_implication_2014}'s~\cite{hellings_implication_2014} axioms in a straightforward way.\footnote{Interestingly, this argument captures the axiomatization of a generalization of \acp{gnfd} with graph patterns without repetitions.
  However, such dependencies are out of the scope of this work.}
\end{proof}

\changed{
\paragraph{Obtaining \acp{gnfd}}
\acp{gnfd} encode domain knowledge of their scopes and are specified by a curator, typically with implicit knowledge of the graph's construction and the subgraphs relevant to the application. \citet{skavantzos_third_2025} report that even for within-node dependencies, automated mining gives only heuristics; mining \acp{gnfd} over more complex graph patterns is correspondingly harder, and automatic retrieval of \acp{gnfd} remains an open problem.
}

\section{Redundancies in Labeled Property Graphs}
\label{sec:native-property-graph-normal-forms}
Since one of the main objectives of normalization is to reduce redundancy, we first identify different types of redundancies that might occur in \acp{lpg} under \acp{gnfd} and show how these can be
resolved in a graph-native way.
Based on these insights, we introduce \acfp{gnnf}.

\subsection{Redundancy Patterns}
Graph data is called \emph{redundant if it represents the same piece of information \changed{multiple times}}. 

\begin{example}\label{ex:redundancy}
Consider again the running example in \Cref{fig:student-in-group} \changed{and dependency~\hyperref[dep:groupno]{$\varDep_\text{GroupNo}$}}, which states that a student group's number determines \changed{the name of the group they are in}. The fact that group 1 is named ``Heroes'' is represented twice, in edges \texttt{e4} and \texttt{e5}.    
    \Cref{fig:red-in-graphs} visually displays the redundancy of the group name in these edges. Preferably, the group name ``Heroes'' \changed{should} only be represented once in combination with group number 1.

    We can characterize these redundancies from a relational perspective. 
    Consider the relation depicted in \Cref{fig:red-in-rel}, which shows the relation represented by the scope of dependency~\changed{\hyperref[dep:groupno]{$\varDep_\text{GroupNo}$}}, including the source and target of the edges. Clearly, if we look at the attributes used in the descriptor of~\hyperref[dep:groupno]{\changed{$\varDep_\text{GroupNo}$}}, i.e., \textsf{groupNo} and \textsf{name}, we can identify the repeated fact.
\end{example}

\begin{figure}
    \centering
    \begin{subfigure}{\linewidth}
    \centering
    \footnotesize
\begin{tikzpicture}[node distance=1.5cm]

\node[vertex style=Node,align=center,minimum height=0.4cm,minimum width=0.6cm] (7) {
\texttt{n7}};

\node[vertex style=Node, right of=7,align=center, xshift=3.5em,minimum height=0.4cm,minimum width=0.6cm] (8) {
\texttt{n8}};

\node[vertex style=Node, right of=8,align=center,minimum height=0.4cm,minimum width=0.6cm] (9) {
\texttt{n9}};

\node[vertex style=Node, right of=9,align=center, xshift=3.5em,minimum height=0.4cm,minimum width=0.6cm] (81) {
\texttt{n8}};

\draw [->, align=center] (7) to   node[text style,rotate=0,below,align=left,yshift=1.25em] {\textbf{:inGroupWith}\\\texttt{e4}\\groupNo: \phantom{``}1\\\dotuline{name: ``Heroes''}} (8);
\draw [->, align=center] (9) to   node[text style,rotate=0,below,align=left,yshift=1.25em] {\textbf{:inGroupWith}\\\texttt{e5}\\groupNo: \phantom{``}1\\\dotuline{name: ``Heroes''}} (81);

\draw[->, color=gray] (0.4,-0.35) to [bend right=45] (0.4,-0.65);
\draw[->, color=gray] (4.25,-0.35) to [bend right=45] (4.25,-0.65);

\end{tikzpicture}
    \caption{Redundancy within the edges}
    \label{fig:red-in-graphs}
    \end{subfigure}

    \begin{subfigure}{\linewidth}
    \centering
    \footnotesize

    \begin{tabular}{l r l l l} 
        & \multicolumn{2}{c}{
\begin{tikzpicture}
    \draw[->, color=gray] (0,0) to [bend left=45] (1,0);
\end{tikzpicture}
        
        }    \\
        \toprule
        $x$ & $x.\mathsf{groupNo}$ & $x.\mathsf{name}$ & $\hat{x}_{\mathsf{src}}$ & $\hat{x}_{\mathsf{tgt}}$ \\
        \midrule
        \texttt{e4} & \multicolumn{2}{r}{\dotuline{1\hspace{1.6em} Heroes}} & \texttt{n7} & \texttt{n8} \\
        \texttt{e5} & \multicolumn{2}{r}{\dotuline{1\hspace{1.6em} Heroes}}  & \texttt{n9} & \texttt{n8} \\
        \bottomrule

    \end{tabular}

    \caption{Redundancy from a relational perspective}
    \label{fig:red-in-rel}
    \end{subfigure}

    \begin{subfigure}{\linewidth}
    \centering
    \footnotesize

    \begin{tabular}[b]{l r l l@{~}} 
    $R_\text{edge}$\vspace{0.2em}\\
        \toprule
        $x$ & $x.\mathsf{groupNo}$ & $\hat{x}_{\mathsf{src}}$ & $\hat{x}_{\mathsf{tgt}}$ \\
        \midrule
        \texttt{e4} & 1 & \texttt{n7} & \texttt{n8} \\
        \texttt{e5} & 1  & \texttt{n9} & \texttt{n8} \\
        \bottomrule
    \end{tabular}
\hspace{2em}
    \begin{tabular}[b]{r l} 
         \multicolumn{2}{l}{$R_\text{groupNo}$\vspace{0.2em}\hspace{1em}
\begin{tikzpicture}
    \draw[->, color=gray] (0,0) to [bend left=45] (1,0);
\end{tikzpicture}
        
        }    \\
        \toprule
         $x.\mathsf{groupNo}$ & $x.\mathsf{name}$ \\
        \midrule
        \multicolumn{2}{r}{1\hspace{1.6em} Heroes} \\
        \bottomrule \\

    \end{tabular}

\vspace*{-6.1em}
\hspace*{-4em}
\begin{tikzpicture}
    \draw[->, dashed] (0,0) to [bend left=20] (2.7,-0.1);
\end{tikzpicture}
\vspace{4.5em}
    \vspace*{-1em}
    \caption{Decomposed relations. The dashed line is a foreign key; the \textcolor{gray}{gray} line is the descriptor of the original dependency.}
    \label{fig:decomposed_sketch}
    \end{subfigure}

    \begin{subfigure}{\linewidth}
    \centering
    \footnotesize
\begin{tikzpicture}[node distance=1.5cm]

\node[vertex style=BurntOrange,align=center,minimum height=0.4cm,minimum width=0.6cm] (7) {
\texttt{n7}};

\node[vertex style=BurntOrange,align=center,minimum height=0.4cm,minimum width=0.6cm, right of=7, yshift=0.75cm] (e4) {
\texttt{e4}};

\node[vertex style=BurntOrange, right of=e4,align=center, minimum height=0.4cm,minimum width=0.6cm, yshift=0.75cm] (8) {
\texttt{n8}};

\node[vertex style=BurntOrange,align=center,minimum height=0.4cm,minimum width=0.6cm, right of=8,yshift=-0.75cm] (e5) {
\texttt{e5}};
\node[vertex style=BurntOrange, right of=e5,align=center,minimum height=0.4cm,minimum width=0.6cm,yshift=-0.75cm] (9) {
\texttt{n9}};

\node[vertex style=BurntOrange, below of=8,align=center,minimum height=0.4cm,minimum width=0.6cm] (o) {
\texttt{o1}};
\node[text style,rotate=0,below,align=left,below of = o, yshift=1cm] {groupNo: \phantom{``}1\\name: ``Heroes''};

\begin{scope}[on background layer]
        \fill[yellow,opacity=0.07] ($(7)-(0.5,0.8)$) rectangle ($(8)+(0.8,0.3)$);
        \fill[blue,opacity=0.07] ($(9)-(-0.5,0.8)$) rectangle ($(8)+(-0.8,0.3)$);

    \end{scope}

\draw [->, align=center] (7) to  
node[text style,rotate=0,below,align=left,yshift=1.25em] {\textbf{:iGW\textsubscript{src}~}
} 
(e4);
\draw [->, align=center] (e4) to  
node[text style,rotate=0,below,align=left,yshift=1.25em] {\textbf{:iGW\textsubscript{tgt}~~}
} 
(8);
\draw [->, align=center] (9) to  
node[text style,rotate=0,below,align=left,yshift=1.25em] {\textbf{~:iGW\textsubscript{src}}
} 
(e5);
\draw [->, align=center] (e5) to  
node[text style,rotate=0,below,align=left,yshift=1.25em] {\textbf{~:iGW\textsubscript{tgt}}
} 
(8);
\draw [->, align=center] (e4) to  
node[text style,rotate=0,below,align=left,yshift=1.25em] {\textbf{~:GDet.}} 
(o);
\draw [->, align=center] (e5) to  
node[text style,rotate=0,below,align=left,yshift=1.25em] {\textbf{~:GDet.}} 
(o);

\draw[->, color=gray] (2.4,-0.3) to [bend left=20] (2.67,-0.1);

\end{tikzpicture}
    \caption{Transformed subgraph of \Cref{fig:red-in-graphs} without redundancies}
    \label{fig:transformed_sketch}
    \end{subfigure}
    \vspace*{-2.3em}
    \caption{Resolving redundancies in \Cref{fig:example} related to \hyperref[dep:groupno]{$\varDep_\text{GroupNo}$}: the original redundancy (a, b), its relational decomposition (c), and the transformed subgraph (d).}
    \label{fig:redundancy_sketch}
    \vspace*{-0.3em}
\end{figure}

To simplify our discussion of redundancies,
we assume that the dependencies are \emph{strictly} within or between objects, 
i.e., for a descriptor $X\determ Y$ at most one object identity variable may appear in \changed{each}, $X$ and $Y$. 
We call such dependencies \emph{strict} \acp{gnfd}. 
They capture natural dependencies, such as the ones used in our examples. 

\begin{figure}
    \centering
    
    \begin{subfigure}{\columnwidth}
    \footnotesize
    \centering

    \begin{tikzpicture}[node distance=1em]
            \node[vertex style=Node,align=center,minimum height=0.4cm,minimum width=0.6cm] (1) {$x_1$};

\node[text style,right of=1,align=center,yshift=0.7em,xshift=2em] (2_id) {$y$};

\draw [<-, bend left=50,color=gray, align=center] (1.north east) to  (2_id.north);

\node[vertex style=Node, right of=1,align=center,minimum height=0.4cm,minimum width=0.6cm,xshift=6em] (3) {$x_2$}
edge [<-,align=center]  (1);

\end{tikzpicture} \hspace{1em}
    \begin{tikzpicture}[node distance=1em]
            \node[vertex style=Node,align=center,minimum height=0.4cm,minimum width=0.6cm] (1) {$x_1$};

\node[text style,right of=1,align=center,yshift=0.7em,xshift=2em] (2_id) {$y$};
\node[text style,right of=1,align=center,yshift=-0.7em,xshift=2em] (2_props1) {$k$};
 
\draw [<-, bend left=30,color=gray, align=center] (2_props1.west) to  (2_id.west);

\node[vertex style=Node, right of=1,align=center,minimum height=0.4cm,minimum width=0.6cm,xshift=6em] (3) {$x_2$}
edge [<-,align=center]  (1);

\end{tikzpicture} 
       
    \begin{tikzpicture}[node distance=1em]
            \node[vertex style=Node,align=center,minimum height=0.4cm,minimum width=0.6cm] (1) {$x_1$};
\node[text style,below of=1,align=center,yshift=-0.6em] (1_props1) {$k$};

\node[text style,right of=1,align=center,yshift=0.7em,xshift=2em] (2_id) {$y$};
 
\draw [<-, bend right=20,color=gray, align=center] (1_props1.east) to  (2_id.west);

\node[vertex style=Node, right of=1,align=center,minimum height=0.4cm,minimum width=0.6cm,xshift=6em] (3) {$x_2$}
edge [<-,align=center]  (1);

\end{tikzpicture} \hspace{1em}
    \begin{tikzpicture}[node distance=1em]
            \node[vertex style=Node,align=center,minimum height=0.4cm,minimum width=0.6cm] (1) {$x_1$};
\node[text style,below of=1,align=center,yshift=-0.6em] (1_props1) {$k$};

\node[text style,right of=1,align=center,yshift=0.7em,xshift=2em] (2_id) {$y$};
 
\draw [<-, bend left=30,color=gray, align=center] (1_props1.west) to  (1.west);

\node[vertex style=Node, right of=1,align=center,minimum height=0.4cm,minimum width=0.6cm,xshift=6em] (3) {$x_2$}
edge [<-,align=center]  (1);

\end{tikzpicture}

    \vspace*{-1em}
    \caption{Structurally implied \acp{gnfd}}
    \label{fig:no-transform-struct-impl}
    \end{subfigure}
    \begin{subfigure}{0.49\columnwidth}
    \centering
    \footnotesize
    \begin{minipage}[t][2.2cm][b]{0.5\textwidth}
    \centering
    \begin{tikzpicture}[node distance=1em]
            \node[vertex style=Node,align=center,minimum height=0.4cm,minimum width=0.6cm] (1) {$x_1$};
\node[text style,below of=1,align=center,yshift=-0.6em] (1_props1) {$k$};

\node[text style,right of=1,align=center,yshift=0.7em,xshift=2em] (2_id) {$y$};
 
\draw [->, bend left=50,color=gray, align=center] (1_props1.west) to  (1.west);

\node[vertex style=Node, right of=1,align=center,minimum height=0.4cm,minimum width=0.6cm,xshift=6em] (3) {$x_2$}
edge [<-,align=center]  (1);

\end{tikzpicture}
    \begin{tikzpicture}[node distance=1em]
            \node[vertex style=Node,align=center,minimum height=0.4cm,minimum width=0.6cm] (1) {$x_1$};

\node[text style,right of=1,align=center,yshift=0.7em,xshift=2em] (2_id) {$y$};
\node[text style,right of=1,align=center,yshift=-0.7em,xshift=2em] (2_props1) {$k$};
 
\draw [->, bend left=50,color=gray, align=center] (2_props1.west) to  (2_id.west);

\node[vertex style=Node, right of=1,align=center,minimum height=0.4cm,minimum width=0.6cm,xshift=6em] (3) {$x_2$}
edge [<-,align=center]  (1);

\end{tikzpicture}
    \end{minipage}
    \vspace*{-1.5em}
    \caption{Superkey \acp{gnfd}}
    \label{fig:no-transform-key}
    \end{subfigure}
    \begin{subfigure}{0.49\columnwidth}
    \centering
    \footnotesize
    \begin{minipage}[t][2.2cm][b]{0.5\textwidth}
    \centering
    \begin{tikzpicture}[node distance=1em]
            \node[vertex style=Node,align=center,minimum height=0.4cm,minimum width=0.6cm] (1) {$x_1$};

\node[text style,right of=1,align=center,yshift=0.7em,xshift=2em] (2_id) {$y$};
 
\draw [->, bend left=50,color=gray, align=center] (1.north east) to  (2_id.north);

\node[vertex style=Node, right of=1,align=center,minimum height=0.4cm,minimum width=0.6cm,xshift=6em] (3) {$x_2$}
edge [<-,align=center]  (1);

\end{tikzpicture}
    
    \begin{tikzpicture}[node distance=1em]
            \node[vertex style=Node,align=center,minimum height=0.4cm,minimum width=0.6cm] (1) {$x_1$};
\node[text style,below of=1,align=center,yshift=-0.6em] (1_props1) {$k$};

\node[text style,right of=1,align=center,yshift=0.7em,xshift=2em] (2_id) {$y$};
 
\draw [->, bend right=20,color=gray, align=center] (1_props1.east) to  (2_id.west);

\node[vertex style=Node, right of=1,align=center,minimum height=0.4cm,minimum width=0.6cm,xshift=6em] (3) {$x_2$}
edge [<-,align=center]  (1);

\end{tikzpicture}
    \end{minipage}
    \vspace*{-1.5em}
    \caption{Cardinality-limiting \acp{gnfd}}
    \label{fig:no-transform-cardinality}
    \end{subfigure}
    \vspace*{-1em}
    \caption{\acp{gnfd} of different shapes for which no redundancies can occur. \ac{gnfd} descriptors are drawn in \textcolor{gray}{gray}. \changed{For simplicity, only one direction of edges is shown.}}
    \label{fig:dep_no_redundancy}
\end{figure}

To aid our discussion of redundancy, 
we discuss the \emph{shape} of a \ac{gnfd}:
the shape is determined by inspecting the \emph{descriptor} $X\determ Y$.
Specifically, our discussion of redundancy relies on the \changed{types} of variables that are present in $X$ and $Y$ (recall that variables are either identity- or property variables, cf. Section~\ref{sec:preliminaries}).

First, \changed{as shown in \Cref{fig:dep_no_redundancy}},
we identify three groups of \acp{gnfd} for which no redundancies occur.\footnote{While only one property key $k$ is depicted in these dependencies, this is only a graphical simplification: generally, these can be sets of property keys.} 
Each group contains a number of dependencies of different shapes, using the graphical notation. 
The source and target of the gray arrow indicate the shape. \changed{For simplicity, only one direction is shown for the edges. Still, the same classes also apply to the inverse direction.}

\begin{theorem}
    No redundancies occur under dependencies \changed{that follow the shapes shown} in \Cref{fig:dep_no_redundancy}.
\end{theorem}

\edbtarxiv{
\begin{proof}[Proof sketch]
Groups (a)--(c) cover patterns for which the same fact cannot appear twice: (a) follows from the structurally implied rules (each edge has unique source/target; each property has at most one value); (b) are superkey constraints, which, by definition, do not allow for repetition; (c) constrains the out-degree to one.
\end{proof}
\vspace{-0.9em}
}{
\begin{proof}[proof]
    Group (a) shows dependencies that are satisfied by every graph.
    They can be derived from the structurally implied dependencies in combination with the reasoning rules (cf. \Cref{sec:dependency}).
    \changed{For such patterns, it is not possible to express the same fact twice:}
    for any given edge, there is exactly one source and target. 
    Similarly, for any graph object, there can be at most one value for a property key.
    Furthermore, any change in the graph structure will result in a graph that satisfies similar dependencies. 
    Group (b) depicts dependencies that act as graph object (super)keys: 
    either for nodes, or for edges. 
    As for the relational setting,
    keys are the only kind of dependencies that are desired to be present in your data. 
    Concretely, if $k$ is a superkey for a certain node, 
    this cannot be represented twice due to the nature of superkeys.
    Finally, group (c) shows dependencies that limit the number of outgoing edges of a node to be at most 1. 
    By definition, such a fact cannot be stated twice.
\end{proof}
}

All other dependency shapes allow for redundancies to occur.
We call dependencies of such shapes \emph{redundancy patterns}. 
They are listed in the second column of Table~\ref{tab:transformation-patterns} ($T$). 

\begin{table}[]
    \centering
    \caption{Graphical representation of the graph transformations that resolve redundancy patterns. 
    The descriptor of a \acp{gnfd} is shown as a \textcolor{gray}{gray} arrow. New nodes are constructed on values $x$ matched in the redundancy pattern (shown using the $\newnode{x}$ notation), similarly, new labels are constructed (shown using the $\newlbl{x}$ notation).}
    \vspace{-1em}
    \label{tab:transformation-patterns}
    \footnotesize
    \begin{tabular}{@{~}l@{\enspace}L{2.7cm}@{\hspace{0.8em}}L{4.1cm}@{}}
    \toprule
    \textbf{\acs{id}} & \textbf{Redundancy pattern (\changed{$T$})} & \textbf{Resolved redundancy pattern (\changed{$T_\textsf{new}$} )} \\
    \midrule
        \multicolumn{3}{l}{\textit{Redundancies and transformations associated with within-graph object \acp{gnfd}}} \\
    \hline
    $\varInfomorphism_{\text{within-n}}$  & \begin{tikzpicture}[node distance=1em]
            \node[vertex style=Node,align=center,minimum height=0.4cm,minimum width=0.6cm] (1) {$x{:}L$};
\node[text style,below of=1,align=center,yshift=-0.6em] (1_props1) {$k_1$};
\node[text style,below of=1_props1,align=center] (1_props2) {$k_2$};

\draw [->, bend right=30,color=gray, align=center] (1_props1.west) to  (1_props2.west);

        \end{tikzpicture}  & \begin{tikzpicture}[node distance=1em]
            \node[vertex style=Node,align=center,minimum height=0.4cm,minimum width=0.6cm] (1) {$x{:}L$};

\node[vertex style=Node,align=center,minimum height=0.4cm,minimum width=0.6cm,right of=1,xshift=7em] (2) {$\newnode{Lk_1}{:}\newlbl{Lk_1}$} edge [<-,align=left] node[text style]{{$\newlbl{Lk_1}$}} (1);
\node[text style,below of=2,align=center,yshift=-0.9em] (2_props1) {$k_1$};
\node[text style,below of=2_props1,align=center] (2_props2) {$k_2$};

\draw [->, bend right=100,color=gray, align=center] (2_props1.east) to  (2.east);

        \end{tikzpicture} \\
    $\varInfomorphism_{\text{within-e}}$ & \begin{tikzpicture}[node distance=1em]
            \node[vertex style=Node,align=center,minimum height=0.4cm,minimum width=0.6cm] (1) {$x_1$};

\node[vertex style=Node, right of=1,align=center,minimum height=0.4cm,minimum width=0.6cm,xshift=4.5em] (3) {$x_2$}
edge [<-,align=center]  (1);

\node[text style,right of=1,align=center,yshift=0.5em,xshift=1.5em] (1_id) {$y:L$};
\node[text style,right of=1,align=center,yshift=-0.5em,xshift=1.5em] (1_props1) {$k_1$};
\node[text style,below of=1_props1,align=center] (1_props2) {$k_2$};

\draw [->, bend right=30,color=gray, align=center] (1_props1.west) to  (1_props2.west);

\end{tikzpicture}  & \begin{tikzpicture}[node distance=1em]

\node[vertex style=Node,align=center,minimum height=0.4cm,minimum width=0.6cm] (1) {$\newnode{y}:L$};

\node[vertex style=Node, left of=1,align=center,minimum height=0.4cm,minimum width=0.6cm,xshift=-6em] (0) {$x_1$}
edge [->,align=center]  (1);

\node[text style,right of=1,align=center,yshift=0.7em,xshift=3em] (2_id) {\newlbl{$y_\text{tgt}$}};
\node[text style,right of=0,align=center,yshift=0.7em,xshift=1.7em] (1_id) {\newlbl{$y_\text{src}$}};

\node[vertex style=Node, right of=1,align=center,minimum height=0.4cm,minimum width=0.6cm,xshift=6em] (3) {$x_2$}
edge [<-,align=center]  (1);

\node[vertex style=Node, below of=1,align=center,minimum height=0.4cm,minimum width=0.6cm,yshift=-3em] (4) {$\newnode{Lk_1}:\newlbl{Lk_1}$}
edge [<-,align=center]  (1);

\node[text style,below of=1,align=center,yshift=-0.9em,xshift=1.4em] (3_id) {\newlbl{$Lk_1$}};

\node[text style,right of=4,align=center,yshift=+0em,xshift=4em] (1_props1) {$k_1$};
\node[text style,below of=1_props1,align=center] (1_props2) {$k_2$};

\draw [->, bend right=50, color=gray, align=center] (1_props1.north) to  (4.north east);

\end{tikzpicture} \\
    \midrule
    \multicolumn{3}{l}{\textit{Redundancies and transformations associated with between-graph object \acp{gnfd}}} \\
    \hline
    $\varInfomorphism_{\text{between-n-ep}}$  & \begin{tikzpicture}[node distance=1em]
            \node[vertex style=Node,align=center,minimum height=0.4cm,minimum width=0.6cm] (1) {$x_1$};

\node[text style,right of=1,align=center,yshift=0.7em,xshift=1.5em] (2_id) {$y{:}L$};
\node[text style,right of=1,align=center,yshift=-0.7em,xshift=1.5em] (2_props1) {$k_1$};

\draw [->, bend right=20, color=gray, align=center] (1.south) to  (2_props1.west);

\node[vertex style=Node, right of=1,align=center,minimum height=0.4cm,minimum width=0.6cm,xshift=4.5em] (3) {$x_2$}
edge [<-,align=center]  (1);

        \end{tikzpicture} & \begin{tikzpicture}[node distance=1em]
            \node[vertex style=Node,align=center,minimum height=0.4cm,minimum width=0.6cm] (1) {$x_1$};
\node[text style,below of=1,align=center,yshift=-0.6em] (2_props2) {$k_1$};

\node[text style,right of=1,align=center,yshift=0.7em,xshift=2em] (2_id) {$y$};
\draw [->, bend right=30, color=gray, align=center] (1.west) to  (2_props2.west);

\node[vertex style=Node, right of=1,align=center,minimum height=0.4cm,minimum width=0.6cm,xshift=5em] (3) {$x_2$}
edge [<-,align=center]  (1);

        \end{tikzpicture}\\
    $\varInfomorphism_{\text{between-np-ep}}$  & \begin{tikzpicture}[node distance=1em]
            \node[vertex style=Node,align=center,minimum height=0.4cm,minimum width=0.6cm] (1) {$x_1{:}L$};
\node[text style,below of=1,align=center,yshift=-0.6em] (1_props1) {$k_1$};

\node[text style,right of=1,align=center,yshift=0.7em,xshift=2em] (2_id) {$y$};
\node[text style,right of=1,align=center,yshift=-0.7em,xshift=2em] (2_props1) {$k_2$};
\draw [->, bend right=10, color=gray, align=center] (1_props1.east) to  (2_props1.west);

\node[vertex style=Node, right of=1,align=center,minimum height=0.4cm,minimum width=0.6cm,xshift=4.5em] (3) {$x_2$}
edge [<-,align=center]  (1);

\end{tikzpicture} &    \begin{tikzpicture}[node distance=1em]
            \node[vertex style=Node,align=center,minimum height=0.4cm,minimum width=0.6cm] (1) {$x_1{:}L$};

\node[text style,right of=1,align=center,yshift=0.7em,xshift=2em] (2_id) {$y$};

\node[vertex style=Node, right of=1,align=center,minimum height=0.4cm,minimum width=0.6cm,xshift=5em] (3) {$x_2$}
edge [<-,align=center]  (1);

\node[vertex style=Node, below of=1,align=center,minimum height=0.4cm,minimum width=0.6cm,yshift=-3em] (4) {{$\newnode{Lk_1}{:}\newlbl{Lk_1}$}}
edge [<-,align=center]  (1);

\node[text style,below of=1,align=center,yshift=-0.9em,xshift=1.4em] (3_id) {$\newlbl{Lk_1}$};

\node[text style,right of=4,align=center,yshift=+0em,xshift=3.5em] (1_props1) {$k_1$};
\node[text style,below of=1_props1,align=center] (1_props2) {$k_2$};

\draw [->, bend right=50, color=gray, align=center] (1_props1.north) to  (4.north east);

\end{tikzpicture}\\[4em]
    $\varInfomorphism_{\text{between-ep-np}}$  & \begin{tikzpicture}[node distance=1em]
            \node[vertex style=Node,align=center,minimum height=0.4cm,minimum width=0.6cm] (1) {$x_1$};
\node[text style,below of=1,align=center,yshift=-0.6em] (1_props1) {$k_1$};

\node[text style,right of=1,align=center,yshift=0.7em,xshift=1.5em] (2_id) {$y{:}L$};
\node[text style,right of=1,align=center,yshift=-0.7em,xshift=1.5em] (2_props1) {$k_2$};
 
\draw [<-, bend right=10, color=gray, align=center] (1_props1.east) to  (2_props1.west);

\node[vertex style=Node, right of=1,align=center,minimum height=0.4cm,minimum width=0.6cm,xshift=4.5em] (3) {$x_2$}
edge [<-,align=center]  (1);

\end{tikzpicture} &  \begin{tikzpicture}[node distance=1em]

\node[vertex style=Node,align=center,minimum height=0.4cm,minimum width=0.6cm] (1) {$\newnode{y}{:}L$};

\node[vertex style=Node, left of=1,align=center,minimum height=0.4cm,minimum width=0.6cm,xshift=-6em] (0) {$x_1$}
edge [->,align=center]  (1);

\node[text style,right of=1,align=center,yshift=0.7em,xshift=3em] (2_id) {\newlbl{$ytgt$}};
\node[text style,right of=0,align=center,yshift=0.7em,xshift=1.7em] (1_id) {\newlbl{$ysrc$}};

\node[vertex style=Node, right of=1,align=center,minimum height=0.4cm,minimum width=0.6cm,xshift=6em] (3) {$x_2$}
edge [<-,align=center]  (1);

\node[vertex style=Node, below of=1,align=center,minimum height=0.4cm,minimum width=0.6cm,yshift=-3em] (4) {$\newnode{yk_2}{:}\newlbl{yk_2}$}
edge [<-,align=center]  (1);

\node[text style,below of=1,align=center,yshift=-0.9em,xshift=1.4em] (3_id) {\newlbl{$yk_2$}};

\node[text style,right of=4,align=center,yshift=+0em,xshift=3.5em] (1_props1) {$k_2$};
\node[text style,below of=1_props1,align=center] (1_props2) {$k_1$};

\draw [->, bend right=50, color=gray, align=center] (1_props1.north) to  (4.north east);

\end{tikzpicture}\\[4em]
    $\varInfomorphism_{\text{between-ep-n}}$  & \begin{tikzpicture}[node distance=1em]
            \node[vertex style=Node,align=center,minimum height=0.4cm,minimum width=0.6cm] (1) {$x_1$};

\node[text style,right of=1,align=center,yshift=0.7em,xshift=1.5em] (2_id) {$y{:}L$};
\node[text style,right of=1,align=center,yshift=-0.7em,xshift=1.5em] (2_props1) {$k_1$};
\draw [<-, bend right=20, color=gray, align=center] (1.south) to  (2_props1.west);

\node[vertex style=Node, right of=1,align=center,minimum height=0.4cm,minimum width=0.6cm,xshift=4.5em] (3) {$x_2$}
edge [<-,align=center]  (1);

\end{tikzpicture} & \begin{tikzpicture}[node distance=1em]

\node[vertex style=Node,align=center,minimum height=0.4cm,minimum width=0.6cm] (1) {$\newnode{y}{:}L$};

\node[vertex style=Node, left of=1,align=center,minimum height=0.4cm,minimum width=0.6cm,xshift=-6em] (0) {$x_1$}
edge [->,align=center]  (1);

\node[text style,right of=1,align=center,yshift=0.7em,xshift=3em] (2_id) {\newlbl{$ytgt$}};
\node[text style,right of=0,align=center,yshift=0.7em,xshift=1.7em] (1_id) {\newlbl{$ysrc$}};

\node[vertex style=Node, right of=1,align=center,minimum height=0.4cm,minimum width=0.6cm,xshift=6em] (3) {$x_2$}
edge [<-,align=center]  (1);

\node[vertex style=Node, below of=1,align=center,minimum height=0.4cm,minimum width=0.6cm,yshift=-3em] (4) {$\newnode{yk_2}{:}\newlbl{yk_2}$}
edge [<-,align=center]  (1);

\node[text style,below of=1,align=center,yshift=-0.9em,xshift=1.4em] (3_id) {\newlbl{$yk_2$}};

\node[text style,right of=4,align=center,yshift=+0em,xshift=3.5em] (1_props1) {$k_1$};

\draw [->, bend right=50, color=gray, align=center] (1_props1.north) to  (4.north east);

\end{tikzpicture} \\
    \bottomrule
\end{tabular}
\end{table}
\begin{table*}[]
    \centering
    \footnotesize
    \caption{Rationale for the redundancy removing transformations}
    \label{tab:trnasform-justifications}
    \vspace*{-1em}
      \begin{tabular}{l@{\enspace}L{2.5cm}@{\enspace}L{2.4cm}@{\enspace}L{10.6cm}}
    \toprule
     \textbf{\acs{id}}  & \textbf{Relation schema of the redundancy pattern} & \textbf{Decomposed relation schema} & \textbf{Transformation rationale} \\
     \midrule
              \multicolumn{4}{l}{\textit{Rationales for within-\ac{gnfd} transformations }}  \\
         \hline
$\varInfomorphism_{\text{within-n}}$ & $\{x,x.k_1,x.k_2\}$
& 
$R_1=\{x,x.k_1\}$ 
\newline\vspace{0.5em}
$R_2=\{x.k_1,x.k_2\}$ 
\newline\vspace{0.5em} 
$x.k_1$ in $R_1$ is a \ac{fk}. 
&
$R_1$ has one referencing attribute: the node $x$ itself. We conclude that $R_1$ represents node $x$. $R_2$ has no referencing attributes, and therefore, it represents a new node $n$. Since $R_1$ has a \ac{fk} to $R_2$, we materialize it as an edge between $x$ and $n$. Since we no longer need the FK, $k_1$ needs not be stored in node $x$. \\

$\varInfomorphism_{\text{within-e}}$   & $\{x_1,x_2,y,y.k_1,y.k_2\}$ & 
$R_1=\{x_1,x_2,y,y.k_1\}$ 
\newline\vspace{0.5em}
$R_2=\{y.k_1,y.k_2\}$ 
\newline\vspace{0.5em} 
$y.k_1$ in $R_1$ is a \ac{fk}.
& $R_1$ has three referencing attributes: edge $y$, its source and target. We conclude that $R_1$ represents edge $y$. $R_2$ has no referencing attributes, and therefore, it represents a new node $n$. Since $R_1$ has a FK to $R_2$, we want to materialize it as an edge. Therefore, we reify $y$, with a reifier node $n'$. Then, we create the edge between $n'$ and $n$. Since we no longer need the FK, $k_1$ needs not be stored in node $n'$. \\
\hline
         \multicolumn{4}{l}
         {\textit{Rationales for between-\ac{gnfd} transformations}} \\
         \hline

      $\varInfomorphism_{\text{between-n-ep}}$&
      $\{x_1, x_2, y,y.k_1\}$ & $R_1=\{x_1, x_2, y\}$ \newline\vspace{0.5em}
      $R_2=\{x_1,y.k_1\}$\newline\vspace{0.5em}
      $x_1$ in $R_1$ is a \ac{fk}. & 
      $R_1$ has three referencing attributes: edge $y$, and its source and target. We conclude that $R_1$ re\-presents edge $y$. $R_2$ has one referencing attribute $x_1$, and therefore, we conclude it represents node $x_1$. Additionally, $R_2$ contains a new property key, which represents moving property key $k_1$ from $y$ to $x_1$.
      \\
$\varInfomorphism_{\text{between-np-ep}}$& $\{x_1, x_2, y,x.k_1,y.k_2\}$ & 
$R_1=\{x_1, x_2, y,x.k_1\}$ 
\newline\vspace{0.5em}
$R_2=\{x.k_1,y.k_2\}$
\newline\vspace{0.5em}
      $x.k_1$ is a \ac{fk}. &
      $R_1$ has three referencing attributes: edge $y$, its source and target. We conclude that $R_1$ represents edge $y$.
      $R_2$ has no referencing attributes, and therefore, it represents a new node $n$. Since $R_1$ has a FK using a property from $x_1$, we materialize it as an edge between $x_1$ and $n$. Since we no longer need the FK, $k_1$ needs not be stored in node $x_1$.\\
      
$\varInfomorphism_{\text{between-ep-np}}$& $\{x_1, x_2, y,x.k_1,y.k_2\}$ &
    $R_1=\{x_1, x_2, y,y.k_2\}$ \newline\vspace{0.5em}
      $R_2=\{y.k_2,x.k_1\}$\newline\vspace{0.5em}
      $X.x_1$ is a \ac{fk}. &
      $R_1$ has three referencing attributes: edge $y$, its source and target. We conclude $R_1$ represents the edge $y$. 
      $R_2$ has no referencing attributes, therefore, it represents a new node $n$. 
      Since $R_1$ has a FK using a property key from $y$, we want to materialize it as an edge. Therefore, we reify $y$, with a reifier node $n'$. Then, we create the edge between $n'$ and $n$. Since we no longer need the FK, $k_1$ needs not be stored in node $n'$.\\

$\varInfomorphism_{\text{between-ep-n}}$& $\{x_1,x_2,y,y.k_1\}$ &$
    R_1=\{x_1,x_2,y,y.k_1\}$ \newline\vspace{0.5em}
      $R_2=\{x_1,y.k_1\}$\newline\vspace{0.5em}
      $y.k_1$ is a \ac{fk}. & $R_1$ has three referencing attributes: edge $y$, and its source and target. We conclude that $R_1$ represents edge $y$. $R_2$ has one referencing attribute $x_1$, however, since it needs to be referenced from the edge $y$, an edge will be materialized between the reified node $n'$ and a new node, specially for $k_1$. \\
    \bottomrule
    \end{tabular}    

\end{table*}

\subsection{Redundancy Removing Transformations}
Since our objective is to eliminate redundant graph data, 
we propose a general technique that, 
for each possible redundancy pattern,
eliminates the redundancy by specifying a specific graph transformation (cf. \Cref{sec:preliminaries}). 
\changed{This is similar to \ac{er} normalization~\cite{thalheim_entity_2000}, where also transformations are employed to avoid modeling that allows for redundancies. }
The transformations are shown in \Cref{tab:transformation-patterns}, which depicts the redundancy pattern ($T$), as well as the resolved pattern ($T_\mathsf{new}$). 
Gray arrows denote the dependencies present in the patterns.
For example, the transformation $\varInfomorphism_\text{between-n-ep}$ applies to \acp{gnfd} of the form $(\pobject {x_1}{\emptyset}{\emptyset})\pedgeRight{\pobject{y}{L}{\{k_1,\dots\}}}()::x_1\determ y.k_1$.
Concretely, this covers, e.g., dependency \hyperref[dep:book]{\changed{$\varDep_\text{Book}$}} (cf. \Cref{sec:preliminaries}).

Generally, \changed{the discussed} transformations follow a principled approach: 
elements mentioned in the descriptor \changed{$X\determ Y$} of a \ac{gnfd} belong together. 
More specifically, they are bundled together into a new node $n_\mathit{new}$ \changed{that is identified by $X$.
This %
node $n_\mathit{new}$ additionally must be reachable from the graph object $o$ of which $X$ was initially a part of}. 
If $o$ is a node, an edge is created to reference $n_\mathit{new}$ from $o$.
However, when $o$ is an edge, a different approach must be used. %

For descriptors where the left-hand side consists of variables originating from an edge, we employ a \emph{reification pattern}:
the edge is replaced by a reifier node $n_\mathit{reif}$, 
which has an edge to the newly created node $n_\mathit{new}$. 
These strategies are \edbtarxiv{}{succinctly }summarized as graph transformations in Table~\ref{tab:transformation-patterns}. 
We motivate this approach with an example.

\begin{example}
\label{ex:redtrans}
    Continuing Example~\ref{ex:redundancy}, we start by observing that relation $R$ resulting from the scope of \changed{$\varphi_\text{GroupNo}$,
    as shown in \Cref{fig:red-in-rel},} has three references to graph objects: 
    edge $x$, its source and its target. 
    This is characteristic for edges. 
    Nodes, on the other hand, never refer to any other graph object. 

    We show that the redundant information discussed in Example~\ref{ex:redundancy} is eliminated by executing the transformation $\varInfomorphism_\text{within-e}$ from Table~\ref{tab:transformation-patterns}.
    In the relational setting, we would decompose $R$ into the two relations $R_\mathit{Edge}$ and $R_\mathit{GroupNo}$ shown in \Cref{fig:decomposed_sketch}.
    However, it is not clear from the outset how these two relations correspond to graph objects. 
    Looking at $R_\mathit{GroupNo}$, it refers to no graph objects, which means it is a node in principle.
    On the other hand, $R_\mathit{Edge}$ refers to \emph{three} objects (the edge itself, source, and target), which corresponds to the edge $x$ itself. Since $R_\mathit{Edge}$ additionally also has a foreign key to $R_\mathit{GroupNo}$, we would like to capture this relationship by introducing an edge. However, we cannot create an edge where one of the endpoints is an edge itself.
    We argue that this must be resolved by introducing a \emph{reifier} node ($\newnode{y}$ in transformation $\varInfomorphism_\text{within-e}$) that can reference nodes from $R_{\mathit{GroupNo}}$ using an edge.
    Furthermore, since we consider elements of $R_{\mathit{GroupNo}}$ to be nodes, they will need an identifier: 
    in  $\varInfomorphism_\text{within-e}$ this identifier is constructed based on the labels and the property keys.
    Thus, we create an identifier based on the label \textsf{inGroupWith} and property key \textsf{groupNo}.
    Additionally, a suitable label $l$ is constructed for this node based on the same label and property. 
    \changed{The result after executing} this transformation is shown in \Cref{fig:transformed_sketch}.
    Notably, while in relational decomposition $R_\mathit{GroupNo}$ the group number is a key, the transformed pattern represents this by adding an explicit key constraint, specifically for nodes with label $l$ where the group number acts as a key for such nodes: $(\pobject{x}{l}{\{\mathsf{groupNo},\mathsf{name}\}})::\mathsf{groupNo}\determ x$.
\end{example}

We generalize the arguments from Example~\ref{ex:redtrans} in Table~\ref{tab:trnasform-justifications}.
Essentially, for each entry in \Cref{tab:transformation-patterns}, we look at the relation schema $R$ induced by the redundancy pattern.
\changed{In particular,} we are interested in the number of attributes that reference graph objects. 
We call such attributes \emph{referencing attributes}.
Indeed, when the pattern is a node, there will be exactly one such attribute.
Otherwise, if the pattern has an edge,
there are three: the source, target, and edge identity. 
Given a descriptor $X\determ Y$ (indicated by the gray arrows in Table~\ref{tab:transformation-patterns}), 
we decompose $R$ into two new relation schemas $R_1$ and $R_2$, with $\attr(R_1) = \attr(R) - Y$ and $\attr(R_2) = X\cup Y$. 
Depending on the number of referencing attributes,
we interpret the effect of the transformation\changed{, as described in the column ``transformation rationale'' of \Cref{tab:trnasform-justifications}}.

\changed{The redundancy removing transformations are both lossless and dependency preserving. Table~\ref{tab:trnasform-justifications} shows the equivalent relational decomposition for which we know these properties hold. Indeed, because the join of the decomposed relations schemas $R_1$ and $R_2$ results again in the relation schema of the redundancy pattern, we can say that the transformations are lossless. Furthermore, Table~\ref{tab:transformation-patterns} specifies exactly in what way the original dependency is preserved (indicated by the grey arrow in the ``Resolved redundancy pattern'' column). Specifically, the left-hand side of the descriptor of the original dependency now acts as a key of the newly created nodes. To reconstruct this original dependency from the transformed graph, the (structurally implied) Object Identity reasoning rule needs to be applied.}

Expressing the elimination of a redundancy pattern as a graph transformation provides insight into the elimination process itself.
When new nodes are created in this process, they may need to be merged together. 
Recall, that the transformations introduce new identifiers based on existing values on the graph data. 
Indeed, we want these new identities to overlap in a controlled manner: this expresses the merging of nodes that express the same information.

\begin{example}
    \Cref{fig:transformed_sketch} shows the subgraph \changed{of \Cref{fig:red-in-graphs}} in which redundancy has been eliminated. The yellow highlight shows the application of the transformation on edge \texttt{e4}, and the blue highlight the transformation of \texttt{e5}. Since the identity of the newly created node \texttt{o1} (representing the fact that group 1 is called ``Heroes'') is the same for edge \texttt{e4} and \texttt{e5}, the reification nodes both point to \texttt{o1}.
\end{example}

\subsection{Graph-Native Normal Forms}
While each individual redundancy pattern can be resolved with its associated transformation,
our goal is to mitigate all redundancy patterns described above. In this section, we define exactly under which conditions these redundancy patterns are mitigated. 
We call these the \acfp{gnnf}. 

In order to define the \acp{gnnf}, we leverage the ideas of relational normal forms (cf. \Cref{tab:rel-norm-form-characteristics}, \cite{alice,saake_datenbanken_2018}), and
\ac{er} diagram normalization (cf. \Cref{sec:related-work}, \cite{thalheim_entity_2000}). Recall the characteristics from the relational normal forms (cf. Table~\ref{tab:rel-norm-form-characteristics}). We adopt the definitions of the first two graph-native normal forms, GN-1NF and GN-2NF, directly from their relational counterparts. We formulate them as:
\begin{definition}[GN-1NF]
    \changed{A graph $\varGraph$ is in \emph{GN-1NF}} if all property values are atomic, i.e., $\forall c_1,c_2\in \infSetPropValues: c_1\notin c_2$. 
\end{definition}
Indeed, GN-1NF is independent of the dependencies that are present. While -- just as in many relational database systems -- non-atomic values are supported in graph database systems, e.g., in form of a list datatype, these are typically not considered in classical approaches to normalization. Since our definition of \acp{lpg} does not consider non-atomic values, all \acp{lpg} satisfy GN-1NF by definition. 

\begin{definition}[GN-2NF]
\changed{A graph $G$ is in \emph{GN-2NF} with respect to a GN-Schema $\Sigma$} if every property mentioned in the schema is fully dependent on a key.
\end{definition}
Recall the structurally implied axioms (cf. Section~\ref{sec:dependency}). Indeed, in any graph, for every property $k$ of graph object $o$, there is a trivial key: the object $o$ itself. Therefore, we conclude that every GN-Schema vacuously satisfies GN-2NF.

To define GN-3NF and GN-BCNF, we first introduce $Q$-GN-3NF and $Q$-GN-BCNF. We call these \emph{scoped} normal forms, in line with our \ac{gnfd} terminology, since they are defined relative to a graph pattern $Q$ (the scope).

\begin{definition}[$Q$-GN-3NF]
 \changed{A graph $G$ is in \emph{$Q$-GN-3NF}} with respect to a GN-Schema $\Sigma$ if whenever $\Sigma\models Q::X\determ Y$, then either: (1) $X$ is a superkey for $\attr(Q)$, (2) $Y$ is prime, i.e., it is a subset of a \changed{key}, or (3) $Q::X\determ Y$ is trivial. 
\end{definition}

Similarly, we define $Q$-GN-BCNF as follows:
\begin{definition}[$Q$-GN-BCNF]
A graph $G$ is in \emph{$Q$-GN-BCNF} with respect to a GN-Schema $\Sigma$ if whenever $\Sigma\models Q::X\determ Y$, then $X$ is a superkey for $\attr(Q)$, or $Q::X\determ Y$ is trivial. 
\end{definition}

These are direct adaptations of the well-known relational normal forms to the setting with \acp{lpg} and \acp{gnfd}. Furthermore, they generalize the $L{:}P$ normal forms of~\cite{skavantzos_third_2025}, %
\changed{which use node patterns for $Q$ and strictly within-node dependencies in $\Sigma$.}

Ultimately, we want to define a normal form for a GN-Schema without specifying a particular scope $Q$. 
In particular, \changed{we want a graph with a GN-Schema $\setOfDeps$} to be in $Q$-GN-3NF or $Q$-GN-BCNF for every $Q$ \changed{present in the dependencies of $\setOfDeps$.}
\begin{definition}[GN-3NF and GN-BCNF] \label{def:gnnf}
\changed{A graph $G$ is in GN-BCNF (resp. 3NF)} with respect to a GN-Schema $\setOfDeps$ if, for every $Q$ \changed{present in the dependencies of $\setOfDeps$}, it is in $Q$-GN-BCNF (resp. 3NF).
\end{definition}
Transforming a graph $G$ that satisfies a schema $\Sigma$ into a graph~$G'$ s.t. it satisfies $\Sigma'$, which is in GN-BCNF (or 3NF), is the process of \emph{graph-native normalization}. We say a graph (schema) is in \acf{gnnf} if it is either in GN-3NF or GN-BCNF.

\section{Graph-Native Normalization}
\label{sec:normalization}
Considering \acp{gnfd} and \acp{gnnf}, 
we are ready to define a graph-native approach for the normalization of \acp{lpg}.
This normalization process considers not only within-node but also within-edge and, generally, between-graph-object dependencies. 
\changed{While Section~\ref{sec:native-property-graph-normal-forms} addressed each redundancy pattern individually, this section eliminates them collectively for a given schema, transforming any graph into \ac{gnnf}. We first define normalization with respect to a single scope $Q$ (\Cref{alg:gn-norm}), and then extend it to full normalization (\Cref{alg:gn-norm-full}).}
\vspace{-0.5em}
\changed{\subsection{Scoped Normalization}}

We will start with the normalization process to transform a given graph $\varGraph$ and GN-schema (i.e., a set of dependencies) such that, 
for a given scope $Q$, it satisfies $Q$-GN-3NF, and may satisfy $Q$-GN-BCNF.

The process is inspired by the \emph{synthesis} algorithm of relational normalization~\cite{alice}, \changed{which always considers all dependencies. An alternative would be the decomposition algorithm from relational normalization~\cite{alice}, which guarantees \ac{bcnf} but is not always dependency-preserving and could thus leave redundancies in normalized graphs. The process consists of four phases, shown in \Cref{alg:gn-norm}.}

\edbtarxiv{

\begin{algorithm}[tbp]
\small
    \SetKwFunction{MinCov}{MinCov}
    \SetKwFunction{List}{List}
    \SetKwFunction{CreateDependencies}{CreateDependencies}
    \SetKwFunction{Transform}{Transform}
    \SetKwInOut{Input}{input}\SetKwInOut{Output}{output}
    \SetKwComment{Comment}{\# }{}

    \textbf{Input:} a scope $Q$; a schema $\setOfDeps$; and an \ac{lpg} $\varGraph$

    \medskip

    \textcolor{gray}{\textit{Phase 1: Determine the dependencies relevant to $Q$}}\\    
    $\setOfDeps_\varQueryPattern := \{\varQueryPattern::D \mid \text{$\varQueryPattern'::D\in\Sigma$ and $\moregeneralthan{\varQueryPattern'}{\varQueryPattern}$}\}$

    \medskip
    \textcolor{gray}{\textit{Phase 2: Compute the minimal cover}}\\
    $\setOfDeps_\varQueryPattern^\mathsf{min} := \MinCov(\Sigma_\varQueryPattern)$

    \medskip
    \textcolor{gray}{\textit{Phase 3: Construct transformations and new dependencies}}\\
    $\setOfTransfs := \{(\varQueryPattern,T_\mathit{new})\mid \text{$(T,T_\mathit{new})\in\setOfRedTransfs$ s.t. $T$ matches $\varDep\in\Sigma_\varQueryPattern^\mathsf{min}$}\}$
    
    $\setOfDeps_\varQueryPattern^\mathsf{new} := \CreateDependencies(\setOfTransfs)$

    \medskip
    \textcolor{gray}{\textit{Phase 4: Apply graph transformations}}\\

    $\varGraph := \Transform(\varGraph,\setOfTransfs)$

    \medskip
    \textbf{Result: } $\setOfDeps_\varQueryPattern^\mathsf{new}$ is the schema and $\varGraph$ is normalized w.r.t. $\varQueryPattern$ 

\caption{The transformation-based approach to graph-native normalization based on the synthesis algorithm}
\label{alg:gn-norm}
\end{algorithm}%

}{
\begin{algorithm}[tbp]
    \SetKwFunction{MinCov}{MinCov}
    \SetKwFunction{List}{List}
    \SetKwFunction{CreateDependencies}{CreateDependencies}
    \SetKwFunction{Transform}{Transform}
    \SetKwInOut{Input}{input}\SetKwInOut{Output}{output}
    \SetKwComment{Comment}{\# }{}

    \textbf{Input:} a scope $Q$; a GN-schema $\setOfDeps$; and an \ac{lpg} $\varGraph$

    \medskip

    \textcolor{gray}{\emph{Phase 1: Determine the set of dependencies $\setOfDeps_\varQueryPattern$ relevant to $Q$}}\\    
    $\setOfDeps_\varQueryPattern := \emptyset$

    \For{$\varphi: Q' ::X\determ Y \in\Sigma$}{
    \If{$\moregeneralthan{Q'}{Q}$}{
        $\setOfDeps_\varQueryPattern := \setOfDeps_\varQueryPattern \cup \{\varphi\}$
        }
    }
    
    \medskip
    \textcolor{gray}{\emph{Phase 2: Apply \emph{Restriction} and compute the minimal cover}}\\
    $\setOfDeps_\varQueryPattern^R  := \emptyset$

    \For{$\varphi: Q' ::X\determ Y \in\Sigma_Q$}{
            $\setOfDeps_\varQueryPattern^R := \setOfDeps_\varQueryPattern^R  \cup \{Q::X\determ Y\}$
    }

    $\setOfDeps_\varQueryPattern^\mathsf{min} := \emptyset$
    
    \For{$\varphi: Q' ::X\determ Y \in\Sigma_Q^R$}{
               \If{$\varDep$ is minimal and not redundant}{
        $\setOfDeps_\varQueryPattern^\mathsf{min} := \setOfDeps_\varQueryPattern^\mathsf{min} \cup \{\varphi\}$
        }
    }

    \medskip
    \textcolor{gray}{\emph{Phase 3: Construct transformations and new dependencies}}\\
    $\setOfTransfs := \emptyset$\\
    $\setOfDeps_\varQueryPattern^\mathsf{new} := \emptyset$

    \For{$\varphi: Q ::X\determ YZ \in \setOfDeps_\varQueryPattern^\mathsf{min}$}{
        \For{$Y' \subseteq YZ$}{
            \For{$(T, T_\textsf{new}) \in \setOfRedTransfs$}{
                \If{$X\determ Y' ~\hat{=}~ \psi$}{
                    $\setOfTransfs := \setOfTransfs \cup \big\{(Q, T_\textsf{new})\big\}$
                    
                    \If{$T \neq \psi_\text{between-n-ep}$}{
                        $\setOfDeps_\varQueryPattern^\mathsf{new} := \setOfDeps_\varQueryPattern^\mathsf{new} \cup T_\textsf{new} :: X^T \determ n$
                    }
                }
            }
        }
    }
    
    \medskip
    \textcolor{gray}{\emph{Phase 4: Apply graph transformations}}\\
    \For{$\psi \in \setOfTransfs$}{
        \texttt{create}$(\psi,G)$ 
        
        \textcolor{gray}{\emph{\% Create new graph objects and set properties as defined by $\psi$}}
    }
    \For{$\psi \in \setOfTransfs$}{
        \texttt{remove}$(\psi,G)$ 
        
        \textcolor{gray}{\emph{\% Remove properties and graph objects as defined by $\psi$}}
    }

    \medskip
    \textbf{Result: } $\setOfDeps_\varQueryPattern^\mathsf{new}$ is the schema and $\varGraph$ is normalized w.r.t. $\varQueryPattern$ 

\caption{\textsc{ScopedLPGNormalize}\\The transformation-based approach to graph-native normalization based on the synthesis algorithm}
\label{alg:gn-norm}
\end{algorithm}
}%

\paragraph{Phase 1: Determining Applicable Dependencies}
\changed{Since we normalize with respect to a particular scope $Q$, we need to determine
which dependencies in $\Sigma$ apply to that scope so that we can properly compute the minimal cover in the next phase. 
We know that a graph $\varGraph$ satisfies a GN-Schema $\setOfDeps$, i.e., $\varGraph\sat \Sigma$. 
Hence,
by the definition of \acp{gnfd}, the evaluation of the scope pattern $Q$ in graph $G$, i.e., $\sem QG$, satisfies every functional dependency $X\determ Y$ of a \ac{gnfd} $Q::X\determ Y$ that is present in $\Sigma$.
However, these are not the only dependencies $\sem QG$ satisfies. 
We also need to consider every \ac{gnfd} whose scope $Q'$ is more general than $Q$, i.e., $Q' \sqsupseteq Q$. 
Intuitively, this is because a more general scope $Q'$ also ``covers'' the matched subgraph patterns of the more specific scope $Q$. Formally, this is an application of 
the \emph{Restriction} reasoning rule (cf. \Cref{sec:dependency}). Recall that \Cref{lem:moregeneralthan} gives us a straightforward way to determine whether $Q'$ is more general than $Q$.
\edbtarxiv{}{Thus, the set of applicable dependencies $\Sigma_Q$ is defined as $\setOfDeps_\varQueryPattern := \{\varQueryPattern::X\determ Y \mid \text{$\varQueryPattern'::X \determ Y\in\Sigma$ and $\moregeneralthan{\varQueryPattern'}{\varQueryPattern}$}\}$. 
}}

\paragraph{Phase 2: Minimal Cover}
We would like to have a minimal cover $\setOfDeps_Q^\mathsf{min}$ of the set of all dependencies that are relevant for the scope~$Q$. 
As in relational normalization, this minimal cover is computed to avoid the consideration of redundant or non-minimal \acp{gnfd} for the determination of transformations $\setOfTransfs$ in Phase 3. 
We say a \ac{gnfd} $\varDep: Q::X\determ Z$ is (i) redundant if a dependency $\varDep$ is implied by all other dependencies of a GN-schema $\setOfDeps$, i.e., $ \setOfDeps -  \{\varDep\} \models\varDep$, 
and is (ii)~minimal if there is no $Y\subset X$ such that $\setOfDeps - \{\varDep\} \models Q::Y\determ Z$. 
We use the set of applicable dependencies from Phase~1~($\Sigma_Q$) to compute~$\setOfDeps_Q^\mathsf{min}$. 
\changed{After applying \emph{Restriction} (cf. \Cref{sec:dependency}) to each dependency in $\Sigma_Q$, we can reason over them using Armstrong's Axioms and compute the minimal cover as for \acp{fd}~\cite{alice}.}

\paragraph{Phase 3: Constructing Transformations}
The core of the graph-native normalization approach are the redundancy-removing transformations
from \Cref{sec:native-property-graph-normal-forms}. 
\edbtarxiv{\changed{We describe the transformation construction in detail in the full version~\cite{schrott2026graphnative}.}}{
\changed{To construct the set of transformations $\setOfTransfs$ needed for normalization, 
as the first step,
every dependency $\varphi: X\determ YZ \in \Sigma_Q^{\textsf{min}}$ with multiple variables on its right-hand side needs to be split up by using the \textit{Decomposition} rule (cf. \Cref{sec:dependency}) such that only one variable occurs on the right-hand side of $\varphi$, i.e., in~$Y$. 
As next step, based on the property- and object-identity variables (cf. \Cref{sec:preliminaries}) used in the descriptor $X\determ Y$, 
every \ac{gnfd} $\varDep : Q:: X\determ Y$ is matched with a redundancy pattern $T$ from a transformation $\psi\in \mathcal{T}_\textsf{red}$ shown in \Cref{tab:transformation-patterns}. For example, to match with the transformation $\varInfomorphism_\text{between-np-ep}$, $X$ must consist of a property variable of a node and $Y$ is a property variable of an edge.
In \Cref{alg:gn-norm}, we use the expression $X\determ Y~\hat =~\psi$ to denote the matching.
Following that, every matched transformation~$\varInfomorphism:(T,T_\textsf{new})$ is 
adapted to match the pattern $Q$, 
and then added to the set of transformations~$\setOfTransfs$ that will be applied in Phase 4.
As shown in \Cref{tab:transformation-patterns} and described in \Cref{sec:native-property-graph-normal-forms}, for transformations that involve the creation of a new node $n$ (all except $\varInfomorphism_\text{between-n-ep}$), 
a key for $n$ is defined based on the transformed left-hand side~$X^T$ of the original dependency~$\varphi$. 
For $\varInfomorphism_\text{between-n-ep}$ no new dependency needs to be added explicitly, since it is already structurally implied.
}
}

\paragraph{Phase 4: Applying Transformations}
Finally, the graph is transformed according to the constructed transformations $\mathcal{T}$.
\changed{Recall from \Cref{sec:preliminaries} that a set of transformations is applied in two phases: new graph data is created first, then superfluous data is removed. Since each individual transformation is lossless and dependency-preserving, so is the combined application of $\mathcal{T}$.}

\begin{example}\label{ex:algo}
\changed{
We show the application of \textsc{Scoped\-LPGNormalize} (\Cref{alg:gn-norm}) and highlight the importance of Phases 1 and 2.
For this we consider a GN-Schema $\Sigma$, which relies on two scopes that use $P=\{\textsf{title},\textsf{program},\textsf{language}\}$:
\begin{itemize}
\small
    \item[$Q_1$:] $(\pobject x{\{\textsf{Course}\}}{P})$ 
    \item[$Q_2$:] $(\pobject x{\{\textsf{Course},\textsf{International}\}}{P})\pedgeRight{\pobject {y}{\{\textsf{teaches}\}}{\{\textsf{usingBook}\}}}()$
\end{itemize}
The considered GN-Schema $\Sigma$ contains the dependencies: %
}
{\small
\begin{align*}
& \varphi_1{:}\;Q_1:: x.\textsf{title}\determ x.\textsf{program};\; \varphi_2{:}\;Q_1:: x.\textsf{title}\determ x.\textsf{language}; \\
& \varphi_3{:}\;Q_2:: x.\textsf{program}\determ x.\textsf{language};\; \varphi_4{:}\;Q_2::x.\textsf{language}\determ y.\textsf{usingBook}\\
& \text{and $\varphi_\text{Course}$ (cf. Example~\ref{ex:formal-deps})}
\end{align*}
}\\[-1.5em]
    \changed{For this application of \Cref{alg:gn-norm}, we decide to perform a scoped normalization regarding scope $Q_2$ and schema $\setOfDeps$ on graph a~$G$. }

    \changed{
    The \emph{first phase} loops over every \ac{gnfd} in $\Sigma$. Since $Q_1$ is more general than $Q_2$ (cf.~\Cref{tab:moregeneralthan}), all \acp{gnfd} that use $Q_1$ are  added to $\Sigma_Q$. 
    Intuitively, we can see that dependencies with scope $Q_1$ are applicable to $Q_2$:
    international courses are also courses, and if the necessary property keys are present, the (more general) course dependencies also apply to international courses.
    Trivially, $Q_2$ is more general than itself, so also every \ac{gnfd} that uses $Q_2$ is added to $\Sigma_Q$. 
    However, dependency~\hyperref[dep:course]{$\varDep_\text{Course}$} is not applicable for scope $Q_2$ since it requires the presence of the property keys $\{\textsf{title}, \textsf{year}\}$, which is not a subset of $P$ due to \textsf{year}.}

    \changed{
    Having collected all \acp{gnfd} that hold for scope $Q_2$,
    \emph{Phase 2} then computes the minimal cover of these dependencies. Even though $\varphi_1$, $\varphi_2$, $\varphi_3$, and $\varphi_4$ are present in $\Sigma_Q$, the minimal cover $\setOfDeps_Q^\mathsf{min}$ will not include $\varphi_2$, as it can be implied by the other three dependencies.}

\changed{
    In \emph{Phase 3} the graph transformations are determined: $\psi_\text{within-n}$ for both $\varphi_1$ and $\varphi_3$, and $\psi_\text{between-np-ep}$ for $\varphi_4$.    
    Indeed, if we had kept $\varphi_2$ in Phase 3, the transformation resulting from that dependency would lead to the creation of a superfluous new node.
    Thus, Phase 2 prevents the creation of such redundant nodes.
    }

\changed{Finally, in \emph{Phase 4}, the transformations are applied to $G$.}
\end{example}
\vspace*{-2em}

\changed{\subsection{Full Normalization}}
\begin{algorithm}[tbp]
\small
    \SetKwFunction{TopologicalSort}{TopologicalSort}
    \SetKwFunction{ScopedNormalize}{ScopedNormalize}
    \SetKwFunction{CreateDependencies}{CreateDependencies}
    \SetKwFunction{Transform}{Transform}
    \SetKwInOut{Input}{input}\SetKwInOut{Output}{output}
    \SetKwComment{Comment}{\# }{}

    \textbf{Input:} a schema $\setOfDeps$; and an \ac{lpg} $\varGraph$
    \medskip

    \textcolor{gray}{\textit{Collect all scopes}}\\    
    $\mathcal{Q} := \{Q \mid Q::X\determ Y\in\Sigma\}$
    
    \medskip
    \textcolor{gray}{\textit{Topologically sort the scopes based on $\sqsubseteq$ (see  \Cref{lem:moregeneralthan})}}\\
    $\mathcal{Q}_\textsf{sorted} := \texttt{Sort}_{\sqsubseteq}(\mathcal{Q})$

    \medskip
    \textcolor{gray}{\textit{Normalize for each scope}}\\
    \For{\text{scope} $Q\in\mathcal{Q}_\textsf{sorted}$}{
    $\textsc{ScopedLPGNormalize}(Q, G,\setOfDeps)$     \textcolor{gray}{\textit{\% see \Cref{alg:gn-norm}}}\\    

     }
    
    \medskip
    \textbf{Result: } $G$ is in \ac{gnnf}

\caption{\textsc{FullLPGNormalize}\\Full normalization algorithm for transforming \acp{lpg} into \ac{gnnf}.}
\label{alg:gn-norm-full}
\end{algorithm}

\noindent
\changed{\Cref{alg:gn-norm} normalizes with respect to one scope, but \ac{gnnf} (Definition~\ref{def:gnnf}) requires normalization for every scope $Q$. \emph{Full} graph-native normalization (\Cref{alg:gn-norm-full}, \textsc{FullLPGNormalize}) achieves this by repeatedly applying \textsc{ScopedLPGNormalize} (\Cref{alg:gn-norm}) to each scope. The order of application is crucial: scopes must be processed from least to most general to avoid side-effects. We therefore topologically sort the scopes in $\setOfDeps$ using the ``more general than'' partial order ($\sqsupseteq$, cf. \Cref{lem:moregeneralthan}). The following example illustrates why more specific scopes must be considered first:}

\begin{example}\label{ex:order}
    Consider the situation from Example~\ref{ex:algo}, where we normalized within the scope $Q_2$ only.
    \changed{We now normalize for all scopes and show the importance of the order of scopes for full normalization.
    Assume we consider the scope $Q_1$ first. 
    This would transform every course according to the transformation rules. For example, by $\varphi_1$, the properties $\textsf{title}$ and $\textsf{program}$ would be moved from the course nodes to a new node due to transformation rule $\varInfomorphism_{\text{within-n}}$. 
    Specifically, this means that the \emph{international} courses would also be affected by this transformation.
    However, as a side effect, this movement of properties makes the \acp{gnfd} $\varphi_3$ and $\varphi_4$ no longer applicable. 
    In particular, in this example, if we want to normalize with scope $Q_2$ after $Q_1$, transformations resulting from $\varphi_3$ and $\varphi_4$ would have no effect.
    }

\changed{
    This is why the topological ordering is crucial.
    Executing normalization with the more specific scope~$Q_2$ first avoids this situation. Then, in the beginning, only international courses will be modified, considering the dependencies relating to regular courses as well (due to the application of Restriction in  Phase 1 of \textsc{ScopedLPGNormalize}). 
    Afterward, normalization with the more general scope~$Q_1$ only affects course nodes, as international courses have already been normalized and no longer match the original dependencies.}

\changed{
    The result is a graph that is normalized with respect to $\setOfDeps$ and every scope $Q$. Hence $G$ is in \ac{gnnf}.}
\end{example}

\section{Evaluation}
\label{sec:evaluation}
\label{sec:experiments}

\paragraph{Technical Setup}
We implemented our approach in Python utilizing an ANTLR grammar~\cite{antlr-website} to parse \acp{gnfd} defined on graph patterns; 
it supports graphs stored in Neo4j and Memgraph. 
Our implementation and evaluation setup incl. Docker, artifacts, etc.
are available on GitHub\footnote{\label{github}\url{https://github.com/dmki-tuwien/lpg-normalization}}.
All experiments were run on a server with a 24-core \SI{2.9}{\giga\hertz} AMD EPYC 9254 CPU and \SI{768}{\giga\byte} DDR5 RAM.
\changed{Neo4j was configured to use a page cache of \SI{700}{\giga\byte} and a heap of \SI{31}{\giga\byte}, as recommended by Neo4j.}

\paragraph{Metrics}
An overview of our metrics and their definitions are provided in Table~\ref{tab:metric-definitions}. 
\changed{We distinguish between metrics measuring the structure and size of an \ac{lpg} and metrics that measure query performance.}
Regarding graph structure and size, on %
a per-graph basis,
we measure the number of nodes ($\metric_{\#\text{Node}}$) and edges ($\metric_{\#\text{Edge}}$), 
as well as the average number of properties per node ($\metric_{\text{AvgPropNode}}$) and edge ($\metric_{\text{AvgPropEdge}}$).
For a set of \acp{gnfd} $\setOfDeps$, i.e., a GN-Schema, 
we measure the number of all \acp{gnfd} ($\metric_{\#\text{Dep.sTotal}}$) as well as within-node \acp{gnfd} ($\metric_{\#\text{Dep.sWithinNode}}$), 
within-edge \acp{gnfd} ($\metric_{\#\text{Dep.sWithinEdge}}$),
and between-\acp{gnfd} ($\metric_{\#\text{Dep.sBetween}}$).

In line with \cite{skavantzos_third_2025},
we also calculate the `potentials for redundancy' for each \ac{gnfd}.
For a given \ac{gnfd} $Q::X\determ Y$, %
and a graph $G$,
the list of redundancy potentials $M$ is calculated by a group-by-count query $M := \mathcal{G}_{\mathit{count}(X\cup Y)}(\sem QG)$, 
effectively counting how often \changed{the same combinations of values for property keys in $X$ and $Y$ occur}.
\changed{The higher the redundancy potentials (i.e., the counts in $M$) are, the more redundancy is present per value combination. If there are no redundant combinations of values, every value in $M$ is $1$.}
\changed{Based on $M$ we calculate $m_\mathit{MaxRed.Pot.}$ as the maximum redundancy potential per \ac{gnfd}, and $m_\mathit{AvgRed.Pot.}$ as the arithmetic mean of the redundancy potentials $M$ per \ac{gnfd}. Due to being based on counts, these two metrics do not have an upper limit.}

Another %
\changed{structural }%
metric we define per \ac{gnfd} is minimality $m_\text{Minim.}$~\cite{hacid_novel_2019}.
\changed{It complements the redundancy-potential-based metrics, as it
\edbtarxiv{is}{calculates} the ratio between the number of distinct and the total number of all combinations
of values for property keys in $X$ and $Y$.
Thus, its result values lie between 0 (completely redundant, there is only a single combination of values for property keys in $X$ and $Y$) and 1 (no redundancy at all, ``minimal'' in the sense that values for property keys in $X$ and $Y$ only appear in unique combinations).}
\edbtarxiv{}{An example calculation of the per-dependency metrics is shown below in Example~\ref{ex:per-dep-metrics}.}

\edbtarxiv{
}{
\begin{table*}
    \centering
    \caption{Definitions of evaluation metrics}
    \label{tab:metric-definitions}
    \vspace*{-1em}
    \footnotesize
        \begin{tabular}{r r l L{2.5cm} l}
    \toprule
&    \textbf{Type} & \textbf{Ref.} & \textbf{Name} & \textbf{Definition} \\
    \midrule
        \multirow{11}{*}{\rotatebox{90}{\parbox{3.6cm}{\textbf{Graph size and strcuture}}}} &
        \multirow{4}{*}{\rotatebox{90}{\parbox{1.3cm}{\centering\textit{Per graph\newline$\varGraph$}}}}& $\metric_{\#\text{Node}}$ & NodeCount & $\left|\setNodeIDs\right|$\vspace{0.1em} \\
        && $\metric_{\#\text{Edge}}$ & EdgeCount & $\left|\setDirectedEdgeIDs\right|$ \\
        && $\metric_{\text{AvgPropNode}}$ & AvgNode\-PropCount & 
        $\funcAvg\Bigl(
            \Bigl\{
                \bigl|
                    \left\{\varPropKey\in\infSetPropKeys:(\varNode,\varPropKey) \in \funcDomain(\funcProperty)\right\}
                \bigr|:\varNode\in\setNodeIDs
            \Bigr\}
        \Bigr)$ \\
        && $\metric_{\text{AvgPropEdge}}$ 
        & AvgEdge\-PropCount & 
        $\funcAvg\Bigl(
            \Bigl\{
                \bigl|
                    \left\{\varPropKey\in\infSetPropKeys:(\varEdge,\varPropKey) \in \funcDomain(\funcProperty)\right\}
                \bigr|:\varEdge\in\setDirectedEdgeIDs
            \Bigr\}
        \Bigr)$  \\

    \cmidrule{2-5}
    
            &\multirow{4}{*}{\rotatebox{90}{\parbox{1.3cm}{\raggedleft \textit{Per}\newline\textit{GN-Schema}\\ \centering$\setOfDeps$}}}   &$\metric_{\#\text{Dep.sTotal}}$ & AllDep.sCount & $\left|\varDep\in\setOfDeps\right|$ \vspace{0.1em}\\
            && $\metric_{\#\text{Dep.sWithinNode}}$ & WithinNodeDep.sCount & $\left|\varDep\in\setOfDeps \text{ where $\varphi$ is a within-node dependency}\right|$ \vspace{0.1em}\\
            && $\metric_{\#\text{Dep.sWithinEdge}}$ & WithinEdgeDep.sCount & $\left|\varDep\in\setOfDeps \text{ where $\varphi$ is a within-edge dependency}\right|$ \vspace{0.1em}\\
            && $\metric_{\#\text{Dep.sBetween}}$ & BetweenG.O.Dep.sCount & $\left|\varDep\in\setOfDeps \text{ where $\varphi$ is a between-graph-object dependency}\right|$ \\
                
    \cmidrule{2-5}
        &\multirow{4}{*}{\rotatebox{90}{\parbox{1.5cm}{\textit{Per \ac{gnfd} $\varDep$}}}}  & $\metric_{\text{MaxRed.Pot.}}$ & Max\-Redundancy\-Count (cf. metric ``\texttt{\#red}'' in \cite{skavantzos_third_2025}) &
        $\funcMax\bigl(\mathcal{G}_{\mathit{count}(D)}(\sem QG)
        \bigr)$\\
        && $\metric_{\text{AvgRed.Pot.}}$ & Avg\-Redundancy\-Count &
        $\funcAvg\bigl(\mathcal{G}_{\mathit{count}(D)}(\sem QG)
        \bigr)$\vspace{0.1em}\\
    
        && $\metric_{\text{Minim}}$ & Minimality\newline (based on \cite{hacid_novel_2019}) &
        ${\begin{cases}1, & \text{if } \bigl|\sem QG\bigr|=1\\
        {\dfrac{  \bigl|
                \mathcal{G}_{\mathit{count}(D)}(\sem QG)
            \bigr|-1}
            { \bigl|\sem QG\bigr|-1}}, & \text{else}\end{cases}}$
        \\
        \midrule
    \multicolumn{2}{c}{\multirow{2}{*}{\rotatebox{90}{\parbox{0.8cm}{\centering \textbf{Query}\newline \textbf{per\-for\-mance}\newline\newline}}}}
     &  $m_\text{Q.Dur.}$ & Query duration (execution time) & As returned by the graph system \vspace{0.2em}\\
     & & $m_\text{Q.DBHit}$ & Database hits & As returned by the graph system \vspace{0.3em}\\
        
    \bottomrule
    \end{tabular}
\end{table*}
}

\begin{table}
\vspace*{-0.8em}
    \caption{Graphs used in the evaluation}
    \label{tab:eval_graphs}
    \vspace*{-1em}
    \centering
    \footnotesize
    \begin{tabular}{L{1.5cm} l @{}r@{\enspace}r@{\enspace}r@{\enspace}r}
    \toprule
       \textbf{Graph} & \textbf{Source} & \textbf{$\metric_{\#\text{Node}}$} & \textbf{$\metric_{\#\text{Edge}}$} & \textbf{$\metric_{\text{AvgPropNode}}$} & \textbf{$\metric_{\text{AvgPropEdge}}$}\\
       \midrule
       
       London Public Transport & \cite{schrott_2026_18479732} & \num{651} & \num{1952} & $\approx 5$ & \num{3}\vspace{0.1em}\\
       Northwind & \cite{northwind_database} & 1035 & 3139 & $\approx 13.30$ & $\approx 3.43$ \vspace{0.1em}\\
       No Socks & \cite{skavantzos_third_2025} & \num{3} & \num{2} & $\approx 3.33$ & \num{0}\vspace{0.1em}\\
       Offshore & \cite{offshore_database}  & \num{2016523} & \num{3339267} & $\approx 9.86$ & $\approx 3.25$\vspace{0.1em}\\
       Train Services & \cite{schrott_2026_18478422} & \num{194179} & \num{1826601} & $\approx 6.98$ & $\approx 3.79$\vspace{0.1em}\\
       University & \Cref{fig:lecture-uses-book-denorm} & \num{4} &\num{3} & \num{1.25} & \num{2}\\
       \bottomrule
\end{tabular}

\end{table}

\begin{table*}
    \centering    
    \caption{Overview of the evaluation scenarios. Durations are in \si{\second}, are rounded to two decimal places.%
    }
    \label{tab:experiment_datasets}
    \vspace*{-1em}
    \scriptsize
    \footnotesize
\begin{tabular}{@{}l@{~}L{1.5cm}@{}r@{~}r@{~}r@{~}r@{~}r@{~}r @{~} L{2.4cm} @{} L{1.6cm} @{} r @{\enspace} r @{} }
\toprule
\multicolumn{2}{l}{\textbf{\emph{Scenario}}}& \multicolumn{3}{c}{\textbf{\emph{Provided \acp{gnfd}}}}  & \multicolumn{3}{c}{\textbf{\emph{Minimal cover of \acp{gnfd}}}} & \textbf{Applicable} &  \multicolumn{3}{r}{\textbf{\emph{Normalization\,duration}}} \\
\textbf{\acs{id}} & \textbf{Graph} & $\metric_{\#\text{Dep.sW.-Node}}$ & $\metric_{\#\text{Dep.sW.-Edge}}$ & $\metric_{\#\text{Dep.sBetween}}$ & $\metric_{\#\text{Dep.sW.-Node}}$ & $\metric_{\#\text{Dep.sW.-Edge}}$ & $\metric_{\#\text{Dep.sBetween}}$ & \textbf{transformations} & \textbf{Note} & \textbf{Memgraph} & \textbf{Neo4j}\\
\midrule
$S_\text{Lon}$ & London Public Transport & \num{3} & \num{2} & \num{0} &  \num{2} & \num{1} & \num{0} & $\varInfomorphism_{\text{within-n}}, \varInfomorphism_\text{within-e}$ & & \num{0.11} & \num{1.33} \\
$S_\text{Nw}$ & Northwind & 2 & 0 & 0 & 2 & 0 & 0 & $\varInfomorphism_{\text{within-n}}$ & as in \cite{skavantzos_third_2025} & \num{0.04} & \num{1.14} \\
$S_\text{No-1}$ & No Socks & 4 & 0 & 0 & 4 & 0 & 0 & $\varInfomorphism_{\text{within-n}}$ & as in \cite{skavantzos_third_2025} & \num{0.03} & \num{1.17} \\
$S_\text{No-2}$ & No Socks & 4 & 0 & 0 & 4 & 0 & 0 & $\varInfomorphism_{\text{within-n}}$ & as in \cite{skavantzos_third_2025} & \num{0.03} & \num{1.41} \\
$S_\text{No-3}$ & No Socks & 5 & 0 & 0 & 5 & 0 & 0 & $\varInfomorphism_{\text{within-n}}$ & as in \cite{skavantzos_third_2025} & \num{0.02} & \num{1.45} \\
$S_\text{Off-1}$ & Offshore & 3 & 0 & 0 & 2 & 0 & 0 & $\varInfomorphism_{\text{within-n}}$ & as in \cite{skavantzos_third_2025} & \num{9.18} & \num{24.77} \\
$S_\text{Off-2}$ & Offshore & 5 & 0 & 0 & 5 & 0 & 0 & $\varInfomorphism_{\text{within-n}}$ & as in \cite{skavantzos_third_2025} & \num{33.92} & \num{116.28} \\
$S_\text{Off-3}$ & Offshore & 10 & 0 & 0 & 5 & 0 & 0 & $\varInfomorphism_{\text{within-n}}$ & as in \cite{skavantzos_third_2025} & \num{25.15} & \num{63.26} \vspace{0.5em}\\
$S_\text{Ts-1}$ & Train Services & 3 & 1 & 2 & 2 & 1 & 2 & $\varInfomorphism_{\text{w.-n}}, \varInfomorphism_{\text{w.-e}},$ $\varInfomorphism_{\text{b.-np-ep}},$ $ \varInfomorphism_{\text{b.-ep-n}}$ & & \num{60.89} & \num{130.27} \vspace{0.5em}\\
$S_{\text{Ts-2}}$ & Train Services & 3 & 1 & 3 & 2 & 1 & 3 & $\varInfomorphism_{\text{w.-n}}, \varInfomorphism_{\text{w.-e}},$ $\varInfomorphism_{\text{b.-np-ep}},$ $ \varInfomorphism_{\text{b.-ep-n}},$ $\varInfomorphism_{\text{b.-ep-np}}$ & & \num{104.16} & \num{408.72}\vspace{0.5em}\\
$S_{\text{Uni-1}}$ & University & 0 & 0 & 1 & 0 & 0 & 1 & $\varInfomorphism_{\text{between-n-ep}}$  & cf. \Cref{fig:lecture-uses-book-denorm} & \num{0.01} & \num{0.19}\\
 $S_{\text{Uni-2}}$ & University & 0 & 0 & 1 & 0 & 0 & 1 & $\varInfomorphism_{\text{between-np-ep}}$ & adapted from \Cref{fig:lecture-uses-book-denorm} & \num{0.01} & \num{0.61}
 \\

\bottomrule
\end{tabular}

\end{table*}

\edbtarxiv{}{
\begin{example} 
\label{ex:per-dep-metrics}
\changed{For the calculation of the per-dependency metrics, $m_\text{MaxRed.Pot}$, $m_\text{AvgRed.Pot}$, and $m_\text{Minim}$ we consider a graph~$G$, shown in 
\Cref{fig:calculation_example}, and a \ac{gnfd} $$\varphi:(x:\{{X}\}:\{\mathsf{k}_x\})\xrightarrow{y:\{{Y}\}:\{\mathsf{k}_y\}}()::x.\mathsf{k}_x\determ y.\mathsf{k_y}$$ that holds eight times in $G$ (denoted by gray arrows).
With respect to $\varphi$, the following redundancy potentials are present in $G$: $M_\varphi=[2,2,3,1]$. This is the result of $(x.\mathsf{k}_x, y.\mathsf{k}_y)=(\text{m},\text{a})$ occurring twice, $(x.\mathsf{k}_x, y.\mathsf{k}_y)=(\text{n},\text{b})$ occurring twice, $(x.\mathsf{k}_x, y.\mathsf{k}_y)=(\text{o},\text{c})$ occurring three times, and $(x.\mathsf{k}_x, y.\mathsf{k}_y)=(\text{p},\text{d})$ occurring once in $G$.}

\changed{
The metrics are then calculated as follows:
\begin{itemize}
    \item $m_\text{MaxRed.Pot}=\mathsf{max}(M_\varphi)=3$
    \item $m_\text{AvgRed.Pot}={\mathsf{avg}(M_\varphi)} = 2$ 
    \item $m_\text{Minim}={|M_\varphi|-1 \over \mathsf{sum}(M_\varphi)-1}\approx0.43$ 
\end{itemize}
}

\begin{figure}
    \centering
    \begin{tikzpicture}[node distance=1em]
            \node[vertex style=Node,align=center,minimum height=0.4cm,minimum width=0.6cm] (1) {$x_1$:${X}$};

\node[text style,right of=1,align=center,yshift=0.6em,xshift=2em] (1_id) {$y_1$:${Y}$};
\node[text style,right of=1,align=center,yshift=-0.6em,xshift=2em] (1_prop) {$\mathsf{k}_y=$ a};
\node[text style,left of=1,align=right,xshift=-2em] (1_v) {$\mathsf{k}_x=$ m};

            \node[vertex style=Node,align=center,minimum height=0.4cm,minimum width=0.6cm,yshift=5em] (2) {$x_2$:${X}$};

\node[text style,right of=2,align=center,yshift=-1.1em,xshift=1.5em] (2_id) {$y_2$:${Y}$};
\node[text style,right of=2,align=center,yshift=-1.8em,xshift=0em] (2_prop) {$\mathsf{k}_y=$ a};
\node[text style,above of=2,align=center,yshift=0.5em] (2_v) {$\mathsf{k}_x=$ m};

            \node[vertex style=Node,align=center,minimum height=0.4cm,minimum width=0.6cm, xshift=6em,yshift=5em] (3) {$x_3$:${X}$};

\node[text style,below of=3,align=center,yshift=-0.3em,xshift=1em] (3_id) {$y_3$:${Y}$};
\node[text style,below of=3,align=center,yshift=-1.2em,xshift=1.2em] (3_prop) {$\mathsf{k}_y=$ b};
\node[text style,above of=3,align=center,yshift=0.5em] (3_v) {$\mathsf{k}_x=$ n};

            \node[vertex style=Node,align=center,minimum height=0.4cm,minimum width=0.6cm, xshift=12em,yshift=5em] (4) {$x_4$:${X}$};

\node[text style,left of=4,align=center,yshift=-1.1em,xshift=-1.2em] (4_id) {$y_4$:${Y}$};
\node[text style,left of=4,align=left,yshift=-1.8em,xshift=0.3em] (4_prop) {$\mathsf{k}_y=$ b};
\node[text style,above of=4,align=center,yshift=0.5em] (4_v) {$\mathsf{k}_x=$ n};

\node[vertex style=Node,align=center,minimum height=0.4cm,minimum width=0.6cm, xshift=12em] (5) {$x_5$:${X}$};

\node[text style,left of=5,align=center,yshift=0.6em,xshift=-2em] (5_id) {$y_5$:${Y}$};
\node[text style,left of=5,align=center,yshift=-0.6em,xshift=-2em] (5_prop) {$\mathsf{k}_y=$ c};
\node[text style,right of=5,align=left,xshift=2em] (5_v) {$\mathsf{k}_x=$ o};

\node[vertex style=Node,align=center,minimum height=0.4cm,minimum width=0.6cm, xshift=12em,yshift=-5em] (6) {$x_6$:${X}$};

\node[text style,left of=6,align=center,yshift=1.8em,xshift=0em] (6_id) {$y_6$:${Y}$};
\node[text style,left of=6,align=center,yshift=0.9em,xshift=-1.3em] (6_prop) {$\mathsf{k}_y=$ c};
\node[text style,below of=6,align=center,yshift=-0.5em] (6_v) {$\mathsf{k}_x=$ o};

            \node[vertex style=Node,align=center,minimum height=0.4cm,minimum width=0.6cm,xshift=6em,yshift=-5em] (7) {$x_7$:${X}$};

\node[text style,above of=7,align=left,yshift=1.7em,xshift=1em] (7_id) {$y$:${Y}$};
\node[text style,above of=7,align=left,yshift=0.8em,xshift=1.3em] (7_prop) {$\mathsf{k}_y=$ c};
\node[text style,below of=7,align=center,yshift=-0.5em] (7_v) {$\mathsf{k}_x=$ o};

\node[vertex style=Node,align=center,minimum height=0.4cm,minimum width=0.6cm,yshift=-5em] (8) {$x_8$:${X}$};

\node[text style,right of=8,align=center,yshift=1.8em,xshift=0em] (8_id) {$y_8$:${Y}$};
\node[text style,right of=8,align=left,yshift=0.9em,xshift=1.6em] (8_prop) {$\mathsf{k}_y=$ d};
\node[text style,below of=8,align=center,yshift=-0.5em] (8_v) {$\mathsf{k}_x=$ p};

\draw [->, bend left=-40,color=gray, align=center] (1_v.south) to  (1_prop.south west);
\draw [->, bend left=-90,color=gray, align=center] (2_v.west) to  (2_prop.west);
\draw [->, bend left=-90,color=gray, align=center] (3_v.west) to  (3_prop.west);
\draw [->, bend left=90,color=gray, align=center] (4_v.east) to  (4_prop.east);
\draw [->, bend left=40,color=gray, align=center] (5_v.south) to  (5_prop.south east);
\draw [->, bend left=60,color=gray, align=center] (6_v.west) to  (6_prop.south);
\draw [->, bend left=-40,color=gray, align=center] (7_v.east) to  (7_prop.south);
\draw [->, bend left=-70,color=gray, align=center] (8_v.east) to  (8_prop.south);

\node[vertex style=Node, right of=1,align=center,minimum height=0.4cm,minimum width=0.6cm,xshift=5em] (9) {$z$}
edge [<-,align=center]  (1) edge [<-,align=center]  (3) edge [<-,align=center]  (5) edge [<-,align=center]  (7) edge [<-,align=center]  (2) edge [<-,align=center]  (4) edge [<-,align=center]  (6) edge [<-,align=center]  (8);

\end{tikzpicture}
    \vspace*{-1em}
    \caption{Example graph for the calculation of per-dependency metrics}
    \label{fig:calculation_example}
\end{figure}

\changed{While the results for $m_\text{MaxRed.Pot}$ and $m_\text{AvgRed.Pot}$ might look acceptable, they only provide a limited view on redundancy.
As both of these two metrics do not have an upper limit, the presence of redundancies is the only statement that can be made based on them.
Contrary, the minimality value always lies between 0 and 1. In this example, it expresses that $\SI{\approx 57}{\percent}$ of the matches of the dependency $\varphi$ are redundant value combinations. 
}
\end{example} }

\changed{For measuring query performance, we rely on the query execution time in milliseconds ($m_\text{Q.Dur.}$) and on the number of database hits ($m_\text{Q.DBHit}$), which is a graph system internal metric to express query effort. 
The absolute values for $m_\text{Q.Dur.}$ depend on the machine hosting the \ac{lpg} system, $m_\text{Q.DBHit}$ is machine-independent. For both metrics lower values denote better performance.  }

\paragraph{Test scenarios}
We evaluate our approach using six graphs (see \Cref{tab:eval_graphs} for details).
\changed{To compare with related work, 
we use the same graphs as in \cite{skavantzos_third_2025} and add additional ones %
to evaluate additional types of dependencies.
As shown in \Cref{tab:experiment_datasets}, 
they are used in \num{12} scenarios with different GN-Schemas that cover all transformations discussed in this paper (see \Cref{tab:transformation-patterns}).}
In lack of appropriate methods for mining constraints in graphs, we follow the approach described by \cite{skavantzos_third_2025} and manually define dependencies for each graph\footnote{\changed{Details on the scenarios including the set of dependencies are available \edbtarxiv{online:\\}{in the \hyperref[sec:appendix]{appendix} and online:} \url{https://dmki-tuwien.github.io/lpg-normalization/evaluation_scenarios/}}} and compute the minimal cover for each GN-Schema. 
\changed{The normalization of each test scenario} was performed 
three times on fresh instances of Neo4j Enterprise, version 2025.12.1, and Memgraph 3.7.2.
\changed{
Query performance experiments were run ten times %
and the average of all runs is reported.}

\begin{figure*}
    \centering
    \vspace{-0.1em}
    {\footnotesize\textsf{Legend:\quad\begin{tikzpicture}
    \draw[draw=black, fill=black, fill opacity=0.12] (0,0) rectangle (0.6em,0.6em);
\end{tikzpicture} Graph object count \quad \begin{tikzpicture}
    \draw[draw=black, fill=black, fill opacity=0.76] (0,0) rectangle (0.6em,0.6em);
\end{tikzpicture} Property count}}\\
\vspace{-0.2em}
    \rotatebox{90}{\parbox{5.2cm}{\centering\footnotesize\textsf{Percentual change to the per-graph metrics\\ on a logarithmic scale (\%)}}}
    \begin{subfigure}[t]{0.48\linewidth}
    \includegraphics[width=\linewidth]{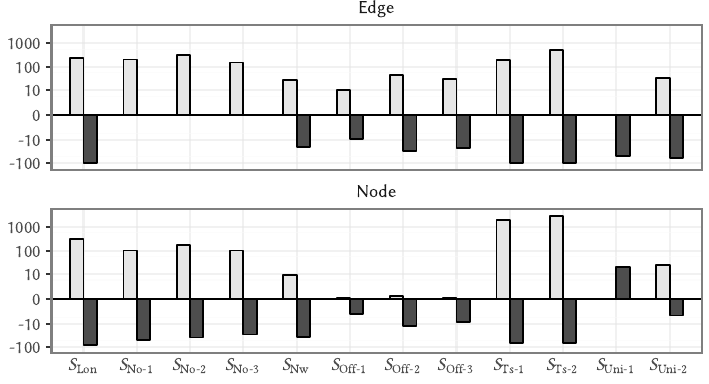}
    \caption{\emph{All} dependencies}
    \label{fig:impact-of-all}
    \end{subfigure}
    \hspace{0.02\linewidth}
    \begin{subfigure}[t]{0.16\linewidth}
    \includegraphics[width=\linewidth]{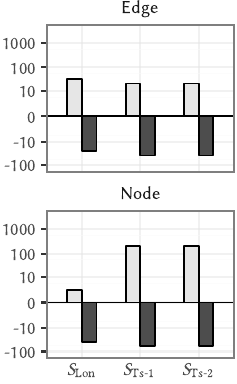}
    \caption{Without reification}
    \label{fig:impact-nodes}
    \end{subfigure}
    \hspace{0.02\linewidth}
    \begin{subfigure}[t]{0.16\linewidth}
    \includegraphics[width=\linewidth]{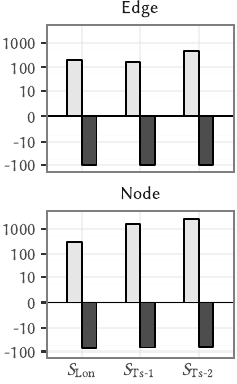}
    \caption{With reification}
    \label{fig:impact-edges}
    \end{subfigure}
    \vspace*{-1em}
    \caption{
    Percentual changes (log scale) to our \edbtarxiv{}{\changed{structural }}per-graph metrics following graph-native \ac{lpg} normalization grouped by graph object, considering all dependencies, only those leading to transformations with reification, and only those without reification.}
    \label{fig:transformation-impact}
\end{figure*}

\edbtarxiv{
    \subsubsection*{\textbf{Redundancy Reduction}}
}
{
    \subsection*{Redundancy Reduction}
    }

\changed{\Cref{tab:impact-on-per-dep-metrics} shows the results of our per-\ac{gnfd} metrics for each scenario 
before and after normalization.
It shows that originally there are redundancies in all scenarios. 
After performing graph-native normalization, i.e., performing the redundancy removing transformations (cf. \Cref{sec:native-property-graph-normal-forms}), all redundancies are eliminated. Thus, the per-\ac{gnfd} metrics 
$\metric_{\text{MaxRed.Pot.}}$,
$\metric_{\text{AvgRed.Pot.}}$, and
$\metric_{\text{Minim}}$ all return \num{1},
as for each dependency $\varphi: Q::X\determ Y$ there is only a single combination of values for the property keys in $X$ and $Y$. 
}

\begin{table}
\vspace*{-1em}
    \centering
        \caption{The results of the per-\ac{gnfd} metrics before and after graph-native \ac{lpg} normalization %
        }
    \label{tab:impact-on-per-dep-metrics}
    \vspace*{-1em}
    \footnotesize
    \begin{tabular}{rr@{\enspace}c@{\enspace}r@{\enspace}c@{\enspace}r@{\enspace}c@{\enspace}}
    \toprule
    & \multicolumn{2}{c}{$\mathsf{avg}(\metric_{\text{MaxRed.Pot.}})$} & \multicolumn{2}{c}{$\mathsf{avg}(\metric_{\text{AvgRed.Pot.}})$} & \multicolumn{2}{c}{$\mathsf{avg}(\metric_{\text{Minim}})$} \\ %
       \textbf{Scenario} & \emph{before} & \emph{after} & \emph{before} & \emph{after} & \emph{before} & \emph{after} \\ %
       \midrule
       $S_\text{Lon}$      & \num{112.60} & \num{1} & $\approx \num{36.42}$ & \num{1} & $\approx \num{0.41}$ & \num{1} \\ %
       $S_\text{Nw}$      & \num{16.00} & \num{1} & $\approx \num{5.16}$ & \num{1} & $\approx \num{0.55}$ & \num{1}\\ %
       $S_\text{No-1}$      & \num{1.25} & \num{1} & \num{1.25} & \num{1} & \num{0.75} & \num{1} \\ %
       $S_\text{No-2}$      & \num{1.25} & \num{1} & \num{1.25} & \num{1} & \num{0.75} & \num{1} \\ %
       $S_\text{No-3}$      & \num{1.20} & \num{1} & \num{1.20} & \num{1} & \num{0.80} & \num{1} \\ %
       $S_\text{Off-1}$      & $\approx \num{264881.33}$ & \num{1} & $\approx \num{613.07}$ & \num{1} & $\approx 0.00$ & \num{1} \\ %
       $S_\text{Off-2}$      & \num{175898.00} & \num{1} & $\approx \num{513.25}$ & \num{1} & $\approx \num{0.00}$ & \num{1} \\ %
       $S_\text{Off-3}$      & \num{239572.20} & \num{1} & $\approx \num{21645.63}$ & \num{1} & $\approx \num{0.00}$ & \num{1}  \\ %
       $S_\text{Ts-1}$      & $\approx \num{6790.33}$ & \num{1} & $\approx \num{563.93}$ & \num{1} & $\approx 0.68$ & \num{1} \\ %
       $S_\text{Ts-2}$      & \num{11631.00} & \num{1} & $\approx \num{964.82}$ & \num{1} & $\approx 0.59$
       & \num{1}\\ %
       $S_\text{Uni-1}$      & \num{3.00} & \num{1} & \num{3.00} & \num{1} & \num{0.00} & \num{1} \\ %
       $S_\text{Uni-2}$      & \num{3.00} & \num{1} & \num{3.00} & \num{1} & \num{0.00} & \num{1} \\ %
\bottomrule
    \end{tabular}
\end{table}

\edbtarxiv{\subsubsection*{\textbf{Effects on \protect{\ac{lpg}} Structure}}}
{
    \subsection*{Effects of Graph-Native Normalization on  \protect{\acp{lpg}} Structure}
}
\changed{Graph-native \ac{lpg} normalization affects the structure and size of an \ac{lpg}, with the precise impact depending on the transformations applied. \Cref{fig:transformation-impact} shows the relative percentage change of our per-graph metrics on a logarithmic scale, broken down by dependency type: \Cref{fig:impact-of-all} for all dependencies, \Cref{fig:impact-edges} for those leading to reification, and \Cref{fig:impact-nodes} for those that do not. This reveals how different redundancy patterns (cf. \Cref{tab:transformation-patterns}) influence normalization.}

\changed{
Transformations that create new nodes without reification (e.g., $\varInfomorphism_\text{within-n}$ and $\varInfomorphism_\text{between-np-ep}$) introduce new nodes for the properties of the descriptor of a dependency and connect them to the existing subgraph via an edge.
As \Cref{fig:impact-nodes} shows, this increases the number of nodes and edges while decreasing the number of properties -- an effect of storing each piece of information only once.
}

\changed{Transformations involving reification (e.g., $\varInfomorphism_\text{within-e}$, $\varInfomorphism_\text{between-ep-np}$, and $\varInfomorphism_\text{between-ep-n}$) have the most significant effect on per-graph metrics. As \Cref{fig:impact-edges} shows, they also reduce the average number of properties but, due to reification, create substantially more new nodes and edges.}

\changed{Overall, \Cref{fig:impact-of-all} shows that graph-native normalization almost always increases the number of graph objects and decreases the average number of properties per object. The exception is $\varInfomorphism_\text{between-n-ep}$, which ``moves'' a property from an edge to a node, as in $S_\text{Uni-1}$. There, $\metric_{\#\text{Node}}$ and $\metric_{\#\text{Edge}}$ remain constant while $\metric_{\text{AvgPropEdge}}$ decreases and $\metric_{\text{AvgPropNode}}$ increases. This reflects the movement of the property values from edges to unique nodes.
}

With respect to runtime, we observed three influencing factors: 
the number of graph objects before normalization, 
the number of dependencies in the minimal cover of the GN-Schema to consider,
and the type of graph transformations applied during normalization.
\changed{The latter correlates with the number of graph objects involved: transformations using reification are the slowest. For example, in $S_\text{Lon}$ on Neo4J, the transformation of type $\varInfomorphism_\text{within-n}$takes $\approx \SI{0.58}{\second}$ on average, while the transformation of type $\varInfomorphism_\text{within-e}$ takes $\approx \SI{0.97}{\second}$.%
\footnote{Detailed results can be found at \url{https://github.com/dmki-tuwien/lpg-normalization}}}

\begin{table}
    \centering
    \caption{Comparison of per-graph metrics after normalizing using our approach and \edbtarxiv{related work~\cite{skavantzos_third_2025}}{the approach from \citet{skavantzos_third_2025}}}
    \label{tab:eval-rel-work}
    \vspace*{-1em}
    \footnotesize
    \begin{tabular}{@{\enspace}r l r@{\enspace}r@{\enspace}r@{\enspace}r@{\enspace}}
    \toprule
    \textbf{Scenario} & \textbf{Approach} & \textbf{$\metric_{\#\text{Node}}$} & \textbf{$\metric_{\#\text{Edge}}$} & \textbf{$\metric_{\text{AvgPropNode}}$} & \textbf{$\metric_{\text{AvgPropEdge}}$}\\
    \midrule
        $S_\text{No-1}$ & Graph-Native & \num{6} & \num{6} & $\approx \num{1.67}$ & \num{0} \\
         \hline
         \multirow{2}{*}{$S_\text{No-2}$} & \cite{skavantzos_third_2025} &  \num{8} & \num{8} & \num{2} & \num{0}  \\
         & Graph-Native &  \num{8} & \num{8} & \num{2} & \num{0}  \\
         \hline
         \multirow{1}{*}{$S_\text{No-3}$} & Graph-Native & \num{6} & \num{5} & $\approx \num{2.33}$ & \num{0} \\
         \hline
      \multirow{2}{*}{$S_\text{Nw}$}   & \cite{skavantzos_third_2025} & \num{1124} & \num{3969} & $\approx \num{8.30}$ & $\approx \num{2.71}$   \\
         & Graph-Native & \num{1124} & \num{3969} & $\approx \num{8.30}$ & $\approx \num{2.71}$  \\
         \hline
         \multirow{2}{*}{$S_\text{Off-1}$} & \cite{skavantzos_third_2025} & \num{2017381} & \num{3681864} & $\approx \num{9.52}$ & $\approx \num{2.95}$ \\
                  & Graph-Native & \num{2017381} & \num{3681864} & $\approx \num{9.52}$ & $\approx \num{2.95}$ \\
         \hline
         \multirow{2}{*}{$S_\text{Off-2}$} & \cite{skavantzos_third_2025} & \num{2021602} & \num{4854231} & $\approx \num{8.60}$ & $\approx \num{2.24}$\\
         & Graph-Native & \num{2021602} & \num{4854231} & $\approx \num{8.60}$ & $\approx \num{2.24}$\\
         \hline
         \multirow{2}{*}{$S_\text{Off-3}$} & \cite{skavantzos_third_2025} &  \num{2017388} & \num{4367058} & $\approx \num{9.01}$ & $\approx \num{2.49}$ \\
                  & Graph-Native & \num{2017388} & \num{4367058} & $\approx \num{9.01}$ & $\approx \num{2.49}$ \\
                  \bottomrule
         
    \end{tabular}
\end{table}

\edbtarxiv{\subsubsection*{\textbf{Effects on Query~Performance}}}
{
    \subsection*{Effects of Graph-Native Normalization on Query~Performance}
}

\begin{figure}
    \centering
    \vspace*{-0.7em}
    {\footnotesize\textsf{Legend:\quad\begin{tikzpicture}
    \draw[draw=black, fill=black, fill opacity=0.76] (0,0) rectangle (0.6em,0.6em);
\end{tikzpicture} Original queries \quad \begin{tikzpicture}
    \draw[draw=black, fill=black, fill opacity=0.12] (0,0) rectangle (0.6em,0.6em);
\end{tikzpicture} Rewritten queries after normalization}}\\
\vspace{-0.1em}

\begin{minipage}{0.05\linewidth}
\rotatebox{90}{\footnotesize\textsf{$m_\text{Q.DBHit}$ on a logarithmic scale}}
\end{minipage}
\begin{minipage}{0.93\linewidth}
  \includegraphics[width=\linewidth]{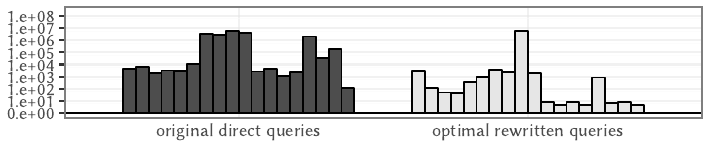} \\[-0.3em]
 \includegraphics[width=\linewidth]{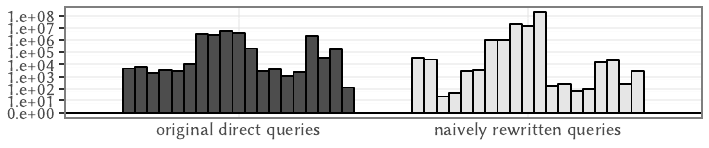} \\[-0.3em]
\includegraphics[width=\linewidth]{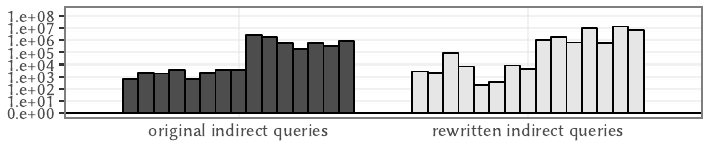} 
\end{minipage}

    \vspace*{-1.5em}
    \caption{Comparison of query performance ($m_\text{Q.DBHit}$) between original and rewritten queries}
    \label{fig:query-performance}
\end{figure}

\edbtarxiv{}{\begin{table*}[]
        \centering
        \caption{Effects of graph-native normalization on query~performance. Better results per graph system are \uline{underlined}. Execution times are the average of 10 runs and are rounded to two decimal places. %
        }
        \label{tab:query-exp-overview}
        \vspace*{-1em}
        \footnotesize
        \begin{tabular}{@{~}l @{\enspace} c @{\hspace{1em}} c @{\hspace{1em}}L{1.8cm}@{\enspace}c@{\enspace}c rr rr rr rr@{~}}
\toprule
\multirow{3}{*}{\textbf{Scenario}} & \multirow{3}{*}{\parbox{0.6cm}{\centering\textbf{Query type}}} & \multirow{3}{*}{\parbox{0.6cm}{\centering\textbf{Query no.}}} & \multirow{3}{1.8cm}{\textbf{Affected by transformation type}} & \multirow{3}{*}{\textbf{Direct}} & \multirow{3}{1.1cm}{\centering\textbf{Optimal rewriting}} & \multicolumn{4}{c}{\textbf{Execution time (in \si{\milli\second})}} & \multicolumn{4}{c}{\textbf{Database hits}} \\
 & & & & & & \multicolumn{2}{c}{\emph{before}} & \multicolumn{2}{c}{\emph{after}} & \multicolumn{2}{c}{\emph{before}} & \multicolumn{2}{c}{\emph{after}} \\
  & & & & & & Memgraph & Neo4j & Memgraph & Neo4j & Memgraph & Neo4j & Memgraph & Neo4j \\
\midrule
\multirow{12}{*}{$S_\text{Lon}$}
     & \multirow{9}{*}{\rotatebox{90}{lookup}} 
       & 1a & $\varInfomorphism_\text{within-e}$ & $\checkmark$ & $\checkmark$ & \num{2.51} & \num{3.94} & \uline{\num{0.47}} & \uline{\num{0.37}} & \num{4589} & \num{5857} & \uline{\num{2736}} & \uline{\num{113}} \\
     & & 1b &  $\varInfomorphism_\text{within-e}$ & $\checkmark$ & $\times$ & \uline{\num{2.51}} & \uline{\num{3.94}} & \num{10.81} & \num{16.09} & \uline{\num{4589}} & \uline{\num{5857}} & \num{35898} & \num{25517} \\
     \cmidrule{3-14}
     & & 2 &  $\varInfomorphism_\text{within-e}$ & $\times$ & $\times$ & \uline{\num{0.33}} & \num{0.88} & \num{0.68} & \uline{\num{0.57}} & \uline{\num{677}} & \uline{\num{1971}} & \num{2709} & \num{2030} \\
     \cmidrule{3-14}
     & & 3a & $\varInfomorphism_\text{within-n}$ & $\checkmark$ & $\checkmark$ & \num{1.19} & \num{1.29} & \uline{\num{0.09}} & \uline{\num{0.37}} & \num{1964} & \num{3241} & \uline{\num{50}} & \uline{\num{46}} \\
     & & 3b & $\varInfomorphism_\text{within-n}$ & $\checkmark$ & $\times$ & \num{1.19} & \num{1.29} & \uline{\num{0.04}} & \uline{\num{0.24}} & \num{1964} & \num{3241} & \uline{\num{22}} & \uline{\num{46}} \\
     \cmidrule{3-14}
     & & 4 &  $\varInfomorphism_\text{within-e}$, $\varInfomorphism_\text{within-n}$ & $\times$ & $\times$ & \uline{\num{0.71}} & \uline{\num{1.71}} & \num{13.92} & \num{1.90} & \uline{\num{1824}} & \uline{\num{3459}} & \num{88164} & \num{6857} \\
     \cmidrule{3-14}
     & & 5 &  $\varInfomorphism_\text{within-e}$ & $\times$ & $\times$ & \num{0.31} & \num{0.76} & \uline{\num{0.12}} & \uline{\num{0.44}} & \num{670} & \num{1971} & \uline{\num{217}} & \uline{\num{355}} \\
     \cmidrule{2-14} & \multirow{3}{*}{\rotatebox{90}{update}}
       & 1a & $\varInfomorphism_\text{within-e}$ & $\checkmark$ & $\times$ & \num{1.09} & \num{1.98} & \uline{\num{0.12}} & \uline{\num{0.38}} & \num{2706} & \num{4001} & \uline{\num{152}} & \uline{\num{243}} \\
     & & 1b & $\varInfomorphism_\text{within-e}$ & $\checkmark$ & $\times$ & \num{1.09} & \num{1.98} & \uline{\num{0.13}} & \uline{\num{0.35}} & \num{2706} & \num{4001} & \uline{\num{156}} & \uline{\num{149}} \\
     & & 1c & $\varInfomorphism_\text{within-e}$ & $\checkmark$ & $\checkmark$ & \num{1.09} & \num{1.98} & \uline{\num{0.07}} & \uline{\num{0.29}} & \num{2706} & \num{4001} & \uline{\num{7}} & \uline{\num{4}} \\
\midrule
\multirow{6}{*}{$S_\text{Nw}$}
     & \multirow{3}{*}{\rotatebox{90}{lookup}} 
       & 1 &  $\varInfomorphism_\text{within-n}$ & $\times$ & $\times$ & \uline{\num{1.52}} & \uline{\num{7.51}} & \num{1.80} & \num{12.33} & \uline{\num{3531}} & \uline{\num{3321}} & \num{8267} & \num{4507} \\
       \cmidrule{3-14}
     & & 2a &  $\varInfomorphism_\text{within-n}$ & $\checkmark$ & $\times$ & \num{1.44} & \num{4.56} & \uline{\num{0.76}} & \uline{\num{1.73}} & \num{2790} & \num{10791} & \uline{\num{2765}} & \uline{\num{3559}} \\
     & & 2b &  $\varInfomorphism_\text{within-n}$ & $\checkmark$ & $\checkmark$ & \num{1.44} & \num{4.56} & \uline{\num{0.24}} & \uline{\num{0.97}} & \num{2790} & \num{10791} & \uline{\num{362}} & \uline{\num{980}} \\
     \cmidrule{2-14} & \multirow{3}{*}{\rotatebox{90}{update}}
       & 1a & $\varInfomorphism_\text{within-n}$ &  $\checkmark$ & $\times$ & \num{0.28} & \num{1.28} & \uline{\num{0.06}} & \uline{\num{0.51}} & \num{1079} & \num{2529} & \uline{\num{65}} & \uline{\num{98}} \\
     & & 1b & $\varInfomorphism_\text{within-n}$ & $\checkmark$  & $\checkmark$ & \num{0.28} & \num{1.28} & \uline{\num{0.02}} & \uline{\num{0.40}} & \num{1079} & \num{2529} & \uline{\num{7}} & \uline{\num{4}} \\
     & & 1c & $\varInfomorphism_\text{within-n}$ &  $\checkmark$  & $\times$ & \num{0.28} & \num{1.28} & \uline{\num{0.07}} & \uline{\num{0.53}} & \num{1079} & \num{2529} & \uline{\num{69}} & \uline{\num{62}} \\
\midrule
\multirow{6}{*}{$S_\text{Off-1}$}
     & \multirow{3}{*}{\rotatebox{90}{lookup}} 
       & 1a &  $\varInfomorphism_\text{within-n}$ & $\checkmark$  & $\times$ & \num{923.41} & \num{406.97} & \uline{\num{170.21}} & \uline{\num{89.19}} & \num{3027613} & \num{2524955} & \uline{\num{1030373}} & \uline{\num{1031379}} \\
     & & 1b &  $\varInfomorphism_\text{within-n}$ & $\checkmark$  & $\checkmark$ & \num{923.41} & \num{406.97} & \uline{\num{1.11}} & \uline{\num{2.12}} & \num{3027613} & \num{2524955} & \uline{\num{3438}} & \uline{\num{2575}} \\
       \cmidrule{3-14}
     & & 2 & $\varInfomorphism_\text{within-n}$ & $\times$ & $\times$ & \num{797.36} & \uline{\num{254.31}} & \uline{\num{283.41}} & \num{282.81} & \num{2702775} & \uline{\num{1715964}} & \uline{\num{1030567}} & \num{1716575} \\
     \cmidrule{2-14} & \multirow{3}{*}{\rotatebox{90}{update}}
       & 1a & $\varInfomorphism_\text{within-n}$ & $\checkmark$ & $\checkmark$ & \num{336.04} & \num{18.27} & \uline{\num{0.71}} & \uline{\num{0.28}} & \num{2026155} & \num{33692} & \uline{\num{866}} & \uline{\num{6}} \\
     & & 1b & $\varInfomorphism_\text{within-n}$ & $\checkmark$ & $\times$ & \num{336.04} & \num{18.27} & \uline{\num{1.97}} & \uline{\num{2.93}} & \num{2026155} & \num{33692} & \uline{\num{15310}} & \uline{\num{14447}} \\
     & & 1c & $\varInfomorphism_\text{within-n}$ & $\checkmark$ & $\times$ & \num{336.04} & \num{18.27} & \uline{\num{4.07}} & \uline{\num{8.14}} & \num{2026155} & \num{33692} & \uline{\num{15306}} & \uline{\num{24071}} \\
\midrule
\multirow{15}{*}{$S_\text{Ts-2}$}
     & \multirow{9}{*}{\rotatebox{90}{lookup}}
       & 1a &  $\varInfomorphism_\text{between-ep-np}$, $\varInfomorphism_\text{between-np-ep}$\vspace{0.4em} & $\checkmark$  & $\checkmark$ & \num{1075.67} & \num{377.66} & \uline{\num{301.39}} & \uline{\num{1.66}} & \num{5675074} & \num{3654287} & \uline{\num{5673812}} & \uline{\num{2169}} \\
     & & 1b &  $\varInfomorphism_\text{between-ep-np}$, $\varInfomorphism_\text{between-np-ep}$  & $\checkmark$  & $\times$ & \uline{\num{1075.67}} & \num{377.66} & \num{2122.51} & \uline{\num{1.47}} & \uline{\num{5675074}} & \uline{\num{3654287}} & \num{21532367} & \num{14623422} \\
       \cmidrule{3-14}
     & & 2 &  $\varInfomorphism_\text{between-ep-np}$, $\varInfomorphism_\text{between-np-ep}$, $\varInfomorphism_\text{within-n}$  & $\times$ & $\times$ & \uline{\num{81.39}} & \num{62.83} & \num{8231.41} & \uline{\num{0.04}} & \uline{\num{200599}} & \uline{\num{589524}} & \num{191562509} & \num{637897} \\
       \cmidrule{3-14}
     & & 3 & $\varInfomorphism_\text{within-n}$ & $\times$ & $\times$ & \uline{\num{78.94}} & \num{57.43} & \num{1042.67} & \uline{\num{54.53}} & \uline{\num{194463}} & \num{581320} & \num{9712778} & \uline{\num{578976}} \\
       \cmidrule{3-14}
     & & 4 &  $\varInfomorphism_\text{within-n}$ (twice) & $\times$ & $\times$ & \uline{\num{139.21}} & \uline{\num{116.81}} & \num{1226.04} & \num{514.43} & \uline{\num{356326}} & \uline{\num{817234}} & \num{13381704} & \num{6310756} \\
     \cmidrule{2-14} & \multirow{5}{*}{\rotatebox{90}{update}}
       & 1a &  $\varInfomorphism_\text{between-ep-np}$, $\varInfomorphism_\text{between-np-ep}$\vspace{0.4em} & $\checkmark$ & $\times$ & \num{30.60} & \num{0.90} & \uline{\num{0.09}} & \uline{\num{0.00}} & \num{194278} & \uline{\num{123}} & \uline{\num{230}} & \num{2960} \\
     & & 1b &  $\varInfomorphism_\text{between-ep-np}$, $\varInfomorphism_\text{between-np-ep}$\vspace{0.4em} & $\checkmark$ & $\times$ & \num{30.60} & \num{0.90} & \uline{\num{0.11}} & \uline{\num{0.00}} & \num{194278} & \uline{\num{123}} & \uline{\num{226}} & \num{2989} \\
     & & 1c &  $\varInfomorphism_\text{between-ep-np}$, $\varInfomorphism_\text{between-np-ep}$ & $\checkmark$ & $\checkmark$ & \num{30.60} & \num{0.90} & \uline{\num{0.07}} & \uline{\num{0.38}} & \num{194278} & \num{123} & \uline{\num{7}} & \uline{\num{4}} \\
\bottomrule
\end{tabular}

    \end{table*}}
\edbtarxiv{
\changed{
Normalization 
changes the structure of an LPG and hence also
affects query performance. We define 17 lookup and update queries
across $S_{\text{Lon}}$, $S_{\text{Nw}}$, $S_{\text{Off-1}}$, and
$S_{\text{Ts-2}}$ (cf. Table 8), classified along two
axes. \emph{Direct} queries retrieve property keys $k \in X \cup Y$ of a single GO-FD $\varphi: P :: X \Rightarrow Y$, i.e., information that is redundant in the original graph but unique after normalization; \emph{indirect} queries consist of larger patterns. For direct queries, the normalized graph allows for two different query rewritings: an \emph{optimal} rewriting that targets only the newly created nodes, and a \emph{naive} rewriting which is non-optimal. Figure~\ref{fig:query-performance} reports $m_{\text{Q.DBHit}}$ before and after normalization for each query and its rewritings. 
}

\changed{Direct queries with optimal rewriting show consistent reductions in database hits, independent of the type of transformation (including reification), and the type of query (lookup/update). 
Query runtimes ($m_{\text{Q.Dur.}}$) improve in a similar way
 as  database hits ($m_{\text{Q.DBHit}}$).
Naive rewritings and indirect queries show a mixed picture: some benefit from reduced redundancy, others incur a cost from the larger patterns.
The takeaway is that direct queries are a meaningful class 
since they represent exactly what users are likely to query:
dependencies typically mark \emph{important} information.
Detailed queries and per-query results can be found in the full version~\cite{schrott2026graphnative}.
}
}{ %
\changed{
Normalization 
changes the structure of an \ac{lpg} and hence also
affects query performance.
To get insight in the precise influence, we consider four scenarios (graph and dependency combinations, cf. \Cref{tab:experiment_datasets}) and define queries for them. 
These queries are formulated using the structure of the graph \emph{before} normalization, we will refer to this as the \emph{original graph}. To properly guide the discussion on the effect of normalization on querying, we designed two types of queries, which we call \emph{direct} and \emph{indirect}. Direct queries are formulated specifically for a given \ac{gnfd} and ask for pieces of information that are redundant in the original graph, but are unique in the normalized graph.
Formally, this corresponds to queries which only ask for property keys $k\in X\cup Y$ of a single \ac{gnfd} $\varphi: P :: X \determ Y$. 
Indirect queries go beyond querying redundant pieces of information, and thus require some more complex navigation.
}

\changed{For all queries defined on the original graph, we need to have a counterpart for the normalized graph. For direct queries, we consider two kinds of counterparts. First, the \emph{optimal} rewriting: these queries only need to consider the newly created graph objects (these represent the normalized non-redundant pieces of information). On the other hand, we have the 
\emph{naive} rewriting: its query pattern is the fully transformed scope. However, given the new structure, this pattern may be unnecessary large.}

\changed{The hypothesis is that direct queries improve in performance under optimal rewriting. For all other queries, it depends on the precise graph and query structure.}

\changed{To show this impact, \num{17} lookup and update queries\footnote{\changed{The queries are available \edbtarxiv{online:\\}{in the \hyperref[sec:appendix]{appendix} and online:} \url{https://dmki-tuwien.github.io/lpg-normalization/evaluation_scenarios/}}} are executed on four different scenarios. \Cref{fig:query-performance} shows the results on a high-level using $m_\text{Q.DBHit}$. Detailed results for both metrics, $m_\text{Q.DBHit}$ and $m_\text{Q.Dur}$,  before and after normalization, are listed per query in \Cref{tab:query-exp-overview}. There, the best performance per query and metric is \uline{underlined}.} 

\changed{The results confirm our hypothesis. For direct queries with optimal rewriting, performance consistently improves over the original graph, both in terms of execution time and database hits. Two mechanisms drive this improvement: first, the elimination of duplicate information directly reduces the number of database hits required to retrieve a given fact; second, the keys introduced by normalization on the newly created nodes are backed by automatic indices, which further accelerates lookups. The contrast between $S_\text{Lon}$-1a and $S_\text{Lon}$-1b illustrates the importance of optimal rewriting: both queries retrieve the same result, but $S_\text{Lon}$-1b retains the unnecessarily large pattern of the fully transformed scope and in fact performs \emph{worse} than on the original graph, while $S_\text{Lon}$-1a uses the optimal rewriting and benefits from the performance gains described above. For indirect queries and direct queries under naive rewriting, the picture is mixed: some queries benefit from the reduced redundancy, while others -- such as $S_\text{Lon}$-2 -- suffer from the additional navigation introduced by new graph objects and reification nodes. Regarding the query type, we observed no systematic difference  between lookup and update queries.}

\changed{We conclude that there is a mixed picture regarding the effects of normalization on query performance (across all types of transformations). 
However, for direct queries with optimal rewriting, performance does improve. We argue that this is a valuable insight: when a dependency is defined, it typically means that the affected information is of importance. Thus, one might often explicitly query for that information. In practice, direct queries with optimal rewriting might therefore not only be a theoretical idea, but a class of queries that naturally arises whenever the dependencies reflect the meaning of the data.}
} %

\edbtarxiv{\subsubsection*{\textbf{Comparison to Related Work}}}{\subsection*{Comparison to Related Work}}

\changed{To show that our approach generalizes \citet{skavantzos_third_2025}, we compare the per-graph metrics on the scenarios supported by both works: $S_\text{No-2}$, $S_\text{Nw}$, $S_\text{Off-1}$, $S_\text{Off-2}$, and $S_\text{Off-3}$. As \Cref{tab:eval-rel-work} shows, both approaches yield identical metrics, experimentally corroborating that $\varInfomorphism_\text{between-n}$ is equivalent to the \ac{lpg} node normalization of \cite{skavantzos_third_2025}.}

\changed{
Scenario $S_\text{No-1}$ could in principle be covered by~\cite{skavantzos_third_2025}, but is omitted due to a subtle difference in the treatment of keys. 
\edbtarxiv{In \cite{skavantzos_third_2025}, \acp{guc} are treated}{\citet{skavantzos_third_2025} treat \acp{guc}} 
as dependencies that determine all property values of a node ($S_\text{No-2}$), whereas we treat keys as dependencies that determine a node's identity, which in turn determines its properties (cf. structurally implied dependencies, \Cref{sec:dependency}). As discussed in \Cref{sec:native-property-graph-normal-forms}, such key dependencies do not cause redundancies and require no transformation, yielding a smaller -- yet normalized -- graph (visible when comparing $S_\text{No-1}$ and $S_\text{No-2}$). Scenario $S_\text{No-3}$ is mentioned in~\cite{skavantzos_third_2025} but without results, so it is excluded from \Cref{tab:eval-rel-work}.
}

\edbtarxiv{\subsubsection*{\textbf{Discussion}}}
{
    \subsection*{Discussion}
}

As discussed above, we see that graph-native normalization extends the existing approach to \ac{lpg} normalization and eliminates redundancies in \acp{lpg}. 
However, it also leads to an increase in the number of graph objects,
i.e., nodes and edges, as \Cref{fig:transformation-impact} shows. 
As an effect, the complexity of a graph's structure increases as well.
While the total number of properties in a graph is reduced as an effect of redundancy reduction,
it can happen that the average number of properties in nodes increases, 
due to reification and the movement of properties, 
as discussed for \changed{$S_\text{Uni-1}$}.
Thus, we conclude that optimizing for redundancy reduction and a minimal number of graph objects in a graph 
simultaneously is not possible. 

\changed{However, with respect to query performance, normalization (and thus redundancy reduction) can lead to improvement. This is an interesting insight compared to the relational data model, where}
denormalized schemas are sometimes intentionally used to ease and speed up querying of relations. %

\changed{Another difference between the relational model and \acp{lpg} concerns the definition of dependencies. 
LPGs provide a larger number of possible modeling choices compared to the relational model. Thus, LPG normalization depends more heavily on the initial structure of the data. This means that, compared to the relational model, more attention must be given to the design of the GO-FDs.
}

\section{Conclusion}
\label{sec:conclusion}
Building upon related work and in analogy with relational database systems,
this paper proposes a graph-native approach to \acf{lpg} normalization. 
The novel process employs graph transformation rules to normalize graphs into \ac{gnnf} 
and 
\changed{subsequently reduces redundancy}. 
Our evaluation experimentally shows the applicability of our approach using a variety of graphs and functional dependencies; 
\changed{it highlights} the tradeoff between having a simple structure and the reduction of redundancies,  \changed{and it emphasizes the possible impact on query performance}.

These observations open several avenues for future research. 
First, while we assume the availability of \acp{gnfd}, 
in practice these dependencies may \changed{need to be determined automatically, e.g., through mining}\edbtarxiv{.}{\changed{. However, following related work~\cite{skavantzos_third_2025}, current mining techniques only can provide heuristics on meaningful dependencies; thus automatic \ac{gnfd} discovery is an interesting next step. }}
Another aspect of future work is to consider a broader range of dependencies using complex graph patterns as their scope. 
\changed{Based on our insights on query performance we see the need for further research on rewriting queries for normalized graphs as well.}
Finally, we plan to investigate to what extent our work can be applied to knowledge graphs in \acs{rdf}.

\section*{Artifacts}
\label{sec:artifacts}
All artifacts associated with this work are available online on GitHub: \url{https://github.com/dmki-tuwien/lpg-normalization}

\bibliographystyle{helpers/ACM-Reference-Format}
\bibliography{helpers/references}


\begin{thebibliography}{40}


\ifx \showCODEN    \undefined \def \showCODEN     #1{\unskip}     \fi
\ifx \showISBNx    \undefined \def \showISBNx     #1{\unskip}     \fi
\ifx \showISBNxiii \undefined \def \showISBNxiii  #1{\unskip}     \fi
\ifx \showISSN     \undefined \def \showISSN      #1{\unskip}     \fi
\ifx \showLCCN     \undefined \def \showLCCN      #1{\unskip}     \fi
\ifx \shownote     \undefined \def \shownote      #1{#1}          \fi
\ifx \showarticletitle \undefined \def \showarticletitle #1{#1}   \fi
\ifx \showURL      \undefined \def \showURL       {\relax}        \fi
\providecommand\bibfield[2]{#2}
\providecommand\bibinfo[2]{#2}
\providecommand\natexlab[1]{#1}
\providecommand\showeprint[2][]{arXiv:#2}

\bibitem[Abiteboul et~al\mbox{.}(1995)]%
        {alice}
\bibfield{author}{\bibinfo{person}{Serge Abiteboul}, \bibinfo{person}{Richard
  Hull}, {and} \bibinfo{person}{Victor Vianu}.}
  \bibinfo{year}{1995}\natexlab{}.
\newblock \bibinfo{booktitle}{\emph{Foundations of Databases}}.
\newblock \bibinfo{publisher}{Addison-Wesley}.
\newblock


\bibitem[Ahmetaj et~al\mbox{.}(2025)]%
        {AhmetajBHHJGMMM25}
\bibfield{author}{\bibinfo{person}{Shqiponja Ahmetaj}, \bibinfo{person}{Iovka
  Boneva}, \bibinfo{person}{Jan Hidders}, \bibinfo{person}{Katja Hose},
  \bibinfo{person}{Maxime Jakubowski}, \bibinfo{person}{Jos{\'{e}} Emilio~Labra
  Gayo}, \bibinfo{person}{Wim Martens}, \bibinfo{person}{Fabio Mogavero},
  \bibinfo{person}{Filip Murlak}, \bibinfo{person}{Cem Okulmus},
  \bibinfo{person}{Axel Polleres}, \bibinfo{person}{Ognjen Savkovic},
  \bibinfo{person}{Mantas Simkus}, {and} \bibinfo{person}{Dominik Tomaszuk}.}
  \bibinfo{year}{2025}\natexlab{}.
\newblock \showarticletitle{Common Foundations for SHACL, ShEx, and PG-Schema}.
  In \bibinfo{booktitle}{\emph{{WWW}}}. \bibinfo{publisher}{{ACM}},
  \bibinfo{pages}{8--21}.
\newblock


\bibitem[Angles et~al\mbox{.}(2023)]%
        {angles_pg-schema_2023}
\bibfield{author}{\bibinfo{person}{Renzo Angles}, \bibinfo{person}{Angela
  Bonifati}, \bibinfo{person}{Stefania Dumbrava}, \bibinfo{person}{George
  Fletcher}, \bibinfo{person}{Alastair Green}, \bibinfo{person}{Jan Hidders},
  \bibinfo{person}{Bei Li}, \bibinfo{person}{Leonid Libkin},
  \bibinfo{person}{Victor Marsault}, \bibinfo{person}{Wim Martens},
  \bibinfo{person}{Filip Murlak}, \bibinfo{person}{Stefan Plantikow},
  \bibinfo{person}{Ognjen Savkovic}, \bibinfo{person}{Michael Schmidt},
  \bibinfo{person}{Juan Sequeda}, \bibinfo{person}{Slawek Staworko},
  \bibinfo{person}{Dominik Tomaszuk}, \bibinfo{person}{Hannes Voigt},
  \bibinfo{person}{Domagoj Vrgoc}, \bibinfo{person}{Mingxi Wu}, {and}
  \bibinfo{person}{Dusan Zivkovic}.} \bibinfo{year}{2023}\natexlab{}.
\newblock \showarticletitle{PG-Schema: Schemas for Property Graphs}.
\newblock \bibinfo{journal}{\emph{Proc. {ACM} Manag. Data}}
  \bibinfo{volume}{1}, \bibinfo{number}{2} (\bibinfo{year}{2023}),
  \bibinfo{pages}{198:1--198:25}.
\newblock


\bibitem[Angles et~al\mbox{.}(2021)]%
        {angles_pg-keys_2021}
\bibfield{author}{\bibinfo{person}{Renzo Angles}, \bibinfo{person}{Angela
  Bonifati}, \bibinfo{person}{Stefania Dumbrava}, \bibinfo{person}{George
  Fletcher}, \bibinfo{person}{Keith~W. Hare}, \bibinfo{person}{Jan Hidders},
  \bibinfo{person}{Victor~E. Lee}, \bibinfo{person}{Bei Li},
  \bibinfo{person}{Leonid Libkin}, \bibinfo{person}{Wim Martens},
  \bibinfo{person}{Filip Murlak}, \bibinfo{person}{Josh Perryman},
  \bibinfo{person}{Ognjen Savkovic}, \bibinfo{person}{Michael Schmidt},
  \bibinfo{person}{Juan~F. Sequeda}, \bibinfo{person}{Slawek Staworko}, {and}
  \bibinfo{person}{Dominik Tomaszuk}.} \bibinfo{year}{2021}\natexlab{}.
\newblock \showarticletitle{PG-Keys: Keys for Property Graphs}. In
  \bibinfo{booktitle}{\emph{{SIGMOD} Conference}}. \bibinfo{publisher}{{ACM}},
  \bibinfo{pages}{2423--2436}.
\newblock


\bibitem[Armstrong(1974)]%
        {armstrong_dependency_1974}
\bibfield{author}{\bibinfo{person}{William~Ward Armstrong}.}
  \bibinfo{year}{1974}\natexlab{}.
\newblock \showarticletitle{Dependency Structures of Data Base Relationships}.
  In \bibinfo{booktitle}{\emph{{IFIP} Congress}}.
  \bibinfo{publisher}{North-Holland}, \bibinfo{pages}{580--583}.
\newblock


\bibitem[Batini and Scannapieco(2016)]%
        {batini_data_2016}
\bibfield{author}{\bibinfo{person}{Carlo Batini} {and} \bibinfo{person}{Monica
  Scannapieco}.} \bibinfo{year}{2016}\natexlab{}.
\newblock \bibinfo{booktitle}{\emph{Data and Information Quality - Dimensions,
  Principles and Techniques}}.
\newblock \bibinfo{publisher}{Springer}.
\newblock


\bibitem[Bonifati(2025)]%
        {bonifati_versatile_2025}
\bibfield{author}{\bibinfo{person}{Angela Bonifati}.}
  \bibinfo{year}{2025}\natexlab{}.
\newblock \showarticletitle{Versatile Property Graph Transformations}.
\newblock \bibinfo{journal}{\emph{Proc. {VLDB} Endow.}} \bibinfo{volume}{18},
  \bibinfo{number}{12} (\bibinfo{year}{2025}), \bibinfo{pages}{5516--5526}.
\newblock


\bibitem[Bonifati et~al\mbox{.}(2018)]%
        {bonifati_querying_2018}
\bibfield{author}{\bibinfo{person}{Angela Bonifati}, \bibinfo{person}{George
  H.~L. Fletcher}, \bibinfo{person}{Hannes Voigt}, {and}
  \bibinfo{person}{Nikolay Yakovets}.} \bibinfo{year}{2018}\natexlab{}.
\newblock \bibinfo{booktitle}{\emph{Querying Graphs}}.
\newblock \bibinfo{publisher}{Morgan {\&} Claypool Publishers}.
\newblock


\bibitem[Bronselaer(2021)]%
        {andreasen_data_2021}
\bibfield{author}{\bibinfo{person}{Antoon Bronselaer}.}
  \bibinfo{year}{2021}\natexlab{}.
\newblock \showarticletitle{Data Quality Management: An Overview of Methods and
  Challenges}. In \bibinfo{booktitle}{\emph{{FQAS}}}
  \emph{(\bibinfo{series}{Lecture Notes in Computer Science},
  Vol.~\bibinfo{volume}{12871})}. \bibinfo{publisher}{Springer},
  \bibinfo{pages}{127--141}.
\newblock


\bibitem[Codd(1970)]%
        {codd_relational_1970}
\bibfield{author}{\bibinfo{person}{E.~F. Codd}.}
  \bibinfo{year}{1970}\natexlab{}.
\newblock \showarticletitle{A Relational Model of Data for Large Shared Data
  Banks}.
\newblock \bibinfo{journal}{\emph{Commun. {ACM}}} \bibinfo{volume}{13},
  \bibinfo{number}{6} (\bibinfo{year}{1970}), \bibinfo{pages}{377--387}.
\newblock


\bibitem[Codd(1971)]%
        {codd_further_1971}
\bibfield{author}{\bibinfo{person}{E.~F. Codd}.}
  \bibinfo{year}{1971}\natexlab{}.
\newblock \showarticletitle{Further Normalization of the Data Base Relational
  Model}.
\newblock \bibinfo{journal}{\emph{Research Report / {RJ} / {IBM} / San Jose,
  California}}  \bibinfo{volume}{{RJ909}} (\bibinfo{year}{1971}).
\newblock


\bibitem[Codd(1974)]%
        {codd_recent_1974}
\bibfield{author}{\bibinfo{person}{E.~F. Codd}.}
  \bibinfo{year}{1974}\natexlab{}.
\newblock \showarticletitle{Recent Investigations in Relational Data Base
  Systems}. In \bibinfo{booktitle}{\emph{{IFIP} Congress}}.
  \bibinfo{publisher}{North-Holland}, \bibinfo{pages}{1017--1021}.
\newblock


\bibitem[Delobel and Casey(1973)]%
        {delobel_decomposition_1973}
\bibfield{author}{\bibinfo{person}{Claude Delobel} {and}
  \bibinfo{person}{Richard~G. Casey}.} \bibinfo{year}{1973}\natexlab{}.
\newblock \showarticletitle{Decomposition of a Data Base and the Theory of
  Boolean Switching Functions}.
\newblock \bibinfo{journal}{\emph{{IBM} J. Res. Dev.}} \bibinfo{volume}{17},
  \bibinfo{number}{5} (\bibinfo{year}{1973}), \bibinfo{pages}{374--386}.
\newblock


\bibitem[Ehrlinger and W{\"{o}}{\ss}(2018)]%
        {hacid_novel_2019}
\bibfield{author}{\bibinfo{person}{Lisa Ehrlinger} {and}
  \bibinfo{person}{Wolfram W{\"{o}}{\ss}}.} \bibinfo{year}{2018}\natexlab{}.
\newblock \showarticletitle{A Novel Data Quality Metric for Minimality}. In
  \bibinfo{booktitle}{\emph{QUAT@WISE}} \emph{(\bibinfo{series}{Lecture Notes
  in Computer Science}, Vol.~\bibinfo{volume}{11235})}.
  \bibinfo{publisher}{Springer}, \bibinfo{pages}{1--15}.
\newblock


\bibitem[Francis et~al\mbox{.}(2023a)]%
        {francis_gpc_2023}
\bibfield{author}{\bibinfo{person}{Nadime Francis},
  \bibinfo{person}{Am{\'{e}}lie Gheerbrant}, \bibinfo{person}{Paolo
  Guagliardo}, \bibinfo{person}{Leonid Libkin}, \bibinfo{person}{Victor
  Marsault}, \bibinfo{person}{Wim Martens}, \bibinfo{person}{Filip Murlak},
  \bibinfo{person}{Liat Peterfreund}, \bibinfo{person}{Alexandra Rogova}, {and}
  \bibinfo{person}{Domagoj Vrgoc}.} \bibinfo{year}{2023}\natexlab{a}.
\newblock \showarticletitle{{GPC:} {A} Pattern Calculus for Property Graphs}.
  In \bibinfo{booktitle}{\emph{{PODS}}}. \bibinfo{publisher}{{ACM}},
  \bibinfo{pages}{241--250}.
\newblock


\bibitem[Francis et~al\mbox{.}(2023b)]%
        {francis_researchers_2023}
\bibfield{author}{\bibinfo{person}{Nadime Francis},
  \bibinfo{person}{Am{\'{e}}lie Gheerbrant}, \bibinfo{person}{Paolo
  Guagliardo}, \bibinfo{person}{Leonid Libkin}, \bibinfo{person}{Victor
  Marsault}, \bibinfo{person}{Wim Martens}, \bibinfo{person}{Filip Murlak},
  \bibinfo{person}{Liat Peterfreund}, \bibinfo{person}{Alexandra Rogova}, {and}
  \bibinfo{person}{Domagoj Vrgoc}.} \bibinfo{year}{2023}\natexlab{b}.
\newblock \showarticletitle{A Researcher's Digest of {GQL} (Invited Talk)}. In
  \bibinfo{booktitle}{\emph{{ICDT}}} \emph{(\bibinfo{series}{LIPIcs},
  Vol.~\bibinfo{volume}{255})}. \bibinfo{publisher}{Schloss Dagstuhl -
  Leibniz-Zentrum f{\"{u}}r Informatik}, \bibinfo{pages}{1:1--1:22}.
\newblock


\bibitem[Hellings et~al\mbox{.}(2014)]%
        {hellings_implication_2014}
\bibfield{author}{\bibinfo{person}{Jelle Hellings}, \bibinfo{person}{Marc
  Gyssens}, \bibinfo{person}{Jan Paredaens}, {and} \bibinfo{person}{Yuqing
  Wu}.} \bibinfo{year}{2014}\natexlab{}.
\newblock \showarticletitle{Implication and Axiomatization of Functional
  Constraints on Patterns with an Application to the {RDF} Data Model}. In
  \bibinfo{booktitle}{\emph{FoIKS}} \emph{(\bibinfo{series}{Lecture Notes in
  Computer Science}, Vol.~\bibinfo{volume}{8367})}.
  \bibinfo{publisher}{Springer}, \bibinfo{pages}{250--269}.
\newblock


\bibitem[Heuer et~al\mbox{.}(2018)]%
        {saake_datenbanken_2018}
\bibfield{author}{\bibinfo{person}{Andreas Heuer}, \bibinfo{person}{Gunter
  Saake}, {and} \bibinfo{person}{Kai{-}Uwe Sattler}.}
  \bibinfo{year}{2018}\natexlab{}.
\newblock \bibinfo{booktitle}{\emph{Datenbanken - Konzepte und Sprachen, 6.
  Auflage}}.
\newblock \bibinfo{publisher}{{MITP}}.
\newblock


\bibitem[Hogan et~al\mbox{.}(2022)]%
        {hogan_knowledge_2022}
\bibfield{author}{\bibinfo{person}{Aidan Hogan}, \bibinfo{person}{Eva
  Blomqvist}, \bibinfo{person}{Michael Cochez}, \bibinfo{person}{Claudia
  d'Amato}, \bibinfo{person}{Gerard de Melo}, \bibinfo{person}{Claudio
  Gutierrez}, \bibinfo{person}{Sabrina Kirrane}, \bibinfo{person}{Jos{\'{e}}
  Emilio~Labra Gayo}, \bibinfo{person}{Roberto Navigli},
  \bibinfo{person}{Sebastian Neumaier}, \bibinfo{person}{Axel{-}Cyrille~Ngonga
  Ngomo}, \bibinfo{person}{Axel Polleres}, \bibinfo{person}{Sabbir~M. Rashid},
  \bibinfo{person}{Anisa Rula}, \bibinfo{person}{Lukas Schmelzeisen},
  \bibinfo{person}{Juan~F. Sequeda}, \bibinfo{person}{Steffen Staab}, {and}
  \bibinfo{person}{Antoine Zimmermann}.} \bibinfo{year}{2022}\natexlab{}.
\newblock \showarticletitle{Knowledge Graphs}.
\newblock \bibinfo{journal}{\emph{{ACM} Comput. Surv.}} \bibinfo{volume}{54},
  \bibinfo{number}{4} (\bibinfo{year}{2022}), \bibinfo{pages}{71:1--71:37}.
\newblock


\bibitem[ISO GQL(2024)]%
        {iso_gql_2024}
ISO GQL.
\newblock \bibinfo{title}{ISO/IEC 39075:2024 - Information technology —
  Database languages — GQL}.
\newblock
\urldef\tempurl%
\url{https://www.iso.org/standard/76120.html}
\showURL{%
\tempurl}


\bibitem[Lavrinovics et~al\mbox{.}(2025)]%
        {LavrinovicsBBH25}
\bibfield{author}{\bibinfo{person}{Ernests Lavrinovics}, \bibinfo{person}{Russa
  Biswas}, \bibinfo{person}{Johannes Bjerva}, {and} \bibinfo{person}{Katja
  Hose}.} \bibinfo{year}{2025}\natexlab{}.
\newblock \showarticletitle{{Knowledge Graphs, Large Language Models, and
  Hallucinations: An {NLP} Perspective}}.
\newblock \bibinfo{journal}{\emph{J. Web Semant.}}  \bibinfo{volume}{85}
  (\bibinfo{year}{2025}), \bibinfo{pages}{100844}.
\newblock


\bibitem[Leser and Naumann(2007)]%
        {leser_informationsintegration_2007}
\bibfield{author}{\bibinfo{person}{Ulf Leser} {and} \bibinfo{person}{Felix
  Naumann}.} \bibinfo{year}{2007}\natexlab{}.
\newblock \bibinfo{booktitle}{\emph{Informationsintegration - Architekturen und
  Methoden zur Integration verteilter und heterogener Datenquellen}}.
\newblock \bibinfo{publisher}{dpunkt.verlag}.
\newblock
\showISBNx{978-3-89864-400-6}


\bibitem[Lin et~al\mbox{.}(2024)]%
        {lin_fraudgt_2024}
\bibfield{author}{\bibinfo{person}{Junhong Lin}, \bibinfo{person}{Xiaojie Guo},
  \bibinfo{person}{Yada Zhu}, \bibinfo{person}{Samuel Mitchell},
  \bibinfo{person}{Erik Altman}, {and} \bibinfo{person}{Julian Shun}.}
  \bibinfo{year}{2024}\natexlab{}.
\newblock \showarticletitle{FraudGT: {A} Simple, Effective, and Efficient Graph
  Transformer for Financial Fraud Detection}. In
  \bibinfo{booktitle}{\emph{{ICAIF}}}. \bibinfo{publisher}{{ACM}},
  \bibinfo{pages}{292--300}.
\newblock


\bibitem[Manouvrier and Belhajjame(2026)]%
        {manouvrier_graph_2025}
\bibfield{author}{\bibinfo{person}{Maude Manouvrier} {and}
  \bibinfo{person}{Khalid Belhajjame}.} \bibinfo{year}{2026}\natexlab{}.
\newblock \showarticletitle{Graph functional dependencies: Analysis and
  translation to PG-schema}.
\newblock \bibinfo{journal}{\emph{Inf. Syst.}}  \bibinfo{volume}{136}
  (\bibinfo{year}{2026}), \bibinfo{pages}{102633}.
\newblock


\bibitem[Memgraph({[n.\,d.]})]%
        {memgraph-datamodel}
\bibfield{author}{\bibinfo{person}{Memgraph}.}
\newblock \bibinfo{title}{Graph data model}.
\newblock
\urldef\tempurl%
\url{https://memgraph.com/docs/data-modeling/graph-data-model}
\showURL{%
\tempurl}


\bibitem[Neo4j({[n.\,d.]})]%
        {neo4j-datamodel}
\bibfield{author}{\bibinfo{person}{Neo4j}.}
\newblock \bibinfo{title}{Graph database concepts}.
\newblock
\urldef\tempurl%
\url{https://neo4j.com/docs/getting-started/appendix/graphdb-concepts/}
\showURL{%
\tempurl}


\bibitem[Neo4j(2026)]%
        {northwind_database}
\bibfield{author}{\bibinfo{person}{Neo4j}.} \bibinfo{year}{2026}\natexlab{}.
\newblock \bibinfo{booktitle}{\emph{Northwind Graph Example}}.
\newblock
\urldef\tempurl%
\url{https://github.com/neo4j-graph-examples/northwind}
\showURL{%
\tempurl}


\bibitem[of~Investigative~Journalists({[n.\,d.]})]%
        {offshore_database}
\bibfield{author}{\bibinfo{person}{International~Consortium of
  Investigative~Journalists}.}
\newblock \bibinfo{title}{Offshore Leaks Database}.
\newblock
\urldef\tempurl%
\url{https://github.com/ICIJ/offshoreleaks-data-packages}
\showURL{%
\tempurl}


\bibitem[Paredaens et~al\mbox{.}(1989)]%
        {paredaens_structure_1989}
\bibfield{author}{\bibinfo{person}{Jan Paredaens}, \bibinfo{person}{Paul~De
  Bra}, \bibinfo{person}{Marc Gyssens}, {and} \bibinfo{person}{Dirk~Van
  Gucht}.} \bibinfo{year}{1989}\natexlab{}.
\newblock \bibinfo{booktitle}{\emph{The Structure of the Relational Database
  Model}}. \bibinfo{series}{{EATCS} Monographs on Theoretical Computer
  Science}, Vol.~\bibinfo{volume}{17}.
\newblock \bibinfo{publisher}{Springer}.
\newblock


\bibitem[Parr({[n.\,d.]})]%
        {antlr-website}
\bibfield{author}{\bibinfo{person}{Terence Parr}.}
\newblock \bibinfo{title}{ANTLR}.
\newblock
\urldef\tempurl%
\url{https://www.antlr.org}
\showURL{%
\tempurl}


\bibitem[Rabbani et~al\mbox{.}(2024)]%
        {RabbaniLBH24}
\bibfield{author}{\bibinfo{person}{Kashif Rabbani}, \bibinfo{person}{Matteo
  Lissandrini}, \bibinfo{person}{Angela Bonifati}, {and} \bibinfo{person}{Katja
  Hose}.} \bibinfo{year}{2024}\natexlab{}.
\newblock \showarticletitle{{Transforming {RDF} Graphs to Property Graphs using
  Standardized Schemas}}.
\newblock \bibinfo{journal}{\emph{Proc. {ACM} Manag. Data}}
  \bibinfo{volume}{2}, \bibinfo{number}{6} (\bibinfo{year}{2024}),
  \bibinfo{pages}{242:1--242:25}.
\newblock


\bibitem[Sakr et~al\mbox{.}(2021)]%
        {SakrBVIAAAABBDV21}
\bibfield{author}{\bibinfo{person}{Sherif Sakr}, \bibinfo{person}{Angela
  Bonifati}, \bibinfo{person}{Hannes Voigt}, \bibinfo{person}{Alexandru Iosup},
  \bibinfo{person}{Khaled Ammar}, \bibinfo{person}{Renzo Angles},
  \bibinfo{person}{Walid~G. Aref}, \bibinfo{person}{Marcelo Arenas},
  \bibinfo{person}{Maciej Besta}, \bibinfo{person}{Peter~A. Boncz},
  \bibinfo{person}{Khuzaima Daudjee}, \bibinfo{person}{Emanuele~Della Valle},
  \bibinfo{person}{Stefania Dumbrava}, \bibinfo{person}{Olaf Hartig},
  \bibinfo{person}{Bernhard Haslhofer}, \bibinfo{person}{Tim Hegeman},
  \bibinfo{person}{Jan Hidders}, \bibinfo{person}{Katja Hose},
  \bibinfo{person}{Adriana Iamnitchi}, \bibinfo{person}{Vasiliki Kalavri},
  \bibinfo{person}{Hugo Kapp}, \bibinfo{person}{Wim Martens},
  \bibinfo{person}{M.~Tamer {\"{O}}zsu}, \bibinfo{person}{Eric Peukert},
  \bibinfo{person}{Stefan Plantikow}, \bibinfo{person}{Mohamed Ragab},
  \bibinfo{person}{Matei Ripeanu}, \bibinfo{person}{Semih Salihoglu},
  \bibinfo{person}{Christian Schulz}, \bibinfo{person}{Petra Selmer},
  \bibinfo{person}{Juan~F. Sequeda}, \bibinfo{person}{Joshua Shinavier},
  \bibinfo{person}{G{\'{a}}bor Sz{\'{a}}rnyas}, \bibinfo{person}{Riccardo
  Tommasini}, \bibinfo{person}{Antonino Tumeo}, \bibinfo{person}{Alexandru
  Uta}, \bibinfo{person}{Ana~Lucia Varbanescu}, \bibinfo{person}{Hsiang{-}Yun
  Wu}, \bibinfo{person}{Nikolay Yakovets}, \bibinfo{person}{Da Yan}, {and}
  \bibinfo{person}{Eiko Yoneki}.} \bibinfo{year}{2021}\natexlab{}.
\newblock \showarticletitle{{The future is big graphs: a community view on
  graph processing systems}}.
\newblock \bibinfo{journal}{\emph{Commun. {ACM}}} \bibinfo{volume}{64},
  \bibinfo{number}{9} (\bibinfo{year}{2021}), \bibinfo{pages}{62--71}.
\newblock


\bibitem[Schrott(2026a)]%
        {schrott_2026_18479732}
\bibfield{author}{\bibinfo{person}{Johannes Schrott}, \bibinfo{person}{Maxime Jakubowski}, {and} \bibinfo{person}{Katja Hose}.}
  \bibinfo{year}{2026}\natexlab{a}.
\newblock \bibinfo{booktitle}{\emph{London Public Transport}}.
\newblock
\href{https://doi.org/10.5281/zenodo.18479731}{doi:\nolinkurl{10.5281/zenodo.18479731}}


\bibitem[Schrott(2026b)]%
        {schrott_2026_18478422}
\bibfield{author}{\bibinfo{person}{Johannes Schrott}, \bibinfo{person}{Maxime Jakubowski}, {and} \bibinfo{person}{Katja Hose}.}
  \bibinfo{year}{2026}\natexlab{b}.
\newblock \bibinfo{booktitle}{\emph{Train Services}}.
\newblock
\href{https://doi.org/10.5281/zenodo.18478421}{doi:\nolinkurl{10.5281/zenodo.18478421}}


\bibitem[Skavantzos and Link(2023)]%
        {skavantzos_normalizing_2023}
\bibfield{author}{\bibinfo{person}{Philipp Skavantzos} {and}
  \bibinfo{person}{Sebastian Link}.} \bibinfo{year}{2023}\natexlab{}.
\newblock \showarticletitle{Normalizing Property Graphs}.
\newblock \bibinfo{journal}{\emph{Proc. {VLDB} Endow.}} \bibinfo{volume}{16},
  \bibinfo{number}{11} (\bibinfo{year}{2023}), \bibinfo{pages}{3031--3043}.
\newblock


\bibitem[Skavantzos and Link(2025)]%
        {skavantzos_third_2025}
\bibfield{author}{\bibinfo{person}{Philipp Skavantzos} {and}
  \bibinfo{person}{Sebastian Link}.} \bibinfo{year}{2025}\natexlab{}.
\newblock \showarticletitle{Third and Boyce-Codd normal form for property
  graphs}.
\newblock \bibinfo{journal}{\emph{{VLDB} J.}} \bibinfo{volume}{34},
  \bibinfo{number}{2} (\bibinfo{year}{2025}), \bibinfo{pages}{23}.
\newblock


\bibitem[Skavantzos et~al\mbox{.}(2021)]%
        {skavantzos_uniqueness_2021}
\bibfield{author}{\bibinfo{person}{Philipp Skavantzos}, \bibinfo{person}{Kaiqi
  Zhao}, {and} \bibinfo{person}{Sebastian Link}.}
  \bibinfo{year}{2021}\natexlab{}.
\newblock \showarticletitle{Uniqueness Constraints on Property Graphs}. In
  \bibinfo{booktitle}{\emph{CAiSE}} \emph{(\bibinfo{series}{Lecture Notes in
  Computer Science}, Vol.~\bibinfo{volume}{12751})}.
  \bibinfo{publisher}{Springer}, \bibinfo{pages}{280--295}.
\newblock


\bibitem[Thalheim(2000)]%
        {thalheim_entity_2000}
\bibfield{author}{\bibinfo{person}{Bernhard Thalheim}.}
  \bibinfo{year}{2000}\natexlab{}.
\newblock \bibinfo{booktitle}{\emph{Entity-relationship modeling - foundations
  of database technology}}.
\newblock \bibinfo{publisher}{Springer}.
\newblock


\bibitem[Thalheim(2020)]%
        {thalheim_schema_2020}
\bibfield{author}{\bibinfo{person}{Bernhard Thalheim}.}
  \bibinfo{year}{2020}\natexlab{}.
\newblock \showarticletitle{Schema Optimisation Instead of (Local)
  Normalisation}. In \bibinfo{booktitle}{\emph{FoIKS}}
  \emph{(\bibinfo{series}{Lecture Notes in Computer Science},
  Vol.~\bibinfo{volume}{12012})}. \bibinfo{publisher}{Springer},
  \bibinfo{pages}{281--300}.
\newblock


\bibitem[Vrgoc et~al\mbox{.}(2023)]%
        {vrgoc_milleniumdb_2023}
\bibfield{author}{\bibinfo{person}{Domagoj Vrgoc}, \bibinfo{person}{Carlos
  Rojas}, \bibinfo{person}{Renzo Angles}, \bibinfo{person}{Marcelo Arenas},
  \bibinfo{person}{Diego Arroyuelo}, \bibinfo{person}{Carlos Buil{-}Aranda},
  \bibinfo{person}{Aidan Hogan}, \bibinfo{person}{Gonzalo Navarro},
  \bibinfo{person}{Cristian Riveros}, {and} \bibinfo{person}{Juan Romero}.}
  \bibinfo{year}{2023}\natexlab{}.
\newblock \showarticletitle{MillenniumDB: An Open-Source Graph Database
  System}.
\newblock \bibinfo{journal}{\emph{Data Intell.}} \bibinfo{volume}{5},
  \bibinfo{number}{3} (\bibinfo{year}{2023}), \bibinfo{pages}{560--610}.
\newblock


\end{thebibliography}

\begin{acronym}

\acro{1nf}[1NF]{first normal form}

\acro{2nf}[2NF]{second normal form}

\acro{3nf}[3NF]{third normal form}

\acro{4nf}[4NF]{fourth normal form}

\acro{5nf}[5NF]{fifth normal form}

\acro{ai}[AI]{artificial intelligence}

\acro{api}[API]{application programming interface}

\acro{bcnf}[BCNF]{Boyce-Codd Normal Form}

\acro{cfd}[CFD]{Conditional Functional Dependency}
\acroplural{cfd}[CFDs]{Conditional Functional Dependencies}

\acro{db}[DB]{database}

\acro{dbms}[DBMS]{database management system}

\acro{dq}[DQ]{data quality}

\acro{er}[ER]{entity-relationship}

\acro{fc}[FC]{Functional Constraint}

\acro{fd}[FD]{functional dependency}
\acroplural{fd}[FDs]{functional dependencies}

\acro{fk}[FK]{foreign key}

\acro{ged}[GED]{Graph Entity Dependency}
\acroplural{ged}[GEDs]{Graph Entity Dependencies}

\acro{gfdFan}[GFD]{Graph Functional Dependency}
\acroplural{gfdFan}[GFDs]{Graph Functional Dependencies}

\acro{gfd}[gFD]{Graph-tailored Functional Dependency}
\acroplural{gfd}[gFDs]{Graph-tailored Functional Dependencies}

\acro{ggd}[GGD]{Graph Generating Dependency}
\acroplural{ggd}[GGDs]{Graph Generating Dependencies}

\acro{gnfd}[GO-FD]{Graph Object Functional Dependency}
\acroplural{gnfd}[GO-FDs]{Graph Object Functional Dependencies}

\acro{gnnf}[GN-NF]{Graph-Native Normal Form}

\acro{gns}[GNS]{GN-Schema}

\acro{gpar}[GPAR]{Graph-Pattern Association Rule}

\acro{gpc}[GPC]{graph pattern calculus}

\acro{gql}[GQL]{Graph Query Language}

\acro{guc}[gUC]{Graph-tailored Uniqueness Constraint}

\acro{id}[ID]{identifier}

\acro{json}[JSON]{JavaScript Object Notation}

\acro{kg}[KG]{knowledge graph}

\acro{llm}[LLM]{large language model}

\acro{lpg}[LPG]{labeled property graph}

\acro{ml}[ML]{machine learning}

\acro{mvd}[MVD]{multi-valued dependency}
\acroplural{mvd}[MVDs]{multi-valued dependencies}

\acro{rag}[RAG]{retrieval augmented generation}

\acro{rdf}[RDF]{Resource Description Framework}

\acro{rest}[REST]{Representational State Transfer}

\acro{sbcnf}[SBCNF]{Succinct Boyce-Codd Normal Form} %

\acro{sql}[SQL]{Structured Query Language}

\acro{tgfd}[TGFD]{Temporal Graph Functional Dependency}
\acroplural{tgfd}[TGFDs]{Temporal Graph Functional Dependencies}

\acro{xml}[XML]{Extensible Markup Language}

\acro{yaml}[YAML]{YAML Ain't Markup Language}

\end{acronym}

\edbtarxiv{}{
\appendix
\onecolumn

\section{Appendix}
\label{sec:appendix}

As supplementary material to the experimental evaluation of the graph-native approach to \ac{lpg} normalization (cf. \Cref{sec:evaluation}), we provide in this appendix (1) the minimal covers of the dependencies of each considered scenario, as well as (2) the queries used for the query performance experiment.

\subsection{Minimal Covers}
The minimal covers of the dependencies of each scenario are listed below and are given using the syntax defined for the Python implementation.

\begin{itemize}
    \item \textbf{$S_\text{Lon}$:}
    \begin{itemize}
      \item $(s:\{\mathsf{Station}\}:\{\mathsf{zone},\mathsf{zoneOriginal}\})::s.\mathsf{zoneOriginal}\determ s.\mathsf{zone}$
      \item $(s:\{\mathsf{Station}\}:\{\mathsf{latitude},\mathsf{longitude}\})::s.\mathsf{latitude},s.\mathsf{longitude}\determ s$
      \item $(s:\{\mathsf{Station}\}:\emptyset)\xrightarrow{c:\{\mathsf{CONNECTED\_THROUGH}\}:\{\mathsf{line},\mathsf{color},\mathsf{type}\}}()::c.\mathsf{line}\determ c.\mathsf{color},c.\mathsf{type}$
    \end{itemize}
    \item \textbf{$S_\text{Nw}$:}
    \begin{itemize}
      \item $(o:\{\mathsf{Order}\}:\{\mathsf{orderID}\})::o.\mathsf{orderID}\determ o$
      \item $(o:\{\mathsf{Order}\}:\{\mathsf{orderID},\mathsf{orderDate},\mathsf{customerID},\mathsf{shipCity},\mathsf{shipPostalCode},\mathsf{shipCountry},\mathsf{shipAddress},\mathsf{shipRegion}\})::$\\$o.\mathsf{customerID}\determ o.\mathsf{shipCity},o.\mathsf{shipPostalCode},o.\mathsf{shipCountry},o.\mathsf{shipAddress},o.\mathsf{shipRegion}$
    \end{itemize}
    \item \textbf{$S_\text{No-1}$:}
    \begin{itemize}
      \item $(e:\{\mathsf{Event}\}:\{\mathsf{company},\mathsf{name},\mathsf{time},\mathsf{venue}\})::e.\mathsf{company},e.\mathsf{time}\determ e$
      \item $(e:\{\mathsf{Event}\}:\{\mathsf{company},\mathsf{name},\mathsf{time},\mathsf{venue}\})::e.\mathsf{name}\determ e.\mathsf{company}$
      \item $(e:\{\mathsf{Event}\}:\{\mathsf{company},\mathsf{name},\mathsf{time},\mathsf{venue}\})::e.\mathsf{name},e.\mathsf{time}\determ e.\mathsf{venue}$
      \item $(e:\{\mathsf{Event}\}:\{\mathsf{company},\mathsf{name},\mathsf{time},\mathsf{venue}\})::e.\mathsf{time},e.\mathsf{venue}\determ e.\mathsf{name}$
    \end{itemize}
    \item \textbf{$S_\text{No-2}$:}
    \begin{itemize}
      \item $(e:\{\mathsf{Event}\}:\{\mathsf{company},\mathsf{name},\mathsf{time},\mathsf{venue}\})::e.\mathsf{company},e.\mathsf{time} \determ e.\mathsf{venue}$
      \item $(e:\{\mathsf{Event}\}:\{\mathsf{company},\mathsf{name},\mathsf{time},\mathsf{venue}\})::e.\mathsf{name} \determ e.\mathsf{company}$
      \item $(e:\{\mathsf{Event}\}:\{\mathsf{company},\mathsf{name},\mathsf{time},\mathsf{venue}\})::e.\mathsf{name},e.\mathsf{time}\determ e.\mathsf{venue}$
      \item $(e:\{\mathsf{Event}\}:\{\mathsf{company},\mathsf{name},\mathsf{time},\mathsf{venue}\})::e.\mathsf{time},e.\mathsf{venue}\determ e.\mathsf{name}$
    \end{itemize}
    \item \textbf{$S_\text{No-3}$:}
    \begin{itemize}
      \item $(e:\{\mathsf{Event},\mathsf{Confirmed}\}:\{\mathsf{company},\mathsf{name},\mathsf{time},\mathsf{venue}\})::e.\mathsf{company},e.\mathsf{time}\determ e$
      \item $(e:\{\mathsf{Event},\mathsf{Confirmed}\}:\{\mathsf{company},\mathsf{name},\mathsf{time},\mathsf{venue}\})::e.\mathsf{name} \determ e.\mathsf{company}$
      \item $(e:\{\mathsf{Event},\mathsf{Confirmed}\}:\{\mathsf{company},\mathsf{name},\mathsf{time},\mathsf{venue}\})::e.\mathsf{name},e.\mathsf{time} \determ e.\mathsf{venue}$
      \item $(e:\{\mathsf{Event},\mathsf{Confirmed}\}:\{\mathsf{company},\mathsf{name},\mathsf{time},\mathsf{venue}\})::e.\mathsf{time},e.\mathsf{venue}\determ e.\mathsf{name}$
      \item $(e:\{\mathsf{Event},\mathsf{Confirmed}\}:\{\mathsf{company},\mathsf{name},\mathsf{time},\mathsf{venue}\})::e.\mathsf{name},e.\mathsf{company} \determ e.\mathsf{time}$
    \end{itemize}
    \item \textbf{$S_\text{Off-1}$:}
    \begin{itemize}
        \item $(e:\{\mathsf{Entity}\}:\{\mathsf{jurisdiction\_description},\mathsf{countries},\mathsf{service\_provider},\mathsf{country\_codes}\}):: e.\mathsf{countries} \determ e.\mathsf{country\_codes}$
        \item $(e:\{\mathsf{Entity}\}:\{\mathsf{jurisdiction\_description},\mathsf{countries},\mathsf{service\_provider},\mathsf{country\_codes}\}):: e.\mathsf{country\_codes}\determ e.\mathsf{countries}$
    \end{itemize}
    \item \textbf{$S_\text{Off-2}$:}
    \begin{itemize}
      \item $(e:\{\mathsf{Entity}\}:\{\mathsf{jurisdiction\_description},\mathsf{valid\_until},\mathsf{countries},\mathsf{sourceID},\mathsf{country\_codes}\})::$\\$e.\mathsf{countries},e.\mathsf{jurisdiction\_description}\determ e.\mathsf{country\_codes}$
      \item $(e:\{\mathsf{Entity}\}:\{\mathsf{jurisdiction\_description},\mathsf{valid\_until},\mathsf{countries},\mathsf{sourceID},\mathsf{country\_codes}\})::$\\$ e.\mathsf{countries},e.\mathsf{sourceID}\determ e.\mathsf{country\_codes}$
      \item $(e:\{\mathsf{Entity}\}:\{\mathsf{jurisdiction\_description},\mathsf{valid\_until},\mathsf{countries},\mathsf{sourceID},\mathsf{country\_codes}\})::$\\$e.\mathsf{countries},e.\mathsf{valid\_until}\determ e.\mathsf{country\_codes}$
      \item $(e:\{\mathsf{Entity}\}:\{\mathsf{jurisdiction\_description},\mathsf{valid\_until},\mathsf{countries},\mathsf{sourceID},\mathsf{country\_codes}\})::$\\$ e.\mathsf{country\_codes},e.\mathsf{sourceID}\determ e.\mathsf{countries}$
      \item $(e:\{\mathsf{Entity}\}:\{\mathsf{jurisdiction\_description},\mathsf{valid\_until},\mathsf{countries},\mathsf{sourceID},\mathsf{country\_codes}\})::$\\$ e.\mathsf{country\_codes},e.\mathsf{valid\_until}\determ e.\mathsf{countries}$

    \end{itemize}
    \item \textbf{$S_\text{Off-3}$:}
    \begin{itemize}
      \item $(e:\{\mathsf{Entity}\}:\{\mathsf{countries}    ,\mathsf{service\_provider},\mathsf{country\_codes},\mathsf{jurisdiction\_description},\mathsf{sourceID},\mathsf{valid\_until}\})::$\\$e.\mathsf{countries} \determ e.\mathsf{country\_codes}$
      \item $(e:\{\mathsf{Entity}\}:\{\mathsf{countries}    ,\mathsf{service\_provider},\mathsf{country\_codes},\mathsf{jurisdiction\_description},\mathsf{sourceID},\mathsf{valid\_until}\})::$\\$e.\mathsf{country\_codes} \determ e.\mathsf{countries}$
      \item $(e:\{\mathsf{Entity}\}:\{\mathsf{countries}    ,\mathsf{service\_provider},\mathsf{country\_codes},\mathsf{jurisdiction\_description},\mathsf{sourceID},\mathsf{valid\_until}\})::$\\$e.\mathsf{service\_provider} \determ e.\mathsf{sourceID}$
      \item $(e:\{\mathsf{Entity}\}:\{\mathsf{countries}    ,\mathsf{service\_provider},\mathsf{country\_codes},\mathsf{jurisdiction\_description},\mathsf{sourceID},\mathsf{valid\_until}\})::$\\$e.\mathsf{sourceID} \determ e.\mathsf{valid\_until}$
      \item $(e:\{\mathsf{Entity}\}:\{\mathsf{countries}    ,\mathsf{service\_provider},\mathsf{country\_codes},\mathsf{jurisdiction\_description},\mathsf{sourceID},\mathsf{valid\_until}\})::$\\$e.\mathsf{valid\_until}\determ \mathsf{e.sourceID}$
    \end{itemize}
    \item \textbf{$S_\text{Ts-1}$:}
    \begin{itemize}    
      \item $()\xrightarrow{t:\{\mathsf{STOPS\_AT}\}:\{\mathsf{code}\}}(s:\{\mathsf{Station}\}:\{\mathsf{name}\})::s.\mathsf{name}\determ t.\mathsf{code}$
      \item $(ts:\{\mathsf{TrainService}\}:\{\mathsf{date},\mathsf{number},\mathsf{type}\})::ts.\mathsf{number},ts.\mathsf{date}\determ ts.\mathsf{type}$
      \item $(ts:\{\mathsf{TrainService}\}:\{\mathsf{date},\mathsf{number},\mathsf{operator}\})\xrightarrow{:\{\mathsf{STOPS\_AT}\}:\emptyset}():: ts.\mathsf{number},ts.\mathsf{date}\determ ts.\mathsf{operator}$
      \item $(ts:\{\mathsf{TrainService}\}:\{\mathsf{serviceid}\})\xrightarrow{:\{\mathsf{STOPS\_AT}\}:\emptyset}()::ts.\mathsf{serviceid} \determ ts$
      \item $()\xrightarrow{t:\{\mathsf{STOPS\_AT}\}:\{\mathsf{stopid}\}}()::t.\mathsf{stopid} \determ t$
      \item $()\xrightarrow{t:\{\mathsf{STOPS\_AT}\}:\{\mathsf{departure},\mathsf{stopid}\}}(s)::t.\mathsf{stopid}\determ s$
    \end{itemize}
    \item \textbf{$S_\text{Ts-2}$:}
    \begin{itemize}
      \item $()\xrightarrow{t:\{\mathsf{STOPS\_AT}\}:\{\mathsf{code}\}}(s:\{\mathsf{Station}\}:\{\mathsf{name}\})::s.\mathsf{name}\determ t.\mathsf{code}$
      \item $()\xrightarrow{t:\{\mathsf{STOPS\_AT}\}:\{\mathsf{code}\}}(s:\{\mathsf{Station}\}:\{\mathsf{name}\})::t.\mathsf{code}\determ s.\mathsf{name}$
      \item $(ts:\{\mathsf{TrainService}\}:\{\mathsf{date},\mathsf{number},\mathsf{type}\})::ts.\mathsf{number},ts.\mathsf{date}\determ ts.\mathsf{type}$
      \item $(ts:\{\mathsf{TrainService}\}:\{\mathsf{date},\mathsf{number},\mathsf{operator}\})\xrightarrow{:\{\mathsf{STOPS\_AT}\}:\emptyset}():: ts.\mathsf{number},ts.\mathsf{date}\determ ts.\mathsf{operator}$
      \item $(ts:\{\mathsf{TrainService}\}:\{\mathsf{serviceid}\})\xrightarrow{:\{\mathsf{STOPS\_AT}\}:\emptyset}()::ts.\mathsf{serviceid} \determ ts$
      \item $()\xrightarrow{t:\{\mathsf{STOPS\_AT}\}:\{\mathsf{stopid}\}}()::t.\mathsf{stopid} \determ t$
      \item $()\xrightarrow{t:\{\mathsf{STOPS\_AT}\}:\{\mathsf{departure},\mathsf{stopid}\}}(s)::t.\mathsf{stopid}\determ s$
    \end{itemize}
    \item \textbf{$S_\text{Uni-1}$:}
    \begin{itemize}
      \item $(c:\{\mathsf{Course}\}:\emptyset)\xleftarrow{t:\{\mathsf{TEACHES}\}:\{\mathsf{usingBook}\}}()::c=>t.\mathsf{usingBook}$
    \end{itemize}
    \item \textbf{$S_\text{Uni-2}$:}
    \begin{itemize}
      \item $(c:\{\mathsf{Course}\}:\{\mathsf{title}\})\xleftarrow{t:\{\mathsf{TEACHES}\}:\{\mathsf{usingBook}\}}()::c.\mathsf{title}\determ t.\mathsf{usingBook}$
    \end{itemize}
\end{itemize}

\subsection{Queries}

This sections contain the queries used in the experiment on query performance. For each query, we show the original query before normalization and the rewritten queries for the normalized graph.

\begin{itemize}
    \item \textbf{$S_\text{Lon}$}
    \begin{enumerate}
        \item \emph{Original query:} direct
        \begin{itemize}
            \item \lstinline|PROFILE MATCH ()-[c:CONNECTED_THROUGH]->() RETURN DISTINCT c.line, c.color|
        \end{itemize}
        \emph{After normalization:}
        \begin{itemize}
            \item Optimal rewrite:\\ \lstinline{PROFILE MATCH (c:Ccolorclinectype) RETURN DISTINCT c.Cline, c.Ccolor} 
            \item Naive rewrite:\\ \lstinline{PROFILE MATCH ()-[]-(d:CONNECTED_THROUGH)-[]-(), (d)-[]-(c:Ccolorclinectype) RETURN DISTINCT c.Cline, c.Ccolor}
        \end{itemize}
        
        \item \emph{Original query:} indirect
        \begin{itemize}
            \item \lstinline{PROFILE MATCH (s:Station)-[c:CONNECTED_THROUGH]-() WHERE s.name = "Piccadilly Circus" RETURN DISTINCT c.line}
        \end{itemize}
        \emph{After normalization:}
        \begin{itemize}
            \item \lstinline{PROFILE MATCH (s:Station)-[]-(:CONNECTED_THROUGH)-[:CCOLORCLINECTYPE]-(c:Ccolorclinectype) WHERE s.name = "Piccadilly Circus" RETURN DISTINCT c.Cline}
        \end{itemize}
        
        \item \emph{Original query:} direct
        \begin{itemize}
            \item \lstinline{PROFILE MATCH (s:Station) WHERE s.zoneOriginal IS NOT NULL RETURN DISTINCT s.zone, s.zoneOriginal}
        \end{itemize}
        \emph{After normalization:}
        \begin{itemize}
            \item Optimal rewrite:\\ \lstinline{PROFILE MATCH (s:SzoneszoneOriginal) WHERE s.SzoneOriginal IS NOT NULL RETURN DISTINCT s.Szone, s.SzoneOriginal}
            \item Naive rewrite:\\ \lstinline{PROFILE MATCH (s:Station)-[]-(s:SzoneszoneOriginal) WHERE s.SzoneOriginal IS NOT NULL RETURN DISTINCT s.Szone, s.SzoneOriginal}
        \end{itemize}
        
        \item \emph{Original query:} indirect
        \begin{itemize}
            \item \lstinline{PROFILE MATCH (s:Station)-[c:CONNECTED_THROUGH]-(:Station) WHERE s.zone = "6" RETURN DISTINCT c.line}
        \end{itemize}
        \pagebreak
        \emph{After normalization:}
        \begin{itemize}
            \item \lstinline{PROFILE MATCH (s:Station)-[]-(c:CONNECTED_THROUGH)-[]-(:Station), (s)-[]-(z:SzoneszoneOriginal), (c)-[]-(d:Ccolorclinectype) WHERE z.Szone = "6" RETURN DISTINCT d.Cline}
        \end{itemize}
        
        \item \emph{Original query:} indirect
        \begin{itemize}
            \item \lstinline{PROFILE MATCH (s:Station)-[c:CONNECTED_THROUGH]->(t:Station) WHERE s.name = "Piccadilly Circus" AND c.line = "Bakerloo" RETURN DISTINCT s.name, t.name}
        \end{itemize}
        \emph{After normalization:}
        \begin{itemize}
            \item \lstinline{PROFILE MATCH (s:Station)-[]->(c:CONNECTED_THROUGH)-[]->(t:Station), (c)-[]-(d:Ccolorclinectype) WHERE s.name = "Piccadilly Circus" and d.Cline ="Bakerloo" RETURN t.name}
        \end{itemize}
        
        \item \emph{Original query:} direct
        \begin{itemize}
            \item \lstinline{PROFILE MATCH ()-[c:CONNECTED_THROUGH]->() WHERE c.line = "Bakerloo" SET c.color = "Rainbow"}
        \end{itemize}
        \emph{After normalization:}
        \begin{itemize}
            \item Optimal rewrite:\\  \lstinline{PROFILE MATCH (c:Ccolorclinectype) WHERE c.Cline = "Bakerloo" SET c.Ccolor = "Rainbow"}
            \item Naive rewrite:\\  \lstinline{PROFILE MATCH (:CONNECTED_THROUGH)-->(c:Ccolorclinectype) WHERE c.Cline = "Bakerloo" SET c.Ccolor = "Rainbow"}
            \item Naive rewrite:\\ \lstinline{PROFILE MATCH (:CONNECTED_THROUGH)-->(c:Ccolorclinectype) WHERE c.Cline = "Bakerloo" WITH DISTINCT c SET c.Ccolor = "Rainbow"}
        \end{itemize}
    \end{enumerate}
    \item \textbf{$S_\text{Nw}$}
    \begin{enumerate}
        \item \emph{Original query:} indirect
        \begin{itemize}
            \item \lstinline{PROFILE MATCH (o:Order) RETURN DISTINCT o.orderID, o.customerID}
        \end{itemize}
        \emph{After normalization:}
        \begin{itemize}
            \item \lstinline{PROFILE MATCH (o:Order)-[]-(c:OcustomerIDoshipAddressoshipCityoshipCountryoshipPostalCodeoshipRegion) RETURN DISTINCT o.orderID, c.OcustomerID}
        \end{itemize}
        
        \item \emph{Original query:} direct
        \begin{itemize}
            \item \lstinline{PROFILE MATCH (o:Order) WHERE o.customerID IS NOT NULL AND o.shipCity IS NOT NULL AND o.shipPostalCode IS NOT NULL AND o.shipCountry IS NOT NULL AND o.shipAddress IS NOT NULL AND o.shipRegion IS NOT NULL RETURN DISTINCT o.customerID}
        \end{itemize}
        \emph{After normalization:}
        \begin{itemize}
            \item Optimal rewrite:\\
            \lstinline{PROFILE MATCH (o:OcustomerIDoshipAddressoshipCityoshipCountryoshipPostalCodeoshipRegion) WHERE o.OcustomerID IS NOT NULL AND o.OshipCity IS NOT NULL AND o.OshipPostalCode IS NOT NULL AND o.OshipCountry IS NOT NULL AND o.OshipAddress IS NOT NULL AND o.OshipRegion IS NOT NULL RETURN DISTINCT o.OcustomerID}
            \item Naive rewrite:\\ \lstinline{PROFILE MATCH (:Order)-->(o:OcustomerIDoshipAddressoshipCityoshipCountryoshipPostalCodeoshipRegion) WHERE o.OcustomerID IS NOT NULL AND o.OshipCity IS NOT NULL AND o.OshipPostalCode IS NOT NULL AND o.OshipCountry IS NOT NULL AND o.OshipAddress IS NOT NULL AND o.OshipRegion IS NOT NULL RETURN DISTINCT o.OcustomerID}
        \end{itemize}
        
        \item \emph{Original query:} direct
        \begin{itemize}
            \item \lstinline{PROFILE MATCH (o:Order) WHERE o.customerID = "FOLKO" SET o.shipCountry = "Rainbow"}
        \end{itemize}
        \emph{After normalization:}
        \begin{itemize}
            \item Naive rewrite:\\ \lstinline{PROFILE MATCH (o:Order)-[]-(c:OcustomerIDoshipAddressoshipCityoshipCountryoshipPostalCodeoshipRegion) WHERE c.OcustomerID = "FOLKO" SET c.OshipCountry = "Rainbow"}
            \item Naive rewrite:\\ \lstinline{PROFILE MATCH (o:Order)-[]-(c:OcustomerIDoshipAddressoshipCityoshipCountryoshipPostalCodeoshipRegion) WHERE c.OcustomerID = "FOLKO" WITH DISTINCT c SET c.OshipCountry = "Rainbow"}
            \item Optimal rewrite:\\ \lstinline{PROFILE MATCH (c:OcustomerIDoshipAddressoshipCityoshipCountryoshipPostalCodeoshipRegion) WHERE c.OcustomerID = "FOLKO" SET c.OshipCountry = "Rainbow"}
        \end{itemize}
    \end{enumerate}

    \pagebreak
    \item \textbf{$S_\text{Off-1}$}
    \begin{enumerate}
        \item \emph{Original query:} direct
        \begin{itemize}
            \item \lstinline{PROFILE MATCH (e:Entity) WHERE e.countries IS NOT NULL AND e.jurisdiction_description IS NOT NULL AND e.country_codes IS NOT NULL RETURN DISTINCT e.countries, e.country_codes}
        \end{itemize}
        \emph{After normalization:}
        \begin{itemize}
            \item Optimal rewrite:\\ \lstinline{PROFILE MATCH (c:EcountriesecountryCodes) WHERE c.Ecountries IS NOT NULL AND c.EcountryCodes IS NOT NULL RETURN DISTINCT c.Ecountries, c.EcountryCodes}
            \item Naive rewrite:\\ \lstinline{PROFILE MATCH (:Entity)-[]->(c:EcountriesecountryCodes) WHERE c.Ecountries IS NOT NULL AND c.EcountryCodes IS NOT NULL RETURN DISTINCT c.Ecountries, c.EcountryCodes}
        \end{itemize}
        
        \item \emph{Original query:} indirect
        \begin{itemize}
            \item \lstinline{PROFILE MATCH (e:Entity) WHERE e.countries IS NOT NULL AND e.service_provider IS NOT NULL AND e.country_codes IS NOT NULL RETURN DISTINCT e.countries,e.service_provider, e.country_codes}
        \end{itemize}
        \emph{After normalization:}
        \begin{itemize}
            \item \lstinline{PROFILE MATCH (e:Entity)-[]-(c:EcountriesecountryCodes) WHERE c.Ecountries IS NOT NULL AND c.EcountryCodes IS NOT NULL RETURN DISTINCT c.Ecountries, c.EcountryCodes, e.service_provider}
        \end{itemize}
        
        \item \emph{Original query:} direct
        \begin{itemize}
            \item \lstinline{PROFILE MATCH (e:Entity) WHERE e.countries IS NOT NULL AND e.jurisdiction_description IS NOT NULL AND e.country_codes IS NOT NULL AND e.countries = "Singapore" SET e.country_codes = "Sp"}
        \end{itemize}
        \emph{After normalization:}
        \begin{itemize}
            \item Optimal rewrite:\\ \lstinline{PROFILE MATCH (c:EcountriesecountryCodes) WHERE c.Ecountries IS NOT NULL AND c.EcountryCodes IS NOT NULL AND c.Ecountries = "Singapore" SET c.EcountryCodes = "Sp"}
            \item Naive rewrite:\\ \lstinline{PROFILE MATCH (:Entity)-[]->(c:EcountriesecountryCodes) WHERE c.Ecountries IS NOT NULL AND c.EcountryCodes IS NOT NULL AND c.Ecountries = "Singapore" SET c.EcountryCodes = "Sp"}
            \item Naive rewrite:\\ \lstinline{PROFILE MATCH (:Entity)-[]->(c:EcountriesecountryCodes) WHERE c.Ecountries IS NOT NULL AND c.EcountryCodes IS NOT NULL AND c.Ecountries = "Singapore" WITH DISTINCT c SET c.EcountryCodes = "Sp"}
        \end{itemize}
    \end{enumerate}

    \item \textbf{$S_\text{Ts-2}$}
    \begin{enumerate}
        \item \emph{Original query:} direct
        \begin{itemize}
            \item \lstinline{PROFILE MATCH ()-[t:STOPS_AT]->(s:Station) WHERE s.name IS NOT NULL AND t.code IS NOT NULL RETURN DISTINCT t.code, s.name}
        \end{itemize}
        \emph{After normalization:}
        \begin{itemize}
            \item Optimal rewrite:\\ \lstinline{PROFILE MATCH (n:Snametcode) RETURN DISTINCT n.Sname, n.Tcode}
            \item Naive rewrite:\\ \lstinline{PROFILE MATCH (s:Station)-->(n:Snametcode), (t:STOPS_AT)-->(n), (t)-->(s) RETURN DISTINCT n.Sname, n.Tcode}
        \end{itemize}
        
        \item \emph{Original query:} indirect
        \begin{itemize}
            \item \lstinline{PROFILE MATCH (ts:TrainService)-[t:STOPS_AT]->(s:Station) WHERE ts.type CONTAINS "Nightjet" RETURN DISTINCT ts.number, ts.type, t.code}
        \end{itemize}
        \emph{After normalization:}
        \begin{itemize}
            \item \lstinline{PROFILE MATCH (ts:TrainService)-->(t:STOPS_AT)-->(s:Station), (ts)-->(d:Tsdatetsnumbertstype), (s)-->(n:Snametcode), (t)-->(n) WHERE d.Tstype CONTAINS "Nightjet" RETURN DISTINCT d.Tsnumber, d.Tstype, n.Tcode}
        \end{itemize}
        
        \item \emph{Original query:} indirect
        \begin{itemize}
            \item \lstinline{PROFILE MATCH (ts:TrainService) WHERE ts.type CONTAINS "Nightjet" RETURN DISTINCT ts.number, ts.completely_cancelled}
        \end{itemize}
        \emph{After normalization:}
        \begin{itemize}
            \item \lstinline{PROFILE MATCH (ts:TrainService)-->(d:Tsdatetsnumbertstype) WHERE d.Tstype CONTAINS "Nightjet" RETURN DISTINCT d.Tsnumber, ts.completely_cancelled}
        \end{itemize}
        \pagebreak
        \item \emph{Original query:} indirect
        \begin{itemize}
            \item \lstinline{PROFILE MATCH (ts:TrainService) WHERE ts.type CONTAINS "Sprinter" RETURN DISTINCT ts.number, ts.operator}
        \end{itemize}
        \emph{After normalization:}
        \begin{itemize}
            \item \lstinline{PROFILE MATCH (ts:TrainService)-->(d:Tsdatetsnumbertstype), (ts)-->(e:Tsdatetsnumbertsoperator) WHERE d.Tstype CONTAINS "Sprinter" RETURN DISTINCT d.Tsnumber, e.Tsoperator}
        \end{itemize}
        
        \item \emph{Original query:} direct
        \begin{itemize}
            \item \lstinline{PROFILE MATCH ()-[t:STOPS_AT]->(s:Station) WHERE s.name IS NOT NULL AND t.code IS NOT NULL AND s.name = "Innsbruck Hbf" SET t.code = "Rainbow"}
        \end{itemize}
        \emph{After normalization:}
        \begin{itemize}
            \item Optimal rewrite:\\\lstinline{PROFILE MATCH (n:Snametcode) WHERE n.Sname = "Innsbruck Hbf" SET n.Tcode = "Rainbow"}
            \item Naive rewrite:\\ \lstinline{PROFILE MATCH (s:Station)-->(n:Snametcode), (t:STOPS_AT)-->(n), (t)-->(s) WHERE n.Sname = "Innsbruck Hbf" SET n.Tcode = "Rainbow"}
            \item Naive rewrite:\\\lstinline{PROFILE MATCH (s:Station)-->(n:Snametcode), (t:STOPS_AT)-->(n), (t)-->(s) WHERE n.Sname = "Innsbruck Hbf" WITH DISTINCT n  SET n.Tcode = "Rainbow"}
        \end{itemize}
    \end{enumerate}
\end{itemize}
}
\end{document}